\documentclass[journal]{IEEEtran}

\usepackage{amsmath}
\usepackage{amsfonts}
\usepackage{amssymb}
\usepackage{dsfont}
\usepackage{bm}
\usepackage{amsthm}
\usepackage{newlfont}
\usepackage{float}
\usepackage{hyperref}
\usepackage{algorithm}
\usepackage{algorithmic}
\usepackage{enumerate}
\usepackage{chngcntr}
\usepackage{mathtools}
\usepackage{breqn}
\usepackage{stmaryrd}
\usepackage{cite}
\usepackage{balance}
\hypersetup{
    colorlinks=true, 
    linkcolor= blue, 
    citecolor=blue 
}
\DeclareMathOperator*{\argmax}{arg\,max}

\begin{document}

\newtheorem{thm}{Theorem} 
\newtheorem{lem}{Lemma}
\newtheorem{prop}{Proposition}
\newtheorem{cor}{Corollary}
\newtheorem{defn}{Definition}
\newtheorem{rem}{Remark}
\newtheorem{ex}{Example}

\renewcommand{\qedsymbol}{ \begin{tiny}$\blacksquare$ \end{tiny} }

\renewcommand{\leq}{\leqslant}
\renewcommand{\geq}{\geqslant}

\title {Polar Coding for Secret-Key Generation} 
\author{
    \IEEEauthorblockN{R\'{e}mi A. Chou, Matthieu R. Bloch, Emmanuel Abbe}
\thanks{R. A. Chou and M. R. Bloch are with the School~of~Electrical~and~Computer~Engineering,~Georgia~Institute~of~Technology, Atlanta,~GA~30332 and with GT-CNRS UMI 2958, Metz, France. E. Abbe is with the School of Engineering and Applied Sciences, Princeton University, Princeton, NJ 08544.}
    \thanks{This work was supported in part by the NSF under Award CCF 1320298 and by ANR with grant 13-BS03-0008.}
 \thanks{E-mail: remi.chou@gatech.edu; matthieu.bloch@ece.gatech.edu; eabbe@princeton.edu. Parts of the results were presented at the 2013 IEEE Information Theory Workshop~\cite{Chou13b}.}
 \thanks{Copyright (c) 2015 IEEE. Personal use of this material is permitted.  However, permission to use this material for any other purposes must be obtained from the IEEE by sending a request to pubs-permissions@ieee.org.}
}
\maketitle

\begin{abstract}
Practical implementations of secret-key generation are often based on sequential strategies, which handle reliability and secrecy in two successive steps, called reconciliation and privacy amplification. In this paper, we propose an alternative approach based on polar codes that jointly deals with reliability and secrecy. Specifically, we propose secret-key capacity-achieving polar coding schemes for the following models: (i) the degraded binary memoryless source (DBMS) model with rate-unlimited public communication, (ii)  the DBMS model with one-way rate-limited public communication, (iii) the $1$-to-$m$ broadcast model and (iv)~the Markov tree model with uniform marginals. For models~(i) and (ii) our coding schemes remain valid for non-degraded sources, although they may not achieve the secret-key capacity. For models (i), (ii) and (iii), our schemes rely on pre-shared secret seed of negligible rate; however, we provide special cases of these models for which no seed is required. Finally, we show an application of our results to secrecy and privacy for biometric systems.
We thus provide the first examples of low-complexity secret-key capacity-achieving schemes that are able to handle vector quantization for model (ii), or multiterminal communication for models (iii) and (iv).
\end{abstract}

\section{Introduction}
Unlike classical cryptography, physical-layer security relies on information-theoretic metrics rather than complexity theory and the supposed hardness of solving certain mathematical problems. In particular, information-theoretic secret-key generation protocols~\cite{Maurer93,Ahlswede93} put no limits on the computational power of the adversary. In such protocols, legitimate users and eavesdropper observe the realizations of correlated random variables. The legitimate users, who can publicly communicate, then aim at extracting a common secret-key from their observations. The maximum number of secret-key bits per observed realization of the random variables is called the secret-key capacity~\cite{Maurer93,Ahlswede93}.

Bounds for the secret-key capacity have been derived for a large variety of models~\cite{Maurer93,Ahlswede93,Csiszar00,Csiszar04,Ye05,Csiszar08,Csiszar10,
Watanabe10b,Nitinawarat12,Chou12b,Csiszar13,Chou13g}. 
Unfortunately, most rely on typicality arguments and do not provide direct insight into the design of practical secret-key capacity-achieving schemes. There are, however, a few exceptions. For instance, there exist constructive schemes for some multiterminal scenarios~\cite{Nitinawarat10,Nitinawarat10b} based on explicit algorithms for tree packing. In addition, sequential methods can be constructed that successively handle reliability and secrecy by means of reconciliation and privacy amplification. While sequential methods lead to low-complexity schemes for unlimited public communication~\cite{Maurer00,Bloch11}, their application to rate-limited public communication~\cite{Nitinawarat12,Chou12b} requires vector quantization for which, to the best of our knowledge, no low-complexity schemes are known.

This paper presents low-complexity secret-key capacity-achieving schemes based on polar codes~\cite{Arikan09} for some classes of source models. 
Note that polar codes have already been successfully used for secrecy in the context of symmetric degraded wire-tap channel model~\cite{Hof10,Koyluoglu10,Andersson10,Mahdavifar11,Sasoglu13}, or more recently, arbitrary broadcast channel with confidential messages~\cite{Gulcu14,Chou14c}, and for the Slepian-Wolf coding problem~\cite{Arikan10,Korada10a,Abbe11a,Sasoglu11}, which is particularly relevant to secret-key generation. Note also that in \cite{Abbe11n}, the journal version of \cite{Abbe11a}, a first application of polar coding to a basic secret key generation setting was proposed. Unlike sequential methods, which successively handle reliability and secrecy, our schemes jointly deal with reliability and secrecy (see Definition~\ref{def} for more details).  Both the sequential reliability-secrecy approach, and the direct approach with polar codes have their advantages. On the one hand, sequential methods offer flexibility in design by separating reliability and secrecy and, unlike polar coding schemes, are known to remain optimal for two-way rate-limited communication and continuous non-degraded sources~\cite{Chou12b}. On the other hand, polar coding schemes may be easier to design and operate at lesser complexity in some scenarios. They also appear to be convenient to deal with vector quantization when the public communication is rate-limited. 

Our main contribution is to develop polar coding schemes that achieve the secret-key capacity for the following models. 
\begin{itemize}
\item The degraded binary memoryless source (DBMS) model with rate-unlimited public communication;
\item The DBMS model with one-way rate-limited public communication;
\item The $1$-to-$m$ broadcast model;
\item The Markov tree model with uniform marginals.
\end{itemize}
For the first two models, the proposed polar coding schemes may also be used to generate secret keys for non-degraded sources, although they may not achieve the secret-key capacity. For the first three models, we assume that the legitimate users initialize their communication with a shared secret seed,\footnote{If one assumes an authenticated public channel~\cite{Maurer93,Ahlswede93} a shared small secret seed in the order of the logarithm of the length of the messages is also required for authentication~\cite{Wegman81}.} whose length is negligible compared to the number of source samples used to generate a key. As shown in Sections~\ref{sec_model1}-\ref{sec_model3}, there also exist special cases of the source statistics for which no seed is required.

Note that~\cite{Sutter13}, obtained independently from \cite{Chou13b}, develops an alternative polar coding solution for the BMS model with rate-unlimited public communication. The major difference between their approach and ours is that their construction is sequential, i.e., it  \emph{successively} deals with reliability and secrecy by means of reconciliation and privacy amplification,  whereas our approach \emph{jointly} deals with reliability and secrecy. The construction in~\cite[Th. 7]{Sutter13} has the advantage of not requiring a seed. On the other hand, our protocol only requires one ``polarization layer,'' whose construction is efficient, whereas the sequential approach of ~\cite{Sutter13} requires an inner and an outer layer, the latter having no known efficient code construction as discussed in~\cite[Section III.C]{Sutter13}.

The remainder of the paper is organized as follows. Section~\ref{SecStatemetn} formally introduces some notation and the general multi-terminal secret-key generation problem, which encompasses all the models specialized in subsequent sections. Section~\ref{sec_model1}, describes a secret-key capacity-achieving scheme with polar codes for the DBMS model with unlimited communication rate. Section~\ref{sec_model2} provides a secret-key capacity-achieving scheme with polar codes for the DBMS model with one-way rate-limited public communication. Section~\ref{sec_model3} develops a secret-key capacity-achieving scheme with polar codes for the $1$-to-$m$ broadcast model. Section~\ref{sec_model4} studies a Markov tree model with uniform marginals and provides a secret-key capacity-achieving scheme with polar codes. Finally, Section~\ref{Sec_bio}, shows how to apply the results to the related problem of privacy and secrecy for key generation in some biometric systems. 
\section{Multiterminal secret-key generation}
\label{SecStatemetn} 
We start by introducing some notation used throughout the paper. We define the integer interval $\llbracket a,b \rrbracket$, as the set of integers between $\lfloor a \rfloor$ and $\lceil b \rceil$.  We denote the Bernoulli distribution with parameter $p \in [0,1]$ by $\mathcal{B}(p)$. For $n \in \mathbb{N}$ and $N \triangleq 2^n$, we let $G_N \triangleq  \left[ \begin{smallmatrix}
       1 & 0            \\[0.3em]
       1 & 1 
     \end{smallmatrix} \right]^{\otimes n} $ be the source polarization transform defined in~\cite{Arikan10}. We note the components of a vector, $X^{1:N}$ with superscripts, i.e., $X^{1:N} \triangleq (X^1 , X^2, \ldots, X^{N})$. Finally, we denote the variational distance and the divergence between two distributions by $\mathbb{V}(\cdot, \cdot)$ and $\mathbb{D}(\cdot || \cdot)$, respectively.

We now recall the general model for multiterminal secret-key generation~\cite{Csiszar04}. Let $m \geq 2$ be the number of terminals that wish to generate a common secret-key. Set $\mathcal{M} \triangleq \llbracket 1, m \rrbracket$, and let $\mathcal{Z}$ and $\mathcal{X}_i$, for $i \in \mathcal{M}$ be arbitrary finite alphabets. Define $\mathcal{X}_{\mathcal{M}}$ as the Cartesian product of $\mathcal{X}_1,\mathcal{X}_2, \ldots, \mathcal{X}_m$. Consider a discrete memoryless multiple source $\left(\mathcal{X}_{\mathcal{M}} \mathcal{Z} , p_{X_{\mathcal{M}}Z} \right)$, where $X_{\mathcal{M}} \triangleq (X_1,X_2, \ldots, X_m)$ and the Cartesian product $\mathcal{X}_{\mathcal{M}} \times \mathcal{Z}$ is abbreviated as $\mathcal{X}_{\mathcal{M}} \mathcal{Z}$. For $i \in\mathcal{M}$, Terminal $i$ observes the component $X_i$ of $\left(\mathcal{X}_{\mathcal{M}} \mathcal{Z} , p_{X_{\mathcal{M}}Z} \right)$, whereas an eavesdropper observes the component $Z$. The source is assumed to be outside the control of all parties, but its statistics are known to all parties. Communication is allowed between terminals over an authenticated noiseless public channel with communication rate $R_p \in \mathbb{R}^+ \cup \{ + \infty\}$. A secret-key generation strategy is then formally defined as follows.
\begin{defn} \label{definition_modelg}
Let $R_p \in \mathbb{R}^+ \cup \{ + \infty\}$. Let  $\mathcal{K}$ be a key alphabet of size $2^{NR}$. The protocol defined by the following steps is called a $(2^{NR},N,R_p)$ secret-key generation strategy with public communication, and is denoted by $\mathcal{S}_N$.
\begin{enumerate}
\item Terminal $i$, $i \in \mathcal{M}$, observes $X_i^{1:N}$.
\item The $m$ terminals communicate, possibly interactively, over the public channel. All the public inter-terminal communications are collectively denoted by $\mathbf{F}$ and satisfy $H(\mathbf{F}) \leq N R_p$. 
\item Terminal $i$, $i \in \mathcal{M}$, computes $K_i(X_i^{1:N},\mathbf{F}) \in \mathcal{K}$.
\end{enumerate}
\end{defn}
Let $K$ be a random variable taking values in $\mathcal{K}$. The performance of a secret-key generation strategy $\mathcal{S}_N$ that allows the terminals in $\mathcal{M}$ to agree on the key $K$ is measured in terms of 
\begin{itemize}
\item the average probability of error between the keys $\textbf{P}_e(\mathcal{S}_N) \triangleq \mathbb{P} [\exists i\in \mathcal{M}: K \neq K_i],$
\item the information leakage to the eavesdropper $\textbf{L}(\mathcal{S}_N) \triangleq {I} (K;Z^{1:N} \mathbf{F}),$
\item the uniformity of the key $\textbf{U}(\mathcal{S}_N) \triangleq \log \lceil 2^{NR} \rceil - {H}(K)$.
\end{itemize}
\begin{defn} \label{def}
A secret-key rate $R$ is achievable if there exists a sequence of $(2^{NR},N,R_p)$ secret-key generation strategies $\left\{ \mathcal{S}_N \right\}_{N \geq 1}$ such that
\begin{align*}
 \displaystyle\lim_{N \to \infty } \textbf{\textup{P}}_e(\mathcal{S}_N)  = & 0, \text{ (reliability) } \\
 \displaystyle\lim_{N \to \infty } \textbf{\textup{L}}(\mathcal{S}_N) = & 0, \text{ (strong secrecy)}  \\
 \displaystyle\lim_{N \to \infty } \textbf{\textup{U}}(\mathcal{S}_N)= & 0. \text { (uniformity)}
\end{align*}
 Moreover, the supremum of achievable rates is called the secret-key capacity and is denoted $C_{\text{WSK}}(R_p)$. In the special case where Eve has no access to the component $Z$ of the source, the secret-key capacity is denoted $C_{\text{SK}}(R_p)$. One also says that perfect secrecy is achieved if $\textbf{\textup{L}}(\mathcal{S}_N) =0$.
\end{defn}
In this paper, we develop low-complexity secret-key capacity-achieving schemes based on polar codes for special cases of the general model presented in Definition~\ref{definition_modelg}.
In the following, the blocklength, $N$, used by the legitimate users is a power of $2$. Moreover, we say that the legitimate users share a secret seed, if they share a secret sequence of $d_N\in \mathbb{N}$ uniformly distributed bits, and we define the seed rate as $d_N /N$. To avoid modifying the secret-key capacity with the introduction of a seed, we only consider schemes with vanishing seed rate.
\section{Model 1: Secret Key Generation with Rate-Unlimited Public Communication} \label{sec_model1}
The precise model and known results are described in Section~\ref{Secstatmod1}. Our proposed polar coding scheme is given in Section~\ref{Sec_scheme1} and analyzed in Section~\ref{SecproofTh1}.

\subsection{Secret-key generation model} \label{Secstatmod1}
\begin{figure}[b]
\centering
  \includegraphics[width=8.5cm]{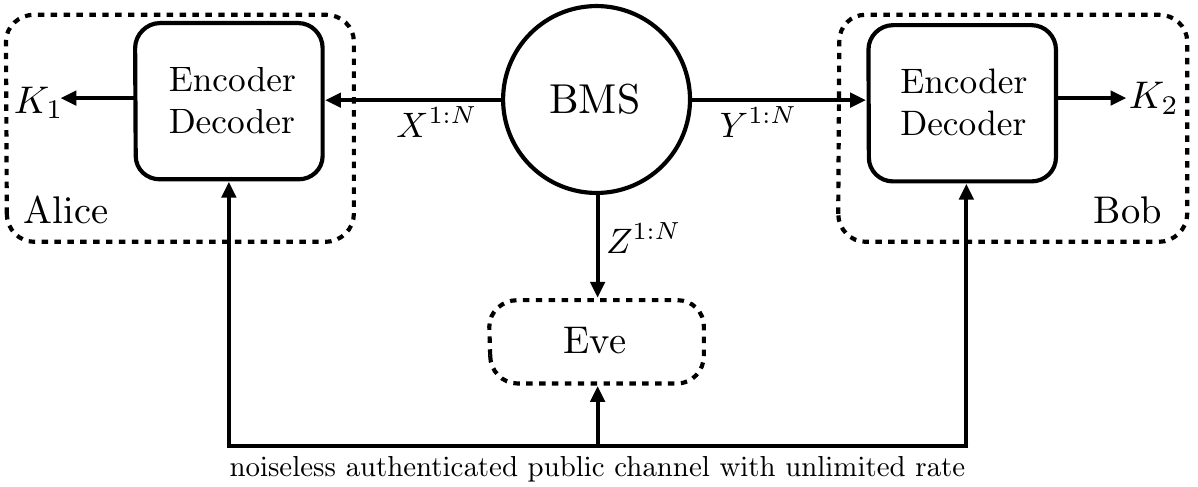}
  \caption{Model 1: Secret-key generation for the BMS model with rate-unlimited public communication}
  \label{figmodelWSK}
  \vspace*{-0.8em}
\end{figure}

As illustrated in Fig.~\ref{figmodelWSK}, Model 1 consists of $m=2$ legitimate terminals. We use $\mathcal{X}$ instead of $\mathcal{X}_1$ and $\mathcal{Y}$ instead of $\mathcal{X}_2$ for convenience. We assume that $\mathcal{X} = \{0,1\}$ and that the public channel has an unlimited communication rate $R_p = + \infty$. We call this setup the BMS model with rate-unlimited public communication. The following results are known for this model.

\begin{thm}[\!\! \cite{Maurer93,Ahlswede93}] \label{thcsiszar}
Consider a BMS $(\mathcal{X}\mathcal{Y}\mathcal{Z},p_{XYZ})$. If $X \to Y \to Z$, then the secret-key capacity $C_{\textup{WSK}} (+ \infty)$ is  
\begin{align*}
 &C_{\textup{WSK}} (+ \infty) = I(X;Y) - I(X;Z).
\end{align*}
Moreover, the secret-key capacity can be achieved by one-way communication.
\end{thm}

When the eavesdropper has no access to the source component $Z$, one obtains the following expression for the secret-key capacity.

\begin{cor}[\!\! \cite{Maurer93,Ahlswede93}] \label{corcsiszar}
Consider a BMS $(\mathcal{X}\mathcal{Y},p_{XY})$. The secret-key capacity $C_{\textup{SK}}(+ \infty)$ is  
\begin{align*}
 &C_{\textup{SK}} (+\infty) = I(X;Y).
\end{align*}
Moreover, the secret-key capacity can be achieved using only one-way communication.
\end{cor}

Such a model is motivated by the sources of randomness that can be generated from wireless communication channel gains~\cite{Ye2010,Pierrot13}. In such settings, the wireless channel gains $c_{A \rightarrow B}$ characterizing the channel from Alice to Bob, $c_{B \rightarrow A}$ characterizing the channel from Bob to Alice, and  the pair $(c_{A \rightarrow E}$, $c_{B \rightarrow E})$, characterizing the channels to Eve, may be used as the variables $X$, $Y$, and $Z$, respectively, of Model 1.

\subsection{Polar coding scheme} \label{Sec_scheme1}

In the following, we assume that $I(X;Y) - I(X;Z) >0$  but we do not assume that $X \to Y \to Z$ forms a Markov chain; we discuss at the end of the section how the coding scheme simplifies when $X \to Y \to Z$ holds.

Let $n \in \mathbb{N}$ and $N \triangleq 2^n$. Set $U^{1:N} \triangleq X^{1:N} G_N$. For any set  $\mathcal{A} \triangleq \{ i_j \}_{j=1}^{|\mathcal{A}|}$ of indices in $\llbracket 1,N \rrbracket$, we define $U^{1:N}[\mathcal{A}] \triangleq \left( U^{i_1}, U^{i_2}, \ldots,U^{i_{|\mathcal{A}|}}\right)$. 
For $\delta_N \triangleq 2^{-N^{\beta}}$, where $\beta \in ]0,1/2[$, define the following sets
\begin{align*}
\mathcal{V}_{X|Z}  \triangleq &  \left\{ i\in \llbracket 1,N\rrbracket : {H} \left(U^i|U^{1:i-1} Z^{1:N} \right) \geq 1- \delta_N \right\} , \\
\mathcal{H}_{X|Y} \triangleq & \left\{ i\in \llbracket 1,N\rrbracket : {H}\left(U^i|U^{1:i-1} Y^{1:N}\right) \geq \delta_N \right\}.
\end{align*}
The exact encoding and decoding algorithms are given in Algorithm~\ref{alg:encoding_1} and Algorithm~\ref{alg:decoding_1}, respectively, and we provide here a high-level discussion of their operation. The set $\mathcal{H}_{X|Y}$ is the set of indices containing ``high-entropy bits'' such that $U^{1:N}[\mathcal{H}_{X|Y}]$ allows Bob to near losslessly reconstruct $U^{1:N}$ from $Y^{1:N}$~\cite{Arikan10}. In our coding scheme, Alice therefore publicly transmits $U^{1:N}[\mathcal{H}_{X|Y}]$ to allow Bob to reconstruct $U^{1:N}$. By construction, the set $\mathcal{V}_{X|Z}$ is the set of indices containing ``very-high entropy bits'' such that $U^{1:N}[\mathcal{V}_{X|Z}]$ is almost uniform and independent of the eavesdropper's observations $Z^{1:N}$. Consequently, the secret-key should be chosen as a subvector of $U^{1:N}[\mathcal{V}_{X|Z}]$; specifically, since $U^{1:N}[\mathcal{H}_{X|Y}]$ is publicly transmitted, it is natural to use $U^{1:N}[\mathcal{V}_{X|Z} \backslash \mathcal{H}_{X|Y} ]$ as the secret key. Unfortunately, $\mathcal{H}_{X|Y} \not\subset \mathcal{V}_{X|Z}$ in general, so that the public communication of $U^{1:N}[\mathcal{H}_{X|Y}]$ leaks some information about $U^{1:N}[\mathcal{V}_{X|Z} \backslash \mathcal{H}_{X|Y} ]$. To circumvent this issue, our protocol uses a secret seed to protect the transmission of the bits in positions $\mathcal{H}_{X|Y} \backslash  \mathcal{V}_{X|Z}$ with a one-time-pad. In addition, our scheme operates over $k$ blocks of size $N$ to handle non-degraded sources and to make the seed rate negligible. In every Block $i \in \llbracket 1,k\rrbracket$ Alice generates a secret key $K_i$ together with a seed $\widetilde{K}_i$ used in the next block. Overall, Alice obtains a vector of secret keys $K_{1:k} \triangleq [ K_{1}, K_2, \ldots, K_k ]$ while Bob obtains a vector of estimates $\widehat{K}_{1:k} \triangleq [ \widehat{K}_{1}, \widehat{K}_2, \ldots, \widehat{K}_k ]$.

\begin{algorithm}
  \caption{Alice's encoding algorithm for Model 1}
  \label{alg:encoding_1}
  \begin{algorithmic}[1]    
    \REQUIRE $\widetilde{K}_0$, a secret key of size $|\mathcal{H}_{X|Y} \backslash  \mathcal{V}_{X|Z} |$ shared by Alice and Bob beforehand; for every Block $i \in \llbracket 1,k\rrbracket$, the observations $X_i^{1:N}$ from the source; $\mathcal{A}_{XYZ}$ a fixed subset of $\mathcal{V}_{X|Z} \backslash \mathcal{H}_{X|Y}$ with size $|\mathcal{H}_{X|Y} \backslash  \mathcal{V}_{X|Z}|$.
    \FOR{Block $i=1$ to $k$}
    \STATE $U_i^{1:N} \leftarrow X_i^{1:N} G_N$
    \STATE $\widetilde{K}_i \leftarrow U_i^{1:N} [\mathcal{A}_{XYZ}]$\COMMENT{Fraction of the key used as a seed for the next block}
    \STATE $K_i \leftarrow U_i^{1:N} [(\mathcal{V}_{X|Z} \backslash \mathcal{H}_{X|Y}) \backslash \mathcal{A}_{XYZ}]$
    \STATE $F_i \leftarrow U_i^{1:N} [\mathcal{V}_{X|Z} \cap \mathcal{H}_{X|Y}]$
    \STATE $F_i' \leftarrow U_i^{1:N} [\mathcal{H}_{X|Y} \backslash  \mathcal{V}_{X|Z} ]$
    \STATE Transmit $M_i\leftarrow [F_i,F'_i\oplus \widetilde{K}_{i-1}]$ publicly to Bob
    \ENDFOR
    \RETURN $K_{1:k} \leftarrow [ K_{1}, K_2, \ldots, K_k ]$
  \end{algorithmic}
\end{algorithm}

\begin{algorithm}
  \caption{Bob's decoding algorithm for Model 1}
  \label{alg:decoding_1}
  \begin{algorithmic}[1]    
    \REQUIRE The secret key $\widetilde{K}_0$ and the set $\mathcal{A}_{XYZ}$ defined in Algorithm \ref{alg:encoding_1}; for every Block $i \in \llbracket 1,k\rrbracket$, the observations $Y_i^{1:N}$ from the source and the message $M_i$ transmitted by Alice.
    \FOR{Block $i=1$ to $k$}
    \STATE Form $U_i^{1:N}[\mathcal{H}_{X|Y}]$ from $M_i$ and $\widetilde{K}_{i-1}$
    \STATE Create an estimate $\widehat{U}_i^{1:N}$ of ${U}_i^{1:N}$ with the successive cancellation decoder of \cite{Arikan10}
    \STATE $\widehat{K}_i\leftarrow \widehat{U}_i^{1:N} [(\mathcal{V}_{X|Z} \backslash \mathcal{H}_{X|Y}) \backslash \mathcal{A}_{XYZ}]$
    \STATE $\widetilde{K}_i\leftarrow \widehat{U}_i^{1:N} [\mathcal{A}_{XYZ}]$
    \ENDFOR
    \RETURN $\widehat{K}_{1:k} \leftarrow [ \widehat{K}_{1}, \widehat{K}_2, \ldots, \widehat{K}_k ]$
  \end{algorithmic}
\end{algorithm}

\begin{rem} \label{rm:efficient_use_key}
For convenience, Algorithm~\ref{alg:encoding_1} does not distinguish the last block from the others; however, there is no need to create a seed in Block $k$, so that one may actually use $U_k^{1:N} [\mathcal{V}_{X|Z} \backslash \mathcal{H}_{X|Y}]$ as the key $K_k$ and slightly increase the key rate. For a large number of blocks $k$, this distinction has negligible impact on the achievable rates.
\end{rem}

\begin{rem} \label{remexpln}
The need for a seed is not an artifact of our proof, but a fundamental requirement of our single polarization approach to generate secret keys and public messages. In fact, a memoryless source cannot be near losslessly compressed at a rate close to the entropy and simultaneously ensure that the encoded messages are nearly uniformly distributed in variational distance~\cite[Section V]{Hayashi08}. In the context of secret-key generation with polar codes, this translates into the condition $\mathcal{H}_{X|Y} \not\subset \mathcal{V}_{X|Z}$ and in the impossibility of simultaneously ensuring strong secrecy and reliability. Our solution follows ideas from~\cite{Dodis05,Chou13}, showing that the impossibility may be circumvented if the encoder and the decoder share a small seed beforehand; without seed, only weak secrecy would be ensured.
\end{rem}

As shown in Section~\ref{SecproofTh1}, a careful analysis of the algorithms leads to the following result.

\begin{thm} \label{Th1}
Consider a BMS $(\mathcal{X}\mathcal{Y}\mathcal{Z},p_{XYZ})$. Assume that Alice and Bob share a secret seed. The secret-key rate $I(X;Y) - I(X;Z)$ is achieved by the polar coding scheme of Algorithm~\ref{alg:encoding_1} and Algorithm~\ref{alg:decoding_1}, which involves a chaining of $k$ blocks of size $N$, and whose computational complexity is $O(kN \log N)$. Moreover, the seed rate can be chosen in $o\left( 2^{-N^{\alpha} }\right)$, $\alpha < 1/2$.
\end{thm}

\begin{proof}[Proof]
See Section~\ref{SecproofTh1}.
\end{proof}

\begin{cor} \label{Cor_1}
When $X \to Y \to Z$, the secret-key capacity of Theorem~\ref{thcsiszar} is achieved by the polar coding scheme of Algorithm~\ref{alg:encoding_1} and Algorithm~\ref{alg:decoding_1}. Moreover, one does not need to encode over several blocks, i.e., one can choose $k=1$, and the seed rate is $o(N)$. However, encoding over several blocks for this case allows one to reduce the seed rate from $o(N)$ to $o(2^{-N^{\alpha}})$, $\alpha < 1/2$.
\end{cor}

\begin{proof}
See Appendix~\ref{App_cor1}.
\end{proof}

Note that, in the special case of a symmetric degraded BMS,\footnote{That is, when $X$, $Y$, and $Z$ are connected by symmetric channels.} Corollary~\ref{Cor_1} may be indirectly obtained from wiretap codes and~\cite{Sasoglu13}, following the approach of~\cite{Maurer93},~\cite[Section 4.2.1]{Bloch11}. However, this indirect proof might not translate into practical implementations because it requires much more public channel communication.

Although the seed rate in Theorem \ref{Th1} or Corollary~\ref{Cor_1} may be made arbitrarily small, it is valuable to identify examples for which no seed is required. We provide two such examples in Proposition~\ref{ex1}, which corresponds to the privacy amplification setting of~\cite{Bennett95}, and in Proposition~\ref{ex3}, which corresponds to a case when the source has uniform marginals and the eavesdropper has no access to correlated observations of the source.
\begin{prop} \label{ex1}
Consider a BMS $ (\mathcal{X}\mathcal{Y}\mathcal{Z},p_{XYZ})$. Assume that Alice and Bob have the same observations, i.e., $X=Y$; then the secret-key capacity $C_{\text{WSK}} = H(X|Z)$ is achievable with a polar coding scheme, whose computational complexity is $O(N \log N)$.
\end{prop}

\begin{proof}
See Appendix \ref{App_Ex1}.
\end{proof}

\begin{prop} \label{ex3}
Consider a BMS $ (\mathcal{X}\mathcal{Y},p_{XY})$ with $X \sim \mathcal{B}(1/2)$. The secret-key capacity $C_{\text{SK}}(+\infty)$ given in Corollary~\ref{corcsiszar} is achievable with perfect secrecy with a polar coding scheme, whose computational complexity is $O(N \log N)$.
\end{prop}

\begin{proof}
See Appendix \ref{App_Ex2}.
\end{proof}
Note that the model studied in Proposition~\ref{ex3} includes \cite[Model 1]{Ye12} as a special case, and does not require the construction of a standard array, whose size grows exponentially with the blocklength.

\subsection{Analysis of polar coding scheme: proof of Theorem~\ref{Th1}} \label{SecproofTh1}
A functional dependence graph of the block encoding scheme of Section~\ref{Sec_scheme1} is depicted in Figure~\ref{figFGD} to help the reader identify the dependencies among the variables introduced by the block-coding scheme.
\begin{figure}
\centering
  \includegraphics[width=8.5cm]{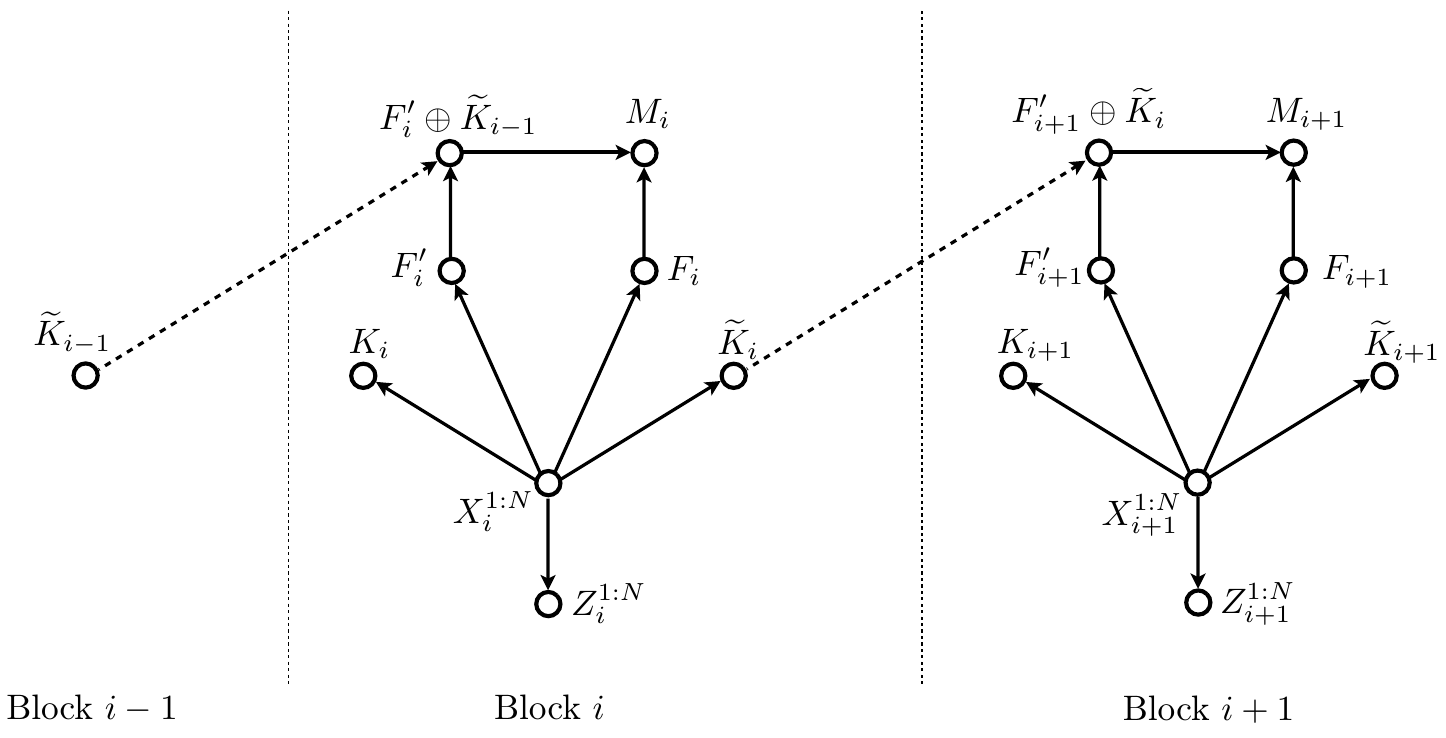}
  \caption{Functional dependence graph of the proposed block encoding scheme described in Algorithm \ref{alg:encoding_1}}
  \label{figFGD}
\end{figure}
\subsubsection{Preliminary result}
We first state a lemma that will be useful for the scheme analysis.
\begin{lem} \label{lemnewcard}
The set $ \mathcal{V}_{X|Z}$ is such that $$\lim_{N \rightarrow + \infty} |\mathcal{V}_{X|Z}| / N = {H}(X|Z).$$
\end{lem}
\begin{proof}
See Appendix~\ref{App_bat}.
\end{proof}
Note that $\lim_{N \rightarrow + \infty} |\mathcal{H}_{X|Y}| / N = {H}(X|Y)$ follows from~\cite{Arikan10}, but Lemma~\ref{lemnewcard} requires a slightly different proof based on Lemma~\ref{lem_counterpart} in the appendix.

\subsubsection{Existence of $\mathcal{A}_{XYZ}$} \label{existence_AXYZ}
Observe that $$|\mathcal{V}_{X|Z} \backslash \mathcal{H}_{X|Y}| - |\mathcal{H}_{X|Y} \backslash  \mathcal{V}_{X|Z}| = | \mathcal{V}_{X|Z}| -|\mathcal{H}_{X|Y}|.$$ Hence, by Lemma~\ref{lemnewcard} and \cite{Arikan10}, we have
\begin{multline*}
\lim_{N \to \infty} (|\mathcal{V}_{X|Z} \backslash \mathcal{H}_{X|Y}| - |\mathcal{H}_{X|Y} \backslash  \mathcal{V}_{X|Z}|) /N\\= H(X|Z) - H(X|Y).
\end{multline*}
Since $I(X;Y) - I(X;Z) >0$ by assumption, we conclude that $|\mathcal{V}_{X|Z} \backslash \mathcal{H}_{X|Y}| - |\mathcal{H}_{X|Y} \backslash  \mathcal{V}_{X|Z}| >0$ for $N$ large enough and $\mathcal{A}_{XYZ}$ exists.

\subsubsection{Asymptotic key rate}
For $i \in \llbracket 1, k \rrbracket$, we note $|K_i|$ the length of the vector $K_i$. The length of the overall key generated is 
\begin{align*}
|K_{1:k}|
& = \sum_{i=1}^{k} |K_i| \\
& = k |(\mathcal{V}_{X|Z} \backslash \mathcal{H}_{X|Y})\backslash \mathcal{A}_{XYZ}|\\
& = k (|\mathcal{V}_{X|Z} \backslash \mathcal{H}_{X|Y}| - |\mathcal{H}_{X|Y} \backslash  \mathcal{V}_{X|Z}|) \\
& = k (|\mathcal{V}_{X|Z}| - |\mathcal{H}_{X|Y} |) .
\end{align*}
Hence, by Lemma \ref{lemnewcard} and \cite{Arikan10}, the asymptotic key rate is
$$
\lim_{N \to \infty} \frac{|K_{1:k}|}{kN} \geq I(X;Y) - I(X;Z).
$$

\subsubsection{Reliability} \label{Sec_errProb1}
Let $i \in \llbracket 2,k\rrbracket$. Note that $F_i'$ is correctly received only when Bob possesses a correct estimate of the seed $\widetilde{K}_{i-1}$, i.e., when $U_{i-1}^{1:N}$ is correctly reconstructed. We note $\widehat{F}_i'$ the estimate of $F_i'$ formed by Bob from $(\widehat{U}_{i-1}^{1:N}, M_i)$ and define the event $\mathcal{E}_{F'_i} \triangleq \{ F_i' \neq  \widehat{F}_i'\}$. Then, 
\begin{align*} 
&\mathbb{P} [K_i \neq \widehat{K}_i] \\
& \leq  \mathbb{P} [U_i^{1:N} \neq \widehat{U}_i^{1:N}]\\
& = \mathbb{P} [U_i^{1:N} \neq \widehat{U}_i^{1:N}| \mathcal{E}_{F'_i}^c] \mathbb{P} [\mathcal{E}_{F'_i}^c] + \mathbb{P} [U_i^{1:N} \neq \widehat{U}_i^{1:N}| \mathcal{E}_{F'_i}] \mathbb{P} [\mathcal{E}_{F'_i}]\\
& \leq \mathbb{P} [U_i^{1:N} \neq \widehat{U}_i^{1:N} | \mathcal{E}_{F'_i}^c]  +  \mathbb{P} [\mathcal{E}_{F'_i}]\\
& \leq \mathbb{P} [U_i^{1:N} \neq \widehat{U}_i^{1:N} | \mathcal{E}_{F'_i}^c]  +  \mathbb{P} [U_{i-1}^{1:N} \neq \widehat{U}_{i-1}^{1:N}] \displaybreak[0]\\
& \stackrel{(a)}{\leq} N \delta_N  +  \mathbb{P} [U_{i-1}^{1:N} \neq \widehat{U}_{i-1}^{1:N}] \displaybreak[0]\\
& \stackrel{(b)}{\leq} (i-1) N \delta_N  +  \mathbb{P} [U_{1}^{1:N} \neq \widehat{U}_{1}^{1:N}]\displaybreak[0]\\
& \stackrel{(c)}{\leq} i N \delta_N  ,
\end{align*}
where $(a)$ follows because Bob can reconstruct $U^{1:N}_i$ from $(F_i,F_i')= U_i^{1:N}[\mathcal{H}_{X|Y}]$ and $Y_i^{1:N}$ with error probability  less than $N\delta_N$~\cite{Arikan10}, $(b)$ holds by induction, $(c)$ holds by~\cite{Arikan10} and because $\widetilde{K}_0$ is known to Bob. Using the union bound,
\begin{align} 
\mathbf{P}_e(\mathcal{S}_N) \nonumber
&= \mathbb{P} [K_{1:k} \neq \widehat{K}_{1:k}] \nonumber \\
& \leq \sum_{i=1}^k \mathbb{P} [K_i \neq \widehat{K}_i] \displaybreak[0] \nonumber\\\nonumber
& \leq \sum_{i=1}^{k} i N\delta_N  \displaybreak[0] \\
& = \frac{k(k+1)}{2} N\delta_N. \label{eq_errorPr}
\end{align}
\subsubsection{Key uniformity} \label{Sec_unif_model1}

We first prove the uniformity of the key in each block $i$ using the following lemma.

\begin{lem} \label{lem_U1}
In every block $i \in \llbracket 1,k \rrbracket$,  the vector $[K_i, \widetilde{K}_i]$ is nearly uniform, in the sense that
\begin{align*}
|K_i| + |\widetilde{K}_i| - H(K_i \widetilde{K}_i) \leq N \delta_N.
\end{align*}
In particular, $|\widetilde{K}_i| - H( \widetilde{K}_i) \leq N \delta_N$ and $|K_i| - H( K_i) \leq N \delta_N$.
\end{lem}

\begin{proof}
For $i \in \llbracket 1,k \rrbracket$, we have
\begin{align*}
&|K_i| + |\widetilde{K}_i| - H(K_i \widetilde{K}_i) \\
& = |K_i| + |\widetilde{K}_i|-  H( U_i^{1:N} [\mathcal{V}_{X|Z} \backslash \mathcal{H}_{X|Y}]) \\
& \stackrel{(a)}{\leq} |K_i| + |\widetilde{K}_i|- \sum_{j \in \mathcal{V}_{X|Z} \backslash \mathcal{H}_{X|Y}} H(U_i^j | U_i^{1:j-1}) \displaybreak[0]\\
& \stackrel{(b)}{\leq} |K_i| + |\widetilde{K}_i|- \sum_{j \in \mathcal{V}_{X|Z} \backslash \mathcal{H}_{X|Y}} (1- \delta_N) \displaybreak[0]\\
& = (|K_i| + |\widetilde{K}_i|) \delta_N \\
& \leq N\delta_N,
\end{align*}
where $(a)$ holds because conditioning reduces entropy, $(b)$ holds by definition of $\mathcal{V}_{X|Z}$ and because conditioning reduces entropy. Finally, note that since $|K_i| - H(K_i| \widetilde{K}_i) >0$, we have
\begin{align*}
|\widetilde{K}_i| - H( \widetilde{K}_i) 
& \leq |\widetilde{K}_i| - H( \widetilde{K}_i)  + |K_i| - H(K_i| \widetilde{K}_i)\\
& = |K_i| + |\widetilde{K}_i| - H(K_i \widetilde{K}_i).
\end{align*}
\end{proof}
It remains to show that the overall key $K_{1:k}$ is uniform, as well. Specifically, we have 
\begin{align*}
H(K_{1:k}) 
&= \sum_{i=1}^k H(K_{i} | K_{1:i-1}) \\
& \stackrel{(a)}{=} \sum_{i=1}^k H(K_{i} )  \displaybreak[0]\\
& \stackrel{(b)}{\geq} \sum_{i=1}^k ( |K_{i}| -N \delta_N)  \\
& =|K_{1:k}| - kN \delta_N,
\end{align*}
where $(a)$ holds because $X_i^{1:N}$ is independent of of $X_{1:i-1}^{1:N}$ for any $i \in \llbracket 1 ,k \rrbracket$, and $(b)$ holds by Lemma~\ref{lem_U1}. 
Hence,
\begin{align} \label{eq_uniformity}
\textbf{\textup{U}}(\mathcal{S}_N)= |K_{1:k}|  - H(K_{1:k}) \leq k N \delta_N. 
\end{align}

\subsubsection{Strong secrecy} \label{sec_secrecyz}

We first show that secrecy holds for each block using the following lemma .

\begin{lem} \label{lem_sec_block}
For each Block $i \in \llbracket 1, k \rrbracket$, $[K_i, \widetilde{K}_i]$ is a secret key. Specifically, 
\begin{align*}
I\left(K_{i}\widetilde{K}_{i}; M_{i} Z_{i}^{1:N}\right) \leq 2N \delta_N.
\end{align*}
\end{lem}

\begin{proof}
We have 
\begin{align}
& I(K_{i}\widetilde{K}_{i}; F_{i} Z_{i}^{1:N}) \nonumber \\ \nonumber
& = H( K_i \widetilde{K}_{i}) - H(K_{i}\widetilde{K}_{i} | F_{i} Z_{i}^{1:N}) \\ \nonumber
& \leq |K_i| + |\widetilde{K}_{i}| - H(K_{i}\widetilde{K}_{i} F_{i} |Z_{i}^{1:N} ) + H(F_{i} | Z_{i}^{1:N}) \\ \nonumber
& \leq |K_i| + |\widetilde{K}_{i}| + |F_i| - H(K_{i}\widetilde{K}_{i} F_{i} |Z_{i}^{1:N} ) \\ \nonumber
& \stackrel{(a)}{=} |\mathcal{V}_{X|Z} \backslash \mathcal{H}_{X|Y}| + |\mathcal{V}_{X|Z} \cap \mathcal{H}_{X|Y}| \\\nonumber
& \phantom{--}- H(U_i^{1:N}[(\mathcal{V}_{X|Z} \backslash \mathcal{H}_{X|Y}) \cup (\mathcal{V}_{X|Z} \cap \mathcal{H}_{X|Y})] |Z_i^{1:N}) \\\nonumber
& = |\mathcal{V}_{X|Z}| - H(U_i^{1:N}[\mathcal{V}_{X|Z}] |Z_i^{1:N}) \displaybreak[0] \\\nonumber
& \stackrel{(b)}{\leq} |\mathcal{V}_{X|Z}| - \sum_{j \in \mathcal{V}_{X|Z}} H(U_i^{j}|U_i^{1:j-1} Z_i^{1:N}) \displaybreak[0] \\ \nonumber \displaybreak[0]
& \stackrel{(c)}{\leq} |\mathcal{V}_{X|Z}| - \sum_{j \in \mathcal{V}_{X|Z}} (1-\delta_N) \\  \nonumber
& = |\mathcal{V}_{X|Z}| \delta_N \displaybreak[0] \\
& \leq N \delta_N, \label{eqsec1} 
\end{align}
where $(a)$ holds by definition of $K_i$, $\widetilde{K}_i$, and $F_i$, $(b)$ holds because conditioning reduces entropy, $(c)$ holds by definition of $\mathcal{V}_{X|Z}$. Therefore, we obtain
\begin{align*}
& I(K_{i}\widetilde{K}_{i}; M_{i} Z_{i}^{1:N}) \\
& \stackrel{(d)}{=} I(K_{i}\widetilde{K}_{i}; F_{i} (F_{i}' \oplus \widetilde{K}_{i-1}) Z_{i}^{1:N}) \\
& = I(K_{i}\widetilde{K}_{i}; F_{i}  Z_{i}^{1:N}) + I(K_{i}\widetilde{K}_{i}; F_{i}' \oplus \widetilde{K}_{i-1}| F_{i}  Z_{i}^{1:N}) \displaybreak[0]\\
& \stackrel{(e)}{\leq} N \delta_N + I(K_{i}\widetilde{K}_{i} F_{i}  Z_{i}^{1:N} F_{i}'; F_{i}' \oplus \widetilde{K}_{i-1} )\displaybreak[0] \\
& = N \delta_N + H(F_{i}' \oplus \widetilde{K}_{i-1}) - H(F_{i}' \oplus \widetilde{K}_{i-1}|K_{i}\widetilde{K}_{i} F_{i}  Z_{i}^{1:N} F_{i}') \displaybreak[0] \\
& = N \delta_N + H(F_{i}' \oplus \widetilde{K}_{i-1}) - H( \widetilde{K}_{i-1}|K_{i}\widetilde{K}_{i} F_{i}  Z_{i}^{1:N} F_{i}') \\
& = N \delta_N + H(F_{i}' \oplus \widetilde{K}_{i-1}) - H( \widetilde{K}_{i-1}) \displaybreak[0]\\
& \leq N \delta_N + |\widetilde{K}_{i-1}| - H( \widetilde{K}_{i-1}) \displaybreak[0]\\
&  \stackrel{(f)}{\leq} 2 N \delta_N,
\end{align*}
where $(d)$ holds by definition of $M_i$, $(e)$ holds by~(\ref{eqsec1}) and positivity of mutual information, $(f)$ holds by Lemma~\ref{lem_U1}.
\end{proof}

We now state two lemmas that will be used to show that secrecy holds for the global scheme.
\begin{lem} \label{lem_KKtilde}
For $i \in \llbracket 1,k \rrbracket$, we have for $N$ large enough
\begin{align*}
I(K_i;\widetilde{K}_i) 
& \leq \delta_{N}^*,
\end{align*}
where 
\begin{equation} \label{delta_Stardef}
 \delta_{N}^* \triangleq 3 \sqrt{2N \delta_N \log 2} \left( N- \log_2 \left( 3 \sqrt{2N \delta_N \log 2}  \right)\right).
 \end{equation} 
\end{lem}

\begin{proof}
See Appendix~\ref{App_lemKKtilde}
\end{proof}
\begin{lem} \label{lem_secrec}
For $i\in \llbracket 2,k \rrbracket$, define
\begin{align*}
\widetilde{L}_e^{1:i} &\triangleq I \left(K_{1:i} \widetilde{K}_i; M_{1:i} Z^{1:N}_{1:i} \right).
\end{align*}
We have 
\begin{align*}
\widetilde{L}_e^{1:i} - \widetilde{L}_e^{1:i-1} \leq I\left(K_{i}\widetilde{K}_{i}; M_{i} Z_{i}^{1:N}\right) + I\left(K_{i-1} ; \widetilde{K}_{i-1} \right).
\end{align*}
\end{lem}

\begin{proof}
See Appendix~\ref{App_lemsecrec}.
\end{proof}

We thus obtain
\begin{align}
&\textbf{L}(\mathcal{S}_N) \nonumber\\
& = {I} (K_{1:k}; M_{1:k} Z^{1:N}_{1:k}) \nonumber \\  \nonumber
& \leq \widetilde{L}_e^{1:k} \\ \nonumber
& = \sum_{i=2}^k (\widetilde{L}_e^{1:i} - \widetilde{L}_e^{1:i-1}) + \widetilde{L}_e^{1} \\ \nonumber
& \stackrel{(a)}{\leq} \sum_{i=2}^k  \left( I\left(K_{i}\widetilde{K}_{i}; M_{i} Z_{i}^{1:N}\right) + I\left(K_{i-1} ; \widetilde{K}_{i-1} \right) \right) + \widetilde{L}_e^{1} \\ \nonumber
& \leq \sum_{i=1}^k  I\left(K_{i}\widetilde{K}_{i}; M_{i} Z_{i}^{1:N} \right) + \sum_{i=2}^k  I\left(K_{i-1} ; \widetilde{K}_{i-1} \right) \\
& \stackrel{(b)}{\leq} 2kN\delta_N + (k-1) \delta_{N}^*, \label{eq_leakage}
\end{align}
where $(a)$ follows by Lemma~\ref{lem_secrec}, $(b)$ follows by Lemma~\ref{lem_KKtilde} and Lemma~\ref{lem_sec_block}.
\subsubsection{Seed rate}
The seed rate required to initialize the coding scheme is negligible since
$$ \lim_{ k \to \infty}\lim_{ N \to \infty}\frac{|\mathcal{H}_{X|Y} \backslash  \mathcal{V}_{X|Z} |} {kN} \leq \lim_{ k \to \infty}\frac{H(X|Y)}{k}=0.$$
Note that the seed rate may be chosen to decrease exponentially fast to zero with $N$ since we may choose $k = 2^{N^{\alpha}}$, $\alpha < \beta$ and still have 
$\lim_{N \to \infty}\mathbf{P}_e(\mathcal{S}_N) = 0$ by~(\ref{eq_errorPr}), $\lim_{N \to \infty}\mathbf{U}_e(\mathcal{S}_N) = 0$ by~(\ref{eq_uniformity}), and $\lim_{N \to \infty}\mathbf{L}_e(\mathcal{S}_N) = 0$ by~(\ref{eq_leakage}) and (\ref{delta_Stardef}).

\section{Model 2: Secret Key Generation with Rate-Limited Public Communication}  \label{sec_model2}

We now move to the second key generation model, which differs from Model 1 by restricting the public communication to be rate-limited and one way from Alice to Bob. The organization follows that of Section~\ref{sec_model1}. 

\subsection{Secret-key generation model} \label{Secstatmod2}
\begin{figure}
\centering
  \includegraphics[width=8.5cm]{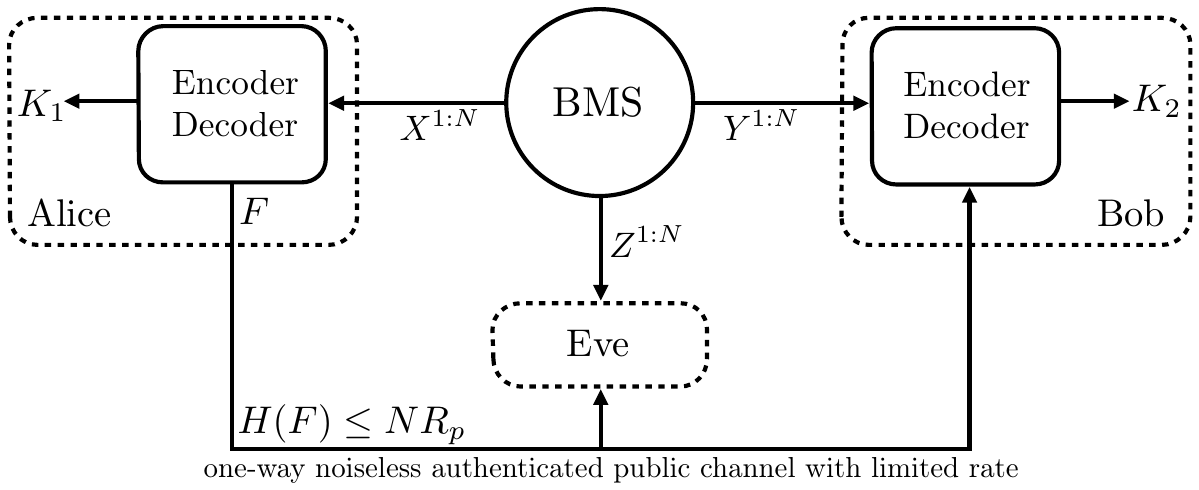}
  \caption{Model 2: Secret-key generation for the BMS model with one-way rate-limited public communication}
  \label{figmodel2}
\end{figure}

As illustrated in Fig.~\ref{figmodel2}, we set again $m=2$ and we use $\mathcal{X}$ instead of $\mathcal{X}_1$, $\mathcal{Y}$ instead of $\mathcal{X}_2$ for convenience. We assume that $\mathcal{X} = \{ 0,1\}$ and that Alice and Bob are constrained to only use one-way communication over an authenticated noiseless public channel with limited rate $R_p \in \mathbb{R} $. We call this setup the BMS model with rate-limited public communication. The following results are known for the model.

\begin{thm}[\!\! {\cite[Th. 2.6]{Csiszar00}}] \label{Th_model2}
Let $(\mathcal{X}\mathcal{Y}\mathcal{Z},p_{XYZ})$ be a BMS and $R_p \in \mathbb{R}_+$ be the public communication rate. If $X \to Y \to Z$, then the one-way rate-limited secret-key capacity is\footnote{See also \cite[Prop. 5.2, Rem. 5.2]{Chou12b} and \cite[Cor. 6]{Watanabe11} for the exact derivation.}
\begin{align*}
& C_{\textup{WSK}}(R_p) = \displaystyle\max_{U} \left( {I}(Y;U) - {I}(Z;U)\right) 
\end{align*}
\vspace*{-1.11em}
\text{ subject to }
\vspace*{-1em}
\begin{align*} 
& R_p = {I}(U;X) - I(U;Y), \\ \nonumber
&  U \to X \to Y\to Z, \\ \nonumber
&|\mathcal{U}| \leq |\mathcal{X}| .
\end{align*}
\end{thm}
Closed-form expressions of the secret-key capacity are only known for specific sources. See the following example.

\begin{ex} \label{examplecapl}
Assume $\mathcal{X} = \mathcal{Y} = \mathcal{Z} = \{ 0,1\}$ and $X \sim \mathcal{B}(1/2)$. Set $Y \triangleq  X \oplus B_1$ and $Z \triangleq Y \oplus B_2$, with $B_1\sim \mathcal{B}(p)$, $B_2\sim \mathcal{B}(q)$, where $\oplus$ denotes the modulo-$2$ addition. Then, by \cite[Prop. 5.3]{Chou12b}, the secret-key capacity is 
\begin{multline*} 
C_{\text{WSK}} (R_p) \\ \triangleq  \begin{cases}
   {H}_b(p \star \beta_0 \star q) - {H}_b(p \star \beta_0), 
 &\text{if } R_p \leq {H}(X|Y), \\
    {H}_b(p \star q)  - {H}_b(p), &\text{if } R_p \geq {H}(X|Y), 
   \end{cases}
\end{multline*}
where $\beta_0$ must satisfy\footnote{ Note that (\ref{beta0}) has two symmetric solutions.}  
\begin{equation} \label{beta0}
{H}_b(p \star \beta_0)-{H}_b(\beta_0)=R_p,
\end{equation}
${H}_b(\cdot)$ is the binary entropy function, and the associative and commutative operation $\star$ is defined as $p \star \beta_0 = (1-\beta_0)p + \beta_0(1-p)$. 
\end{ex}

 When the eavesdropper has no access to the source component $Z$, one obtains the following expression for the secret-key capacity.

\begin{cor} \label{cor_ratelim}
Let $(\mathcal{X}\mathcal{Y}\mathcal{Z},p_{XYZ})$ be a BMS and $R_p \in \mathbb{R}_+$ be the public communication rate. The one-way rate-limited secret-key capacity is
\begin{align*}
 C_{\textup{SK}}(R_p) = \displaystyle\max_{U} {I}(Y;U) 
\end{align*}
\vspace*{-1.11em}
\text{ subject to }
\vspace*{-1em}
\begin{align*} 
& R_p = {I}(U;X) - I(U;Y), \\ \nonumber
&  U \to X \to Y, \\ \nonumber
&|\mathcal{U}| \leq |\mathcal{X}|.
\end{align*}
\end{cor}
 
The practical justification for Model 2 is similar to that for Model 1; however, Model 2 allows us to account for rate-limited communication constraints, which is relevant in applications with stringent bandwidth constraints, such as wireless sensor networks. We will also see in Section~\ref{Sec_bio} that such constraint may account for privacy-leakage constraints in biometric systems.
 
The main challenge in designing a coding scheme for Model~2 is to address the problem of vector quantization with side information at the receiver. Previous polar coding results on lossy source coding with lossless reconstruction of the vector quantized version of the source are reported in~\cite{Korada10,Honda13}; our contribution is to extend these results when side information is available at the receiver, and to show how to apply such technique to secret-key generation with rate-limited public communication. 

\subsection{Polar coding scheme} \label{Sec_scheme2}

Let $n \in \mathbb{N}$ and $N \triangleq 2^n$. Fix a joint probability distribution $p_{XU}$ such that $I(Y;U)-I(Z;U)>0$, but we do not assume $X \to Y \to Z$. Denote $V^{1:N} \triangleq U^{1:N} G_N$, the polar transform of a vector $U^{1:N}$ with i.i.d components according to the marginal distribution $p_U$. For $\delta_N \triangleq 2^{-N^{\beta}}$, where $\beta \in ]0,1/2[$, define the following sets.
\begin{align*}
\mathcal{H}_{U}  \triangleq &  \left\{ i\in \llbracket 1,N\rrbracket : {H} \left(V^i|V^{1:i-1} \right) \geq  \delta_N \right\} , \\
\mathcal{V}_{U|Z}  \triangleq &  \left\{ i\in \llbracket 1,N\rrbracket : {H} \left(V^i|V^{1:i-1} Z^{1:N} \right) \geq 1- \delta_N \right\} , \\
\mathcal{V}_{U|Y}  \triangleq &  \left\{ i\in \llbracket 1,N\rrbracket : {H} \left(V^i|V^{1:i-1} Y^{1:N} \right) \geq 1- \delta_N \right\} , \\
\mathcal{H}_{U|Y} \triangleq & \left\{ i\in \llbracket 1,N\rrbracket : {H}\left(V^i|V^{1:i-1} Y^{1:N}\right) \geq \delta_N \right\}, \\
\mathcal{H}_{U|X} \triangleq & \left\{ i\in \llbracket 1,N\rrbracket : {H}\left(V^i|V^{1:i-1} X^{1:N}\right) \geq \delta_N \right\}.
\end{align*}

The encoding and decoding algorithms are given in Algorithm~\ref{alg:encoding_2} and Algorithm~\ref{alg:decoding_2}. The high-level principles are similar to that of Algorithm~\ref{alg:encoding_1} and Algorithm~\ref{alg:decoding_1}, and we only highlight here the differences. Instead of directly operating on the source symbols, Alice first constructs a vector quantized version $\widetilde{V}^{1:N}$ of $X^{1:N}$, whose distribution is close to that of $V^{1:N}$. This statement is made more precise in Lemma~\ref{lemDivprob}, but a crucial part of the proof is to introduce a stochastic encoder, as in successive cancellation encoding for lossy source coding~\cite{Korada10,Honda13}. The randomness $R_1$ used in the encoder is publicly transmitted to Bob and reused over several blocks so that its rate vanishes to zero as the number of blocks increases; however, reusing $R_1$ creates additional dependencies between the variables of the different blocks, which must be carefully taken into account in the secrecy analysis. The choice of public messages and keys is then similar to those in Section~\ref{Sec_scheme1}, using $\widetilde{V}^{1:N}$ instead of $X^{1:N}$.

\begin{algorithm}[]
  \caption{Alice's encoding algorithm for Model 2}
  \label{alg:encoding_2}
  \begin{algorithmic} [1]   
    \REQUIRE $\widetilde{K}_0$, a secret key of size $|(\mathcal{H}_{U|Y} \backslash {\mathcal{V}}_{U|X}) \backslash  \mathcal{V}_{U|Z}|$ shared by Alice and Bob beforehand; for every Block $i \in \llbracket 1,k\rrbracket$, the observations $X_i^{1:N}$ from the source; $\mathcal{A}_{UYZ}$ a subset of $\mathcal{V}_{U|Z} \backslash \mathcal{H}_{U|Y}$ with size $|(\mathcal{H}_{U|Y} \backslash {\mathcal{V}}_{U|X}) \backslash  \mathcal{V}_{U|Z}|$; a vector $R_1$ of $|\mathcal{V}_{U|X}|$ uniformly distributed bits.
    \STATE Transmit $R_1$ publicly to Bob
    \FOR{Block $i=1$ to $k$}
    \STATE $R_i \leftarrow R_1$
    \STATE $\widetilde{V}_i^{1:N}[\mathcal{V}_{U|X}]\leftarrow R_i$
    \STATE Given $X_i^{1:N}$, successively draw the remaining bits of $\widetilde{V}_i^{1:N}$ according to $\widetilde{p}_{V_i^{1:N}X_i^{1:N}} \triangleq \prod_{j=1}^N \widetilde{p}_{V_i^{j}|V_i^{j-1} X^{1:N}} p_{X^{1:N}}$ with 
    \begin{align}\label{eq_VQ_def}
&\widetilde{p}_{V_i^j|V_i^{1:j-1}X^{1:N}} (v^j|\widetilde{V}_i^{1:j-1}X_i^{1:N}) \nonumber \\  &\triangleq \! \begin{cases}
 {p}_{V^j|V^{1:j-1}X^{1:N}} (v^j|\widetilde{V}_i^{1:j-1}X_i^{1:N}) & \!\!\!  \text{if } j \! \in \mathcal{H}_U \backslash {\mathcal{V}}_{U|X}\\
 {p}_{V^j|V^{1:j-1}} (v^j|\widetilde{V}_i^{1:j-1}) & \!\!\! \text{if } j \!\in \mathcal{H}_U^c
 \end{cases}
    \end{align}
    \STATE $\widetilde{K}_i \leftarrow \widetilde{V}_i^{1:N} [\mathcal{A}_{UYZ}]$
    \STATE $K_i \leftarrow \widetilde{V}_i^{1:N} [(\mathcal{V}_{U|Z} \backslash \mathcal{H}_{U|Y}) \backslash \mathcal{A}_{UYZ}]$
    \STATE $F_i \leftarrow \widetilde{V}_i^{1:N} [(\mathcal{H}_{U|Y} \backslash {\mathcal{V}}_{U|X}) \cap \mathcal{V}_{U|Z}]$
    \STATE $F_i' \triangleq \widetilde{V}_i^{1:N} [(\mathcal{H}_{U|Y} \backslash {\mathcal{V}}_{U|X}) \backslash\mathcal{V}_{U|Z} ]$
    \STATE Transmit $M_i \leftarrow [F_i, F_i' \oplus \widetilde{K}_{i-1}]$ publicly to Bob.
    \ENDFOR
    \RETURN ${K}_{1:k} \leftarrow [ {K}_{1}, {K}_2, \ldots, {K}_k ]$
  \end{algorithmic}
\end{algorithm}

\begin{algorithm}[]
  \caption{Bob's decoding algorithm for Model 2}
  \label{alg:decoding_2}
  \begin{algorithmic}  [1]  
    \REQUIRE The secret key $\widetilde{K}_0$ and the set $\mathcal{A}_{UYZ}$ defined in Algorithm \ref{alg:encoding_2}; for every Block $i \in \llbracket 1,k\rrbracket$, the observations $Y_i^{1:N}$ from the source, the message $M_i$. transmitted by Alice; the vector $R_1$ transmitted by Alice.
    \FOR{Block $i=1$ to $k$}
    \STATE Form $\widetilde{V}_i^{1:N}[\mathcal{H}_{U|Y}]$ from $(M_i,\widetilde{K}_{i-1},R_i)$
    \STATE Create an estimate $\widehat{V}_i^{1:N}$ of ${V}_i^{1:N}$ with the successive cancellation decoder of \cite{Arikan10}
    \STATE $\widehat{K}_i\leftarrow \widehat{V}_i^{1:N} [(\mathcal{V}_{U|Z} \backslash \mathcal{H}_{U|Y}) \backslash \mathcal{A}_{UYZ}]$
    \STATE $\widetilde{K}_i\leftarrow \widehat{V}_i^{1:N} [\mathcal{A}_{UYZ}]$
    \ENDFOR
    \RETURN $\widehat{K}_{1:k} \leftarrow [ \widehat{K}_{1}, \widehat{K}_2, \ldots, \widehat{K}_k ]$
  \end{algorithmic}
\end{algorithm}

\begin{rem} \label{rm:efficient_use_key2}
One may actually use $U_k^{1:N} [\mathcal{V}_{U|Z} \backslash \mathcal{H}_{U|Y}]$ as the key $K_k$ and slightly increase the key rate in Algorithm~\ref{alg:encoding_2}. However, one does not distinguish the last block from the others for convenience -- see Remark \ref{rm:efficient_use_key}.
\end{rem}
As shown in Section~\ref{SecproofTh2}, the analysis of Algorithm~\ref{alg:encoding_2} and Algotithm~\ref{alg:decoding_2} leads to the following result.
\begin{thm} \label{Th2}
Consider a BMS $(\mathcal{X}\mathcal{Y}\mathcal{Z},p_{XYZ})$. Assume that Alice and Bob share a secret seed and let $R_p \in \mathbb{R}^+$ be the public communication rate. The secret-key rate defined by
\begin{align*}
&\displaystyle\max_{U} \left( {I}(Y;U) - {I}(Z;U)\right) \\
\text{ subject to } \hspace*{1em}
& R_p = {I}(U;X) - I(U;Y), \\ \nonumber
&  U \to X \to Y, \\ \nonumber
&|\mathcal{U}| \leq |\mathcal{X}| .
\end{align*}
 is achieved by the polar coding scheme of Algorithm~\ref{alg:encoding_2} and Algorithm~\ref{alg:decoding_2}, which involves a chaining of $k$ blocks of size $N$, and whose computational complexity is $O(kN \log N)$. Moreover, the seed rate can be chosen in $o\left( 2^{-N^{\alpha} }\right)$, $\alpha < 1/2$.
\end{thm}
\begin{proof} 
See Section~\ref{SecproofTh2}.
\end{proof}
The following corollary states sufficient conditions to avoid block encoding.
\begin{cor} \label{Cor_2}
If $X \to Y \to Z$, $X \sim \mathcal{B}(1/2)$, and the test-channels $p_{Y|X}$ and $p_{Z|X}$ are symmetric,\footnote{As in Example \ref{examplecapl} for instance} then the secret-key capacity of Theorem~\ref{Th_model2} is achieved by the polar coding scheme for Block 1 in Algorithm~\ref{alg:encoding_2} with $\mathcal{A}_{UYZ} = \emptyset$, $R_1$ a constant sequence, and a seed rate in $o(N)$. 
\end{cor}
\begin{proof}
See Appendix~\ref{App_cor2}.
\end{proof}
Finally, the following proposition provides sufficient conditions to avoid block encoding and a pre-shared seed. The proof is similar to that of Theorem~\ref{Th2} and Corollary \ref{Cor_2} and is omitted.
\begin{prop}
If the eavesdropper has no access to correlated observations of the source, $X \sim \mathcal{B}(1/2)$, and the test-channel $p_{Y|X}$ is symmetric, then the secret-key capacity of Corollary~\ref{cor_ratelim} is achieved by the polar coding scheme for Block $1$ in Algorithm~\ref{alg:encoding_2} with $\mathcal{A}_{UYZ} = \emptyset$, $Z = \emptyset$, $F_1' = \emptyset$, $K_1 \triangleq \widetilde{V}_1^{1:N}[\mathcal{H}_{U|Y}^c]$, $F_1 \triangleq  \widetilde{V}_1^{1:N}[\mathcal{H}_{U|Y} \backslash \mathcal{V}_{U|X}] $, and $R_1 $ a constant sequence. 
\end{prop}

\subsection{Analysis of polar coding scheme: Proof of Theorem~\ref{Th2}} \label{SecproofTh2}
A functional dependence graph for the coding scheme of Section~\ref{Sec_scheme2} is depicted in Fig.~\ref{figFGD2} for convenience.
\begin{figure}
\centering
  \includegraphics[width=8.5cm]{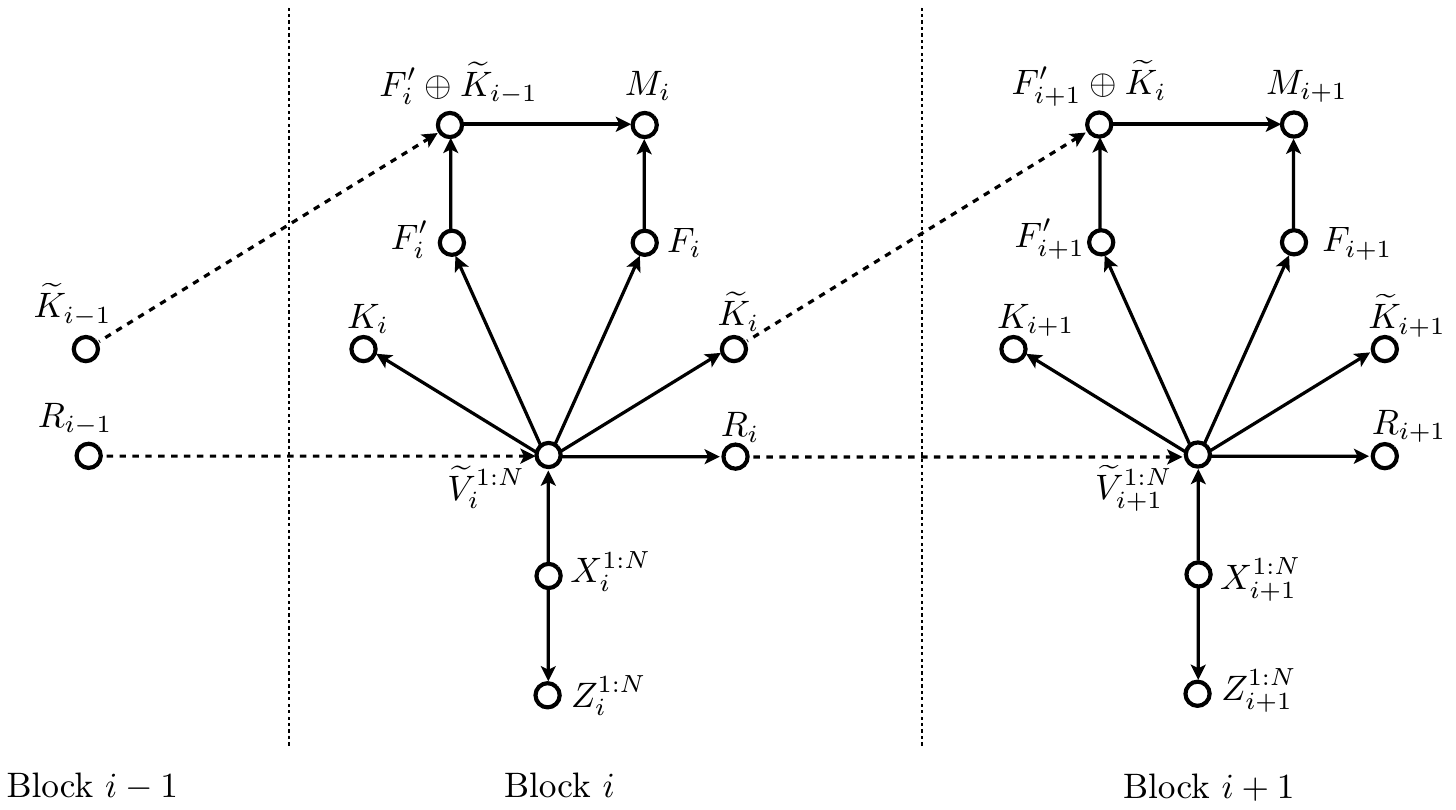}
  \caption{Functional dependence graph of the block encoding scheme}
  \label{figFGD2}
\end{figure}

\subsubsection{Preliminary result}
\begin{lem} \label{lemDivprob}
For every $i\in \llbracket 1,k\rrbracket$, the random variable $\widetilde{V}_i^{1:N}$ resulting from Algorithm~\ref{alg:encoding_2} has a joint distribution $\widetilde{p}_{X_i^{1:N}V_i^{1:N}} \triangleq \widetilde{p}_{V_i^{1:N}|X^{1:N}} p_{X^{1:N}}$ with $X_i^{1:N}$ such that
$$
\mathbb{D}(p_{X^{1:N}V^{1:N}} || \widetilde{p}_{X_i^{1:N}V_i^{1:N}}) \leq N \delta_N,
$$
Hence, by Pinsker's inequality
$$
\mathbb{V}(p_{X^{1:N}V^{1:N}}, \widetilde{p}_{X_i^{1:N}V_i^{1:N}}) \leq \sqrt{2 \log2} \sqrt{N \delta_N}.
$$
\end{lem}
\begin{proof}
See Appendix~\ref{App_lemdiv}.
\end{proof}

\subsubsection{Existence of $\mathcal{A}_{UYZ}$} \label{existence_AUYZ}
Observe that 
\begin{align*}
&|\mathcal{V}_{U|Z} \backslash \mathcal{H}_{U|Y}| - |(\mathcal{H}_{U|Y} \backslash {\mathcal{V}}_{U|X}) \backslash  \mathcal{V}_{U|Z}| \\
& = |\mathcal{V}_{U|Z}| -| \mathcal{H}_{U|Y}| + | (\mathcal{H}_{U|Y} \cap {\mathcal{V}}_{U|X})  \backslash  \mathcal{V}_{U|Z}| \\
&\geq |\mathcal{V}_{U|Z}| -| \mathcal{H}_{U|Y}|.
\end{align*}
Hence, by Lemma~\ref{lemnewcard} and~\cite{Arikan10}, we have
\begin{multline*}
	\lim_{N \to \infty} (|\mathcal{V}_{U|Z} \backslash \mathcal{H}_{U|Y}| - |(\mathcal{H}_{U|Y} \backslash {\mathcal{V}}_{U|X}) \backslash  \mathcal{V}_{U|Z}|) /N \\ \geq  H(U|Z) - H(U|Y).
\end{multline*}
Since $I(Y;U) - I(Z;U) >0$, we have for $N$ large enough $|\mathcal{V}_{U|Z} \backslash \mathcal{H}_{U|Y}| - |(\mathcal{H}_{U|Y} \backslash {\mathcal{V}}_{U|X}) \backslash  \mathcal{V}_{U|Z}| >0$, and we conclude that $\mathcal{A}_{UYZ}$ exists.

\subsubsection{Communication rate}
The total communication is
\begin{align*}
& |R_1| + \sum_{i=1}^{k} (|F_i| + |F_i'|) \\
& = |R_1| + \sum_{i=1}^{k} |\mathcal{H}_{U|Y} \backslash {\mathcal{V}}_{U|X}| \\
& = |{\mathcal{V}}_{U|X}| +  k |\mathcal{H}_{U|Y} \backslash {\mathcal{V}}_{U|X}|\\
& = |{\mathcal{V}}_{U|X}| + k (|\mathcal{H}_{U|Y}| - |{\mathcal{V}}_{U|X}|) 
\end{align*}
where the last equality holds because $U \to X \to Y$ and thus $ \mathcal{V}_{U|X} \subset \mathcal{V}_{U|Y} \subset \mathcal{H}_{U|Y}$. Hence, the communication rate is by Lemma \ref{lemnewcard} and~\cite{Arikan10},
\begin{multline*}
\lim_{N \to \infty} \frac{ |{\mathcal{V}}_{U|X}| + k (|\mathcal{H}_{U|Y}| - |{\mathcal{V}}_{U|X}|)}{kN} \\= I(X;U) - I(Y;U) + \frac{H(U|X)}{k}.
\end{multline*}

\subsubsection{Key rate}
The length of the key generated is 
\begin{align*}
|K_{1:k}|
& = \sum_{i=1}^{k} |K_i| \\
& =k |(\mathcal{V}_{U|Z} \backslash \mathcal{H}_{U|Y})\backslash \mathcal{A}_{UYZ}|\\
& = k (|\mathcal{V}_{U|Z}| -| \mathcal{H}_{U|Y}| + | (\mathcal{H}_{U|Y} \cap {\mathcal{V}}_{U|X})  \backslash  \mathcal{V}_{U|Z}|) \\
& \geq k (|\mathcal{V}_{U|Z}| - |\mathcal{H}_{U|Y} |) .
\end{align*}
Hence, the key rate is by Lemma \ref{lemnewcard} and \cite{Arikan10},
$$
\lim_{N \to \infty} \frac{|K_{1:k}|}{kN} \geq I(Y;U) - I(Z;U).
$$

\subsubsection{Reliability} \label{Sec_reliability_model2}
For $i \in \llbracket 1,k\rrbracket$, Bob forms $\widehat{V}^{1:N}_i$ from $(F_i,F_i',R_i)= \widetilde{V}_i^{1:N}[\mathcal{H}_{U|Y}]$ and $Y_i^{1:N}$ with the successive cancellation encoder of \cite{Arikan10}. Consider an optimal coupling~\cite{Aldous83,Korada10} between $\widetilde{p}_{V_i^{1:N}}$ and $p_{V_i^{1:N}}$ such that $\mathbb{P} [\mathcal{E}] = \mathbb{V}(\widetilde{p}_{V_i^{1:N}} ,p_{V_i^{1:N}})$, where $\mathcal{E} \triangleq \{ \widetilde{V}_i^{1:N} \neq {V}_i^{1:N} \}$. For $i \in \llbracket 2,k\rrbracket$, note that $F_i'$ is correctly received only when Bob has correctly estimated $\widetilde{K}_{i-1}$, i.e., when $\widetilde{V}_{i-1}^{1:N}$ is correctly reconstructed. We note $\widehat{F}_i'$ the estimate of $F_i'$ formed by Bob from $\widetilde{V}_{i-1}^{1:N}$ and define $\mathcal{E}_{F'_i} \triangleq \{ F_i' \neq  \widehat{F}_i'\}$. We then have
\begin{align*}
&\mathbb{P}[\widehat{V}^{1:N}_i \neq \widetilde{V}^{1:N}_i ]\\
& = \mathbb{P}[\widehat{V}^{1:N}_i \neq \widetilde{V}^{1:N}_i | \mathcal{E} \cup \mathcal{E}_{F'_i} ] \mathbb{P} [\mathcal{E} \cup \mathcal{E}_{F'_i}] \\
& \phantom{--}+ \mathbb{P}[\widehat{V}^{1:N}_i \neq \widetilde{V}^{1:N}_i | \mathcal{E}^c \cap \mathcal{E}_{F'_i}^c] \mathbb{P} [\mathcal{E}^c \cap \mathcal{E}_{F'_i}^c] \\
& \leq \mathbb{P} [\mathcal{E} \cup \mathcal{E}_{F'_i}] + \mathbb{P}[\widehat{V}^{1:N}_i \neq \widetilde{V}^{1:N}_i | \mathcal{E}^c \cap \mathcal{E}_{F'_i}^c] \\
& \leq \mathbb{P} [\mathcal{E}] +  \mathbb{P} [\mathcal{E}_{F'_i}] + \mathbb{P}[\widehat{V}^{1:N}_i \neq \widetilde{V}^{1:N}_i | \mathcal{E}^c \cap \mathcal{E}_{F'_i}^c] \\
 & = \mathbb{V}(\widetilde{p}_{V_i^{1:N}} ,p_{V_i^{1:N}}) +  \mathbb{P} [\mathcal{E}_{F'_i}] + \mathbb{P}[\widehat{V}^{1:N}_i \neq \widetilde{V}^{1:N}_i | \mathcal{E}^c \cap \mathcal{E}_{F'_i}^c] \displaybreak[0] \\
& = \mathbb{V}(\widetilde{p}_{V_i^{1:N}} ,p_{V_i^{1:N}}) +  \mathbb{P} [\mathcal{E}_{F'_i}] + \mathbb{P}[\widehat{V}^{1:N}_i \neq {V}^{1:N}_i | \mathcal{E}^c \cap \mathcal{E}_{F'_i}^c] \displaybreak[0] \\
& \leq \mathbb{V}(\widetilde{p}_{X_i^{1:N}V_i^{1:N}} ,p_{X_i^{1:N}V_i^{1:N}}) \\
& \phantom{--}+   \mathbb{P} [\mathcal{E}_{F'_i}] + \mathbb{P}[\widehat{V}^{1:N}_i \neq {V}^{1:N}_i | \mathcal{E}^c \cap \mathcal{E}_{F'_i}^c] \displaybreak[0] \\
& \stackrel{(a)}{\leq} \sqrt{2 \log2} \sqrt{N \delta_N} +  \mathbb{P} [\mathcal{E}_{F'_i}] + \mathbb{P}[\widehat{V}^{1:N}_i \neq {V}^{1:N}_i | \mathcal{E}^c\cap \mathcal{E}_{F'_i}^c] \\
& \stackrel{(b)}{\leq } \sqrt{2 \log2} \sqrt{N \delta_N} +  \mathbb{P} [\mathcal{E}_{F'_i}] + N \delta_N \\
& \leq  \sqrt{2 \log2} \sqrt{N \delta_N} + N \delta_N +  \mathbb{P}[\widehat{V}^{1:N}_{i-1} \neq \widetilde{V}^{1:N}_{i-1} ]  \\
& \stackrel{(c)}{\leq } (i-1) (\sqrt{2 \log2} \sqrt{N \delta_N}  + N \delta_N ) +  \mathbb{P}[\widehat{V}^{1:N}_{1} \neq \widetilde{V}^{1:N}_{1} ]  \\
& \stackrel{(d)}{\leq } i (\sqrt{2 \log2} \sqrt{N \delta_N}  + N \delta_N ),
\end{align*}
where $(a)$ holds by Lemma~\ref{lemDivprob}, $(b)$ holds because $\mathbb{P}[\widehat{V}^{1:N}_i \neq {V}^{1:N}_i | \mathcal{E}^c \cap \mathcal{E}_{F'_i}^c] \leq N \delta_N$ by \cite{Arikan10}, $(c)$ holds by induction, $(d)$ holds by \cite{Arikan10} and because $\widetilde{K}_0$ is known to Bob.

Hence, $\mathbb{P} [K_i \neq \widehat{K}_i] \leq i(\sqrt{2 \log2} \sqrt{N \delta_N} + N \delta_N)$. Then, similarly to Section~\ref{Sec_errProb1}, we obtain with a union bound
\begin{align} 
\mathbf{P}_e(\mathcal{S}_N) 
 \leq \frac{k(k+1)}{2} (\sqrt{2 \log2} \sqrt{N \delta_N} + N \delta_N) . \label{eq_errorPr2}
\end{align}
\subsubsection{Key uniformity}
We first show that the key is nearly uniform for every block in the following lemma.
\begin{lem} \label{lem_U2}
For every block  $i \in \llbracket 1,k \rrbracket$, the vector $[K_i, \widetilde{K}_i, F_i,R_1]$ is nearly uniform, in the sense that 
\begin{align*}
\mathbb{V} ( p_{K_i, \widetilde{K}_i,F_i, R_1}, q_{\mathcal{U}_{K, \widetilde{K},F, R}} ) \leq 2 \sqrt{2 \log2} \sqrt{N \delta_N},
\end{align*}
where $q_{\mathcal{U}_{K, \widetilde{K},F, R}} $ is the uniform distribution over $\llbracket 1, 2^{|(\mathcal{V}_{U|Z} \backslash \mathcal{H}_{U|Y}) \cup ((\mathcal{H}_{U|Y} \backslash {\mathcal{V}}_{U|X}) \cap \mathcal{V}_{U|Z})  \cup \mathcal{V}_{U|X}|}\rrbracket$.
\end{lem}

\begin{proof}
We have
\begin{align*}
&\mathbb{V} ( p_{K_i, \widetilde{K}_i ,F_i,R_i}, q_{\mathcal{U}_{K, \widetilde{K},F, R}} )\\
& \stackrel{(a)}{\leq} \mathbb{V} ( \widetilde{p}_{V_i^{1:N} [\mathcal{V}_{U}]}, q_{\mathcal{U}_{\mathcal{V}_{U}}} )\\
& \stackrel{(b)}{\leq} \mathbb{V} ( \widetilde{p}_{V_i^{1:N} [\mathcal{V}_{U}]}, {p}_{V_i^{1:N} [\mathcal{V}_{U}]} ) + \mathbb{V} ({p}_{V_i^{1:N} [\mathcal{V}_{U}]}, q_{\mathcal{U}_{\mathcal{V}_{U}}} ) \\
& \stackrel{(c)}{\leq} \sqrt{2 \log2} \sqrt{N \delta_N} + \mathbb{V} ({p}_{V_i^{1:N} [\mathcal{V}_{U}]}, q_{\mathcal{U}_{\mathcal{V}_{U}}} ) \\
& \stackrel{(d)}{\leq} \sqrt{2 \log2} \sqrt{N \delta_N} + \sqrt{2\log2} \sqrt{\mathbb{D} ({p}_{V_i^{1:N} [\mathcal{V}_{U}]} || q_{\mathcal{U}_{\mathcal{V}_{U}}} ) } \\
& = \sqrt{2 \log2} \sqrt{N \delta_N} + \sqrt{2\log2} \sqrt{ |\mathcal{V}_{U}| - H(V_i^{1:N} [\mathcal{V}_{U}])} \\
& \stackrel{(e)}{\leq} 2 \sqrt{2 \log2} \sqrt{N \delta_N},
\end{align*}
where $(a)$ holds because $\mathcal{V}_{U|Z} \subset \mathcal{V}_{U}$ and $\mathcal{V}_{U|X} \subset \mathcal{V}_{U}$ with $q_{\mathcal{U}_{\mathcal{V}_{U}}}$ the uniform distribution over $\llbracket 1, 2^{|\mathcal{V}_{U}|}\rrbracket$, $(b)$ holds by the triangle inequality, $(c)$ holds by Lemma~\ref{lemDivprob}, $(d)$ holds by Pinsker's inequality, $(e)$ holds because similar to the proof of Lemma~\ref{lem_U1}, we have $|\mathcal{V}_{U}| - H(V_i^{1:N} [\mathcal{V}_{U}]) \leq N \delta_N$.
\end{proof}
From Lemma~\ref{lem_U2}, we derive the following lemmas.

\begin{lem} \label{lem_U2_Div}
For $i \in \llbracket 1,k \rrbracket$, we have for $N$ large enough
\begin{align*}
|K_i| + |\widetilde{K}_i|  - H(K_i \widetilde{K}_i ) \leq \delta_N^{(1)},
\end{align*}
where 
\begin{equation} \label{delta_1_def}
\delta_N^{(1)} \triangleq 2 \sqrt{2 \log2} \sqrt{N \delta_N} ( N - \log_2  (2 \sqrt{2 \log2} \sqrt{N \delta_N})).
\end{equation}
In particular, we also have $|K_i|   - H(K_i) \leq \delta_N^{(1)}$ and $|\widetilde{K}_i|   - H(\widetilde{K}_i) \leq \delta_N^{(1)}$.
\end{lem}

\begin{proof}
See Appendix~\ref{App_lem_U2_Div}.
\end{proof}

\begin{lem}  \label{lem_KKtilde2}
For $i \in \llbracket 1,k \rrbracket$, we have for $N$ large enough
\begin{align*}
I(K_i;\widetilde{K}_i R_1)  \leq \delta_N^{(2)} \quad\text{and}\quad I(\widetilde{K}_i;R_1)  \leq \delta_N^{(2)},
\end{align*}
where 
\begin{equation} \label{delta_2_def}
 \delta_{N}^{(2)} \triangleq 6\sqrt{2 \log2} \sqrt{N \delta_N} (N - \log_2( 6\sqrt{2 \log2} \sqrt{N \delta_N})).
\end{equation} 
\end{lem}

\begin{proof}
See Appendix \ref{App_lem_KKtilde2}.
\end{proof}

We now show that the global key $K_{1:k}$ is uniform. Specifically, we have 
\begin{align*}
H(K_{1:k}) 
&= \sum_{i=1}^k H(K_{i} | K_{1:i-1}) \\
& \geq \sum_{i=1}^k H(K_{i} | K_{1:i-1} R_1 ) \\
& \stackrel{(a)}{=} \sum_{i=1}^k H(K_{i} | R_1) \\
& = \sum_{i=1}^k H(K_{i} )  - \sum_{i=1}^k I (K_{i} ; R_1) \\
& \stackrel{(b)}{\geq} \sum_{i=1}^k H(K_{i} ) - k \delta_N^{(2)} \displaybreak[0]\\
& \stackrel{(c)}{\geq} \sum_{i=1}^k ( |K_{i}| - \delta_N^{(1)}) - k \delta_N^{(2)} \displaybreak[0]\\
& = |K_{1:k}| - k (\delta_N^{(1)} + \delta_N^{(2)})
\end{align*}
where $(a)$ holds because $K_i \to R_1 \to  K_{1:i-1}$ for any $i\in \llbracket 1 ,k \rrbracket$, $(b)$ holds by Lemma~\ref{lem_KKtilde2}, $(c)$ holds by Lemma~\ref{lem_U2_Div}. 
Hence,
\begin{align} \label{eq_uniformity2}
\textbf{\textup{U}}(\mathcal{S}_N)= |K_{1:k}|  - H(K_{1:k}) \leq k (\delta_N^{(1)} + \delta_N^{(2)}).
\end{align}

\subsubsection{Strong secrecy} \label{sec_secrecy2}

Because of the successive cancellation encoding, the secrecy analysis is more involved than for Model~1.

\begin{lem} \label{lem_secre_cra}
For $i \in \llbracket 1,k \rrbracket$, we have for $N$ large enough
\begin{equation*}
I(\widetilde{V}_i^{1:N}[\mathcal{V}_{U|Z}] ;Z_i^{1:N}) \leq \delta_N^{(3)},
\end{equation*}
where 
\begin{equation} \label{delta_3_def}
\delta_N^{(3)} \triangleq 3 \sqrt{2 \log2} \sqrt{ N \delta_N} ( N - \log_2 (3 \sqrt{2 \log2} \sqrt{ N \delta_N} ) ).
\end{equation}
\end{lem}
\begin{proof}
See Appendix~\ref{App_lem_secre_cra}.
\end{proof}
The following lemma shows that secrecy holds for each block.
\begin{lem} \label{lem_sec_block2}
For each Block $i \in \llbracket 1, k \rrbracket$, $[K_i, \widetilde{K}_i]$ is a secret key in the sense that
\begin{align*}
I\left(K_{i}\widetilde{K}_{i}; R_1 M_{i} Z_{i}^{1:N}\right) \leq 2\delta_N^{(1)} + \delta_N^{(2)} + \delta_N^{(3)}.
\end{align*}
\end{lem}

\begin{proof}
By the proof of Lemma~\ref{lem_U2}, we have
\begin{align*}
\mathbb{V} ( \widetilde{p}_{V_i^{1:N}[\mathcal{V}_{U|Z}]}, q_{\mathcal{U}_{\mathcal{V}_{U|Z}}} ) \leq 2 \sqrt{2 \log2} \sqrt{N \delta_N},
\end{align*}
where $q_{\mathcal{U}_{\mathcal{V}_{U|Z}}}$ is the uniform distribution over $\llbracket 1, 2^{|\mathcal{V}_{U|Z} |}\rrbracket$, and by the proof of Lemma~\ref{lem_U2_Div}, we have 
\begin{equation} \label{eq_Unif_int}
|\mathcal{V}_{U|Z}|  - H(V_i^{1:N}[\mathcal{V}_{U|Z}]) \leq \delta_N^{(1)}.
\end{equation}
Therefore,
\begin{align}
& I(K_{i}\widetilde{K}_{i}; R_1 F_{i} Z_{i}^{1:N}) \nonumber \\ \nonumber
& = H( K_i \widetilde{K}_{i}) - H(K_{i}\widetilde{K}_{i} | R_1 F_{i} Z_{i}^{1:N}) \\ \nonumber
& \leq |K_i| + |\widetilde{K}_{i}| - H(K_{i}\widetilde{K}_{i} R_1 F_{i} Z_{i}^{1:N} ) + H( R_1 F_{i} Z_{i}^{1:N}) \\ \nonumber
& = |K_i| + |\widetilde{K}_{i}| - H(K_{i}\widetilde{K}_{i} R_1 F_{i} |Z_{i}^{1:N} ) + H(F_{i} R_1 | Z_{i}^{1:N}) \\ \nonumber
& \leq |K_i| + |\widetilde{K}_{i}| + |F_i| + |R_1| - H(K_{i}\widetilde{K}_{i} R_1 F_{i} |Z_{i}^{1:N} ) \\ \nonumber
& \stackrel{(a)}{\leq} |\mathcal{V}_{U|Z}| - H(\widetilde{V}_i^{1:N}[\mathcal{V}_{U|Z}] |Z_i^{1:N}) \displaybreak[0] \\\nonumber
& = |\mathcal{V}_{U|Z}| - H(\widetilde{V}_i^{1:N}[\mathcal{V}_{U|Z}]) + I(\widetilde{V}_i^{1:N}[\mathcal{V}_{U|Z}] ;Z_i^{1:N}) \displaybreak[0] \\ \nonumber
& \stackrel{(b)}{\leq} \delta_N^{(1)} + I(\widetilde{V}_i^{1:N}[\mathcal{V}_{U|Z}] ;Z_i^{1:N}) \displaybreak[0] \\ 
& \stackrel{(c)}{\leq} \delta_N^{(1)} + \delta_N^{(3)}, \label{eqsec2__} 
\end{align}
where $(a)$ holds because ($K_i$, $\widetilde{K}_i$, $R_1$, $F_i$) is a subvector of $\widetilde{V}_i^{1:N}[\mathcal{V}_{U|Z}]$ noting that $\mathcal{V}_{U|X} \subset \mathcal{V}_{U|Z}$ since $U \to X \to Z$, $(b)$ holds by~(\ref{eq_Unif_int}), $(c)$ holds by Lemma~\ref{lem_secre_cra}.

Then, we obtain
\begin{align}
& I(K_{i}\widetilde{K}_{i}; R_1 M_{i} Z_{i}^{1:N}) - I(K_{i}\widetilde{K}_{i}; R_1 F_{i}  Z_{i}^{1:N}) \nonumber \\ \nonumber
& \stackrel{(d)}{=}  I(K_{i}\widetilde{K}_{i}; F_{i}' \oplus \widetilde{K}_{i-1}| R_1 F_{i}  Z_{i}^{1:N}) \\ \nonumber
& \stackrel{(e)}{\leq}  I( R_1 K_{i}\widetilde{K}_{i} F_{i}   F_{i}' Z_{i}^{1:N}; F_{i}' \oplus \widetilde{K}_{i-1} ) \\ \nonumber
& =  H(F_{i}' \oplus \widetilde{K}_{i-1}) - H(F_{i}' \oplus \widetilde{K}_{i-1}|R_1K_{i}\widetilde{K}_{i} F_{i}   F_{i}' Z_{i}^{1:N}) \\ \nonumber
& =  H(F_{i}' \oplus \widetilde{K}_{i-1}) - H( \widetilde{K}_{i-1}| R_1 K_{i}\widetilde{K}_{i} F_{i}   F_{i}' Z_{i}^{1:N}) \\ \nonumber
& \stackrel{(f)}{=}  H(F_{i}' \oplus \widetilde{K}_{i-1}) - H( \widetilde{K}_{i-1}|R_1) \\ \nonumber
& \leq  |\widetilde{K}_{i-1}|  - H( \widetilde{K}_{i-1}|R_1) \\ \nonumber
& =  |\widetilde{K}_{i-1}| - H( \widetilde{K}_{i-1}) + I( \widetilde{K}_{i-1}; R_1) \\
&  \stackrel{(g)}{\leq}  \delta_N^{(1)} + \delta_N^{(2)}, \label{eqsec2___} 
\end{align}
where $(d)$ holds by definition of $M_i$, $(e)$ holds by the chain rule and positivity of mutual information, $(f)$ holds because $\widetilde{K}_{i-1} \to  R_1 \to  K_{i}\widetilde{K}_{i} F_{i}   F_{i}' Z_{i}^{1:N}$, $(g)$ holds by Lemma~\ref{lem_U2_Div} and Lemma~\ref{lem_KKtilde2}.
Finally, we conclude by combining (\ref{eqsec2__}) and (\ref{eqsec2___}).
\end{proof}

We now state a lemma that will be used to show that secrecy holds for the global scheme.
\begin{lem} \label{lem_secrec2}
For $i\in \llbracket 2,k \rrbracket$, define
\begin{align*}
\widetilde{L}_e^{1:i} &\triangleq I \left(K_{1:i} \widetilde{K}_i; R_{1} M_{1:i} Z^{1:N}_{1:i} \right).
\end{align*}
We have 
\begin{multline*}
\widetilde{L}_e^{1:i} - \widetilde{L}_e^{1:i-1} \leq  I \left(K_{i} \widetilde{K}_i; R_1 M_{i}  Z^{1:N}_{i}  \right) \\+ \sum_{j=1}^{i-1} I\left(K_{j} ;  R_1 \right) +  I\left(K_{i-1} ; \widetilde{K}_{i-1}  R_1 \right).
\end{multline*}
\end{lem}

\begin{proof}
See Appendix \ref{App_lem_secrec2}.
\end{proof}

We thus obtain
\begin{align}
&\textbf{L}(\mathcal{S}_N) \nonumber \\\displaybreak[0]
& = {I} (K_{1:k}; M_{1:k} Z^{1:N}_{1:k}) \nonumber \\\displaybreak[0]  \nonumber
& \leq \widetilde{L}_e^{1:k} \\\displaybreak[0] \nonumber
& = \sum_{i=2}^k (\widetilde{L}_e^{1:i} - \widetilde{L}_e^{1:i-1}) + \widetilde{L}_e^{1} \\\displaybreak[0] \nonumber
& \stackrel{(a)}{\leq} \sum_{i=2}^k  \left( I \left(K_{i} \widetilde{K}_i; R_1 M_{i}  Z^{1:N}_{i}  \right) \right. \\  \nonumber
& \phantom{---}\left.+ \sum_{j=1}^{i-1} I\left(K_{j} ;  R_1 \right) +  I\left(K_{i-1} ; \widetilde{K}_{i-1}  R_1 \right) \right) + \widetilde{L}_e^{1} \\\displaybreak[0] \nonumber
& \stackrel{(b)}{\leq} \sum_{i=2}^k  \left( I \left(K_{i} \widetilde{K}_i; R_1 M_{i}  Z^{1:N}_{i}  \right) + i \delta_N^{(2)} \right) + \widetilde{L}_e^{1} \\\displaybreak[0]  \nonumber
& = \frac{(k-1)(k+2)}{2}\delta_N^{(2)} + \widetilde{L}_e^{1} + \sum_{i=2}^k  I \left(K_{i} \widetilde{K}_i; R_1 M_{i}  Z^{1:N}_{i}  \right) \\\displaybreak[0]
& \stackrel{(c)}{\leq} \frac{(k-1)(k+2)}{2}\delta_N^{(2)} + k(2\delta_N^{(1)} + \delta_N^{(2)} + \delta_N^{(3)}) \label{eq_leakage2}
\end{align}
where $(a)$ follows from Lemma~\ref{lem_secrec2}, $(b)$ holds by Lemma~\ref{lem_KKtilde2}, $(c)$ holds by Lemma~\ref{lem_sec_block2}.
\subsubsection{Seed rate}
The seed rate required to initialize the coding scheme is
$$ \lim_{k\to \infty}\lim_{N\to \infty} \frac{|(\mathcal{H}_{U|Y} \backslash {\mathcal{V}}_{U|X}) \backslash  \mathcal{V}_{U|Z} |} {kN}  \leq \lim_{k\to \infty} \frac{H(U|Y)}{k}=0.$$
Note that the seed rate could be chosen to decrease exponentially fast to zero with $N$, since we may choose $k = 2^{N^{\alpha}}$, $\alpha < \beta$, and still have 
$\lim_{N \to \infty}\mathbf{P}_e(\mathcal{S}_N) = 0$ by~(\ref{eq_errorPr2}), $\lim_{N \to \infty}\mathbf{U}_e(\mathcal{S}_N) = 0$ by~(\ref{eq_uniformity2}), and $\lim_{N \to \infty}\mathbf{L}_e(\mathcal{S}_N) = 0$ by~(\ref{eq_leakage2}) along with (\ref{delta_1_def}), (\ref{delta_2_def}), (\ref{delta_3_def}).
\section{Model 3: A Multiterminal Broadcast Model} \label{sec_model3}
In this section, we develop a polar coding scheme for a multiterminal broadcast model. Sections~\ref{Secstatmod3} to~\ref{SecproofTh3} analyze a model with an arbitrary number of terminals but specific source statistics. The extension of the model to general sources is discussed in Section~\ref{SecExt} for the case of three terminals.

\subsection{Secret-key generation model} \label{Secstatmod3}
\begin{figure}
\centering
  \includegraphics[width=8.5cm]{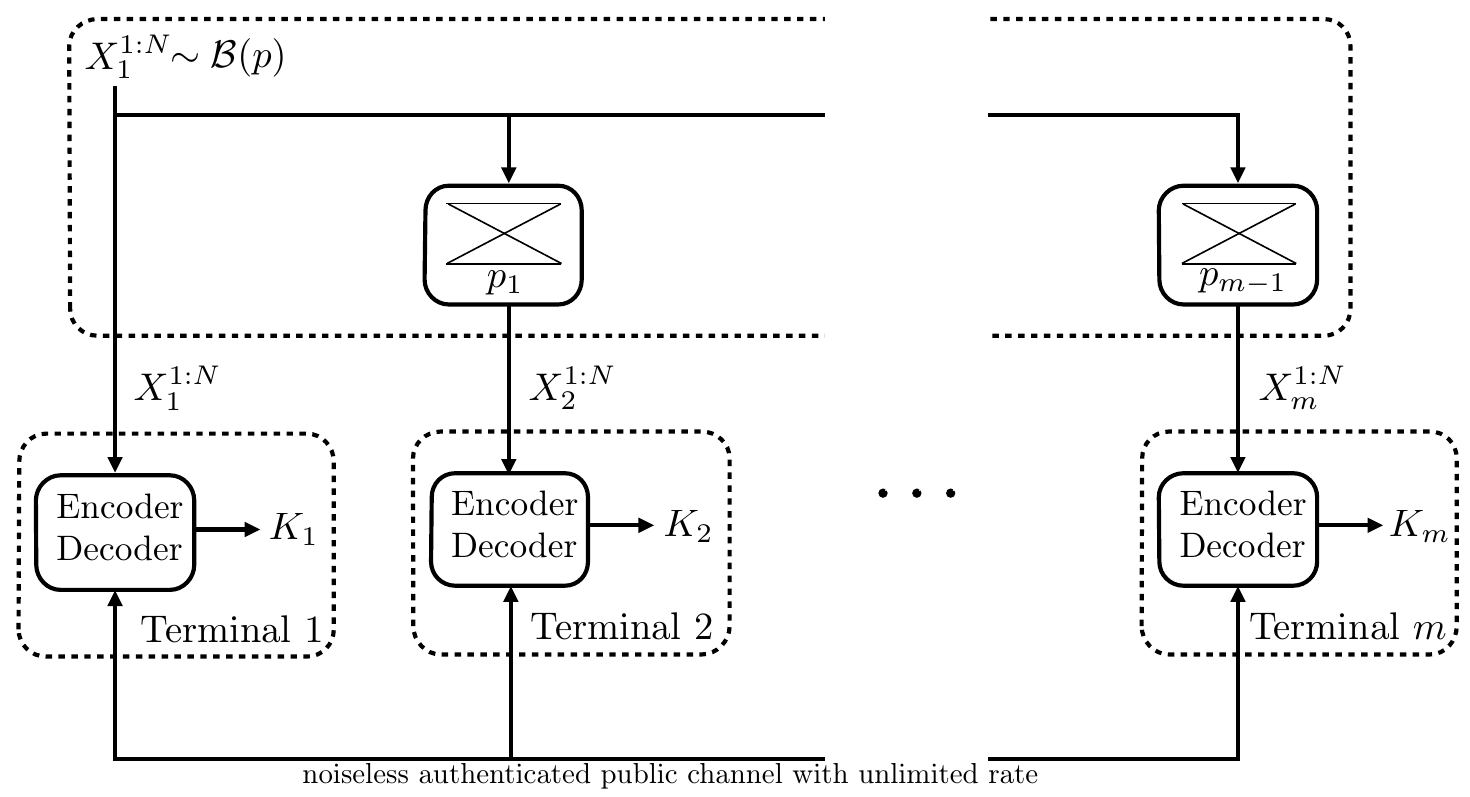}
  \caption{Model 3: Secret-key generation for the $1$-to-$m$ broadcast model}
  \label{figmodel3}
\end{figure}
As illustrated in Fig.~\ref{figmodel3}, we assume that every Terminal $i \in \mathcal{M} \backslash \{ 1\}$ observes a degraded version of the observation of Terminal $1$. For $i \in \mathcal{M}$, we assume that $\mathcal{X}_i = \{ 0,1\}$ and for $i \in \mathcal{M} \backslash \{ 1\}$, we set $X_i = X_1 \oplus B_i$, with $X_1\sim \mathcal{B}(p)$ and $B_i\sim \mathcal{B}(p_{i-1})$, $p_{i-1} \in [0,1]$, independent of $X_1$. Furthermore, we suppose that the eavesdropper does not have access to an observation of the source. We call this setup the $1$-to-$m$ broadcast model, and we recall expression of the secret-key capacity in the next proposition.

\begin{prop}[\!\! \cite{Csiszar00}] \label{prop_broadcast}
Consider the $1$-to-$m$ broadcast model. The secret-key capacity $C_{\text{SK}} (+\infty)$ is given by  
$$C_{\text{SK}} (+\infty)= \min_{i \in \mathcal{M} \backslash \{ 1 \}  } I(X_1;X_i).$$
\end{prop}

\subsection{Polar coding scheme} \label{Sec_scheme3}
Define $i_{\textup{min}} \triangleq \text{argmin}_{i \in  \mathcal{M}\backslash\{1\} } I(X_1;X_i)$ such that $$i_{\textup{min}} - 1 = \argmax_{i \in  \mathcal{M}\backslash\{m\} } p_i.$$ Let $n \in \mathbb{N}$ and $N \triangleq 2^n$. We set $U^{1:N} \triangleq X_{1}^{1:N} G_N$. For $\delta_N \triangleq 2^{-N^{\beta}}$, where $\beta \in ]0,1/2[$, define for $j\in \mathcal{M}\backslash\{1\}$ the sets %
\begin{align*}
\mathcal{H}_{X_1|X_j} & \triangleq \left\{ i\in \llbracket 1,N \rrbracket : H \left( {U}^i | {U}^{1:i-1}X_j^{1:N} \right) \geq \delta_N \right\}.
\end{align*}
We also define the sets
\begin{align*}
\mathcal{V}_{X_1\phantom{'}} & \triangleq \left\{ i\in \llbracket 1,N \rrbracket : H \left( {U}^i | {U}^{1:i-1} \right) \geq 1- \delta_N \right\},\\
\mathcal{H}_{X_1} & \triangleq \left\{ i\in \llbracket 1,N \rrbracket : H \left( {U}^i | {U}^{1:i-1} \right) \geq \delta_N \right\}.
\end{align*}
The encoding and decoding algorithms are given in Algorithm~\ref{alg:encoding_3} and Algorithm~\ref{alg:decoding_3}, respectively. The high-level principle behind the operation of the algorithm is the following. The set $\mathcal{H}_{X_1|X_i}$  contains the indices such that $U^{1:N}[\mathcal{H}_{X_1|X_i}]$ allows Terminal $i \in \mathcal{M}\backslash \{1\}$ to near losslessly reconstruct $U^{1:N}$ from $X_i^{1:N}$ by~\cite{Arikan10}. Using a universality argument formalized in Lemma~\ref{lemincl}, we will show that it is actually sufficient to transmit $U^{1:N}[\mathcal{H}_{X_1|X_{i_{\textup{min}}}}]$ to allow all the terminals to near losslessly reconstruct $U^{1:N}$. The secret key common to all terminals may then be chosen as a subset of $U^{1:N}[\mathcal{V}_{X_1}]$; since $U^{1:N}[\mathcal{H}_{X_1|X_{i_{\textup{min}}}}]$ has been publicly transmitted, the secret-key is chosen as ${U}^{1:N}[{\mathcal{V}}_{X_1} \backslash \mathcal{H}_{X_1|X_{i_{\textup{min}}}}]$. In general, $\mathcal{H}_{X_1|X_{i_{\textup{min}}}} \not\subset {\mathcal{V}}_{X_1}$, and the public communication may leak some information about the key; consequently, as in Model~1 and Model 2, the protocol requires a pre-shared seed to protect the transmission of $U^{1:N}\mathcal[{H}_{X_1|X_{i_{\textup{min}}}} \backslash \mathcal{V}_{X_1} ]$. 
\begin{algorithm}[]
  \caption{Encoding algorithm for Terminal $1$ in Model 3}
  \label{alg:encoding_3}
  \begin{algorithmic} [1]   
    \REQUIRE $\widetilde{K}$, a secret key of size $|\mathcal{H}_{X_1|X_{i_{\textup{min}}}} \backslash \mathcal{V}_{X_1}  |$ shared by all terminals beforehand; the observations $X_1^{1:N}$ from the source.
    \STATE $U^{1:N} \leftarrow X_{1}^{1:N} G_N$
    \STATE $ K \leftarrow {U}^{1:N}[{\mathcal{V}}_{X_1} \backslash \mathcal{H}_{X_1|X_{i_{\textup{min}}}}]$
    \STATE $F \triangleq U^{1:N}[ \mathcal{V}_{X_1} \cap \mathcal{H}_{X_1|X_{i_{\textup{min}}}}]$
    \STATE $F' \triangleq U^{1:N}[  \mathcal{H}_{X_1|X_{i_{\textup{min}}}} \backslash \mathcal{V}_{X_1} ]$
    \STATE Transmit $M\leftarrow[F,F'\oplus \widetilde{K}]$ publicly to Terminals $\{\mathcal{X}_j\}_{j \in \mathcal{M} \backslash{1}}$
    \RETURN $K$
  \end{algorithmic}
\end{algorithm}

\begin{algorithm}[]
  \caption{Decoding algorithm for Terminal $j\in\mathcal{M}\setminus\{1\}$ for Model 3}
  \label{alg:decoding_3}
  \begin{algorithmic}  [1]  
    \REQUIRE $\widetilde{K}$, a secret key of size $|\mathcal{H}_{X_1|X_{i_{\textup{min}}}} \backslash \mathcal{V}_{X_1}  |$ shared by all terminals beforehand; the observations $X_j^{1:N}$ from the source, the message $M$ transmitted by Terminal 1.
    \STATE Form $\widehat{U}^{1:N}$ from $M$ and $\widetilde{K}$ using the successive cancellation decoder of~\cite{Arikan10}.
    \STATE $\widehat{K}\leftarrow \widehat{U}^{1:N}[{\mathcal{V}}_{X_1} \backslash \mathcal{H}_{X_1|X_{i_{\textup{min}}}}]$
    \RETURN $\widehat{K}$
  \end{algorithmic}
\end{algorithm}

As shown in Section~\ref{SecproofTh3}, we have the following result.

\begin{thm} \label{ThBroadcast}
Consider the $1$-to-$m$ broadcast model of Section~\ref{Secstatmod3}. Assume that all terminals share a seed, whose rate can be chosen in $o(N)$. The secret-key capacity $C_{\text{SK}}(+\infty)$ given in Proposition~\ref{prop_broadcast} is achieved by the polar coding scheme in Algorithm~\ref{alg:encoding_3} and Algorithm~\ref{alg:decoding_3}, whose computational complexity is $O(N \log N)$. 
\end{thm}

\begin{proof}
See Section \ref{SecproofTh3}.
\end{proof}

The following corollary shows that no seed is required when the source has uniform marginals.
\begin{cor} \label{PropBroadcast}
Consider the $1$-to-$m$ broadcast model. Assume that the source has uniform marginal, that is, $X_1 \sim \mathcal{B}(1/2)$. The secret-key capacity $C_{\text{SK}}(+\infty)$ given in Proposition~\ref{prop_broadcast} is achievable with perfect secrecy with the polar coding scheme of Algorithm~\ref{alg:encoding_3} and Algorithm~\ref{alg:decoding_3} choosing $F' = \emptyset$ and replacing the set $\mathcal{V}_{X_1}$ by $\mathcal{H}_{X_1}$ wherever it appears.
\end{cor}
We omit the proof of Corollary~\ref{PropBroadcast}, which is similar to the ones of Theorem \ref{ThBroadcast} and Proposition \ref{ex3}. Note that the model studied in Corollary~\ref{PropBroadcast} is a particular case of \cite[Model 3]{Ye12}. However, the construction proposed in \cite[Model 3]{Ye12} relies again on a standard array, whose size grows exponentially with the blocklength. Note also that a polar coding scheme is proposed in Section \ref{sec_model4} for \cite[Model 3]{Ye12}.

\subsection{Analysis of polar coding scheme: Proof of Theorem \ref{ThBroadcast}} \label{SecproofTh3}

\subsubsection{Key rate}
Similarly to the proof of Theorem~\ref{Th2}, we can show that the key rate is
$$\lim_{N \rightarrow + \infty}  \frac{|\mathcal{V}_{X_1} \backslash \mathcal{H}_{X_1|X_{i_{\textup{min}}}}|}{N} =I(X_1;X_{i_{\textup{min}}}).$$

\subsubsection{Seed rate}
Similarly to the proof of Theorem~\ref{Th2}, we can show that the seed rate is
$$\lim_{N \rightarrow + \infty}\frac{ | \mathcal{H}_{X_1|X_{i_{\textup{min}}}} \backslash \mathcal{V}_{X_1}|}{N} = 0.$$

\subsubsection{Reliability}
We make use of the following lemma.
\begin{lem} \label{lemincl}
For $j \in \mathcal{M} \backslash \{ 1,i_{\textup{min}}\}$, we have $\mathcal{H}_{X_1|{X}_j} \subset \mathcal{H}_{X_1|X_{i_{\textup{min}}}}$.
\end{lem}
\begin{proof}
Let $j \in \mathcal{M} \backslash \{ 1,i_{\textup{min}}\}$. We define $\tilde{B}_{i_{\textup{min}}}^{(j)} \triangleq B_j + \Delta_j$, with $\Delta_j$ independent  of $B_j$ and such that $p_{\tilde{B}_{i_{\textup{min}}}^{(j)}} = p_{B_{i_{\textup{min}}}} $. We set $\tilde{X}_{i_{\textup{min}}}^{(j)} \triangleq X_1 + \tilde{B}_{i_{\textup{min}}}^{(j)} $. Hence, since $B_{i_{\textup{min}}} \sim \mathcal{B} (p_{i_{\textup{min}-1}})$, we have for any $x,y \in \{0 ,1 \}$,
\begin{align*}
&p_{\tilde{X}^{(j)}_{i_{\textup{min}}} | X_1} (x|y) \\
& = (1 -\mathds{1}\{ x= y \}) p_{i_{\textup{min}-1}} + \mathds{1}\{ x= y \} (1- p_{i_{\textup{min}-1}})  \\
& = p_{X_{i_{\textup{min}}} | X_1} (x|y),
\end{align*}
that is, $p_{X_1\tilde{X}^{(j)}_{i_{\textup{min}}}} = p_{X_1X_{i_{\textup{min}}}}$. 
We now define the sets
\begin{align*}
\mathcal{H}_{X_1|\tilde{X}_{i_{\textup{min}}}^{(j)}} & \triangleq \left\{ i\in \llbracket 1,N \rrbracket : H \!\left( \!{U}_i | {U}^{i-1} \left( \tilde{X}_{i_{\textup{min}}}^{(j)} \right)^{1:N} \right) \! \geq \delta_N \! \right\}.
\end{align*}
By the data processing equality, we have  $\mathcal{H}_{X_1|{X}_j} \subset \mathcal{H}_{X_1|\tilde{X}^{(j)}_{i_{\textup{min}}}}$ but we also have $\mathcal{H}_{X_1|\tilde{X}^{(j)}_{i_{\textup{min}}}} = \mathcal{H}_{X_1|X_{i_{\textup{min}}}}$ since $p_{X_1\tilde{X}^{(j)}_{i_{\textup{min}}}} = p_{X_1X_{i_{\textup{min}}}}$, whence $\mathcal{H}_{X_1|{X}_j} \subset \mathcal{H}_{X_1|X_{i_{\textup{min}}}}$.
\end{proof}  

By \cite[Theorem 3]{Arikan10} and by Lemma \ref{lemincl}, for $j \in \mathcal{M} \backslash \{ 1\}$, Terminal $j$ can reconstruct $K$ from $[F,F'] = U^{N}[\mathcal{H}_{X_1|X_{i_{\textup{min}}}}] \supset U^{N}[\mathcal{H}_{X_1|{X}_j}]$ with error probability $\mathbf{P}_e(\mathcal{S}_N) \leq N\delta_N$.

\subsubsection{Strong secrecy and key uniformity}

Secrecy and uniformity hold since, 
\begin{align*}
&\mathbf{L}(\mathcal{S}_N) + \mathbf{U}(\mathcal{S}_N) \\
&=  I\left(K;F\right) +  \log|\mathcal{K}| - H(K) \\
&= |K| - H\left(K|F\right) \\
&= |K| - H\left(K F\right) + H(F) \\
& \leq |F| + |K| - H\left(K F\right)\\
& = |\mathcal{V}_{X_1} \cap \mathcal{H}_{X_1|X_{i_{\textup{min}}}}| + |\mathcal{V}_{X_1} \backslash \mathcal{H}_{X_1|X_{i_{\textup{min}}}}| - H(U^{1:N}[\mathcal{V}_{X_1}]) \\
& = |\mathcal{V}_{X_1}| - H(U^{1:N}[\mathcal{V}_{X_1}]) \\
& \leq N \delta_N,
\end{align*}
where the last inequality can be shown as in the proof of Theorem~\ref{Th2}.

\subsection{An extension to general sources} \label{SecExt}

The multiterminal model described in Section~\ref{Secstatmod3} only considers binary symmetric channels between the components of the source. A natural question is whether a similar coding scheme may be developed for general sources. We answer this by the affirmative for the case of three terminals; however, the coding scheme is significantly more involved than the one in Section \ref{Sec_scheme3}.  In the following, we can assume without loss of generality that 
$$
I(X_1;X_2) = \displaystyle\max_{ j \in \{ 1,2,3 \}} \min_{i \in \{ 1,2,3 \} \backslash \{ j\} } I(X_j;X_i). 
$$
Let $n \in \mathbb{N}$ and $N \triangleq 2^n$. We note $U^{1:N} \triangleq {X_2}^{1:N} G_N$, and for $\delta_N \triangleq 2^{-N^{\beta}}$, where $\beta \in ]0,1/2[$, we define the following sets
\begin{align*}
\mathcal{V}_{X_2}  \triangleq &  \left\{ i\in \llbracket 1,N\rrbracket : {H} \left(U^i|U^{1:i-1} \right) \geq 1- \delta_N \right\} , \\
\mathcal{H}_{X_2|X_1} \triangleq & \left\{ i\in \llbracket 1,N\rrbracket : {H}\left(U^i|U^{1:i-1} X_1^{1:N}\right) \geq \delta_N \right\}, \\
\mathcal{H}_{X_2|X_3} \triangleq & \left\{ i\in \llbracket 1,N\rrbracket : {H}\left(U^i|U^{1:i-1} X_3^{1:N}\right) \geq \delta_N \right\}.
\end{align*}
We also define 
\begin{align*}
\mathcal{K}_{X_{\mathcal{M}}} &\triangleq (\mathcal{V}_{X_2} \backslash \mathcal{H}_{X_2|X_1})\backslash \mathcal{H}_{X_2|X_3}\\
\bar{\mathcal{K}}_{X_{\mathcal{M}}} & \triangleq (\mathcal{V}_{X_2} \backslash \mathcal{H}_{X_2|X_1}) \cap \mathcal{H}_{X_2|X_3},
\end{align*}
which are such that $\mathcal{V}_{X_2} \backslash \mathcal{H}_{X_2|X_1} = \mathcal{K}_{X_{\mathcal{M}}} \cup \bar{\mathcal{K}}_{X_{\mathcal{M}}} $ and $\mathcal{K}_{X_{\mathcal{M}}} \cap \bar{\mathcal{K}}_{X_{\mathcal{M}}} = \emptyset$. Finally, we define
\begin{align*}
\mathcal{F}_{X_2|X_1} & \triangleq \mathcal{H}_{X_2|X_1} \cap \mathcal{V}_{X_2},\\ 
\bar{\mathcal{F}}_{X_2|X_1} & \triangleq \mathcal{H}_{X_2|X_1}  \backslash \mathcal{V}_{X_2},\\
\mathcal{F}_{X_2|X_3} & \triangleq \mathcal{H}_{X_2|X_3} \cap \mathcal{V}_{X_2},\\ 
\bar{\mathcal{F}}_{X_2|X_3} & \triangleq \mathcal{H}_{X_2|X_3}  \backslash \mathcal{V}_{X_2},
\end{align*}
which are such that $\mathcal{H}_{X_2|X_1} = \mathcal{F}_{X_2|X_1} \cup \bar{\mathcal{F}}_{X_2|X_1}$, $\mathcal{F}_{X_2|X_1} \cap \bar{\mathcal{F}}_{X_2|X_1} = \emptyset$, $\mathcal{H}_{X_2|X_3} = \mathcal{F}_{X_2|X_3} \cup \bar{\mathcal{F}}_{X_2|X_3}$, and $\mathcal{F}_{X_2|X_3} \cap \bar{\mathcal{F}}_{X_2|X_3}= \emptyset$.

The encoding and decoding algorithms are provided in Algorithm~\ref{alg:encoding_3bis}, Algorithm~\ref{alg:decoding_3bis1}, and Algorithm~\ref{alg:decoding_3bis3}. The underlying principle is to make Terminals $1$ and $3$ reconstruct $X_2^{1:N}$ and to choose the secret key as a subset of $U^{1:N}$. For the public communication, we perform universal source coding with side information with an idea similar to~\cite{Ye14}. Terminal $2$ thus performs encoding over $k$ blocks of size $N$ to transmit the side information necessary to reconstruct $X_2^{1:{kN}}$ at Terminals $1$ and $3$. Specifically, Terminal $1$ decodes the blocks in order from $1$ to $k$, so that it is able to estimate $U_i^{1:N}[\mathcal{H}_{X_2|X_1}]$ by processing the observations and the public communication in blocks $1$ to $i$. In contrast, Terminal $3$ decodes the blocks in reverse order starting from $k$ down to $1$, so that it is able to estimate $U_i^{1:N}[\mathcal{H}_{X_2|X_3}]$ by processing the observations and the public communication in blocks $k$ down to $i$. One of the challenges is to extract a uniform key from $U_{1:k}^{1:N}$ independent of the public communication messages, which we address by protecting some of the public communication corresponding to Block $i$ with part of the secret-key extracted in Block $i-1$. Moreover, similarly to Algorithms \ref{alg:encoding_1} and \ref{alg:encoding_2}, a small secret seed must be shared by the users to protect the bits in positions $\mathcal{H}_{X_2|X_1}  \backslash \mathcal{V}_{X_2} \cup \mathcal{H}_{X_2|X_3}  \backslash \mathcal{V}_{X_2}$, which must be revealed to allow reconstruction of the secret-key by Terminals $1$ and $3$, but that may also leak information about the secret key. 

The following remarks clarify why Algorithms \ref{alg:encoding_3bis}, \ref{alg:decoding_3bis1}, and  \ref{alg:decoding_3bis3} achieve the desired behavior.

\begin{rem}
  \label{rem:justif_decoding_1}
  In every block $i$, Terminal 1 observes $M_i=[F_i^{(1)}\oplus \bar{K}_{i-1},F_i^{(2)},F'_i\oplus\widetilde{K}_{i}]$. Using its estimate of the key $\bar{K}_i$ from the previous block, Terminal 1 estimates $[F_i^{(1)},F_i^{(2)},F'_i]$, which contains $U_{i}^{1:N}[\mathcal{H}_{X_2|X_1}]$ by construction. Hence, Terminal 1 has ability to run the successive cancellation decoder and reconstruct ${U}_i^{1:N}$.
\end{rem}

\begin{rem}  \label{rem:justif_decoding_3}
  In Block $k$, Terminal 3 has access to $F_k^{(2)}$, $F'_k$, and $\bar{F}_k$ using $M_k$ and $\widetilde{K}_k$. Since $\mathcal{F}_{X_{\mathcal{M}}} \subset {\mathcal{F}}_{X_2|X_1} \backslash {\mathcal{F}}_{X_2|X_3}$, note that
  \begin{align*}
    {\mathcal{F}}_{X_2|X_1} \backslash \mathcal{F}_{X_{\mathcal{M}}} 
    & = {\mathcal{F}}_{X_2|X_1} \cap \mathcal{F}_{X_{\mathcal{M}}}^c \\
    & \supset {\mathcal{F}}_{X_2|X_1} \cap ({\mathcal{F}}_{X_2|X_1} \backslash {\mathcal{F}}_{X_2|X_3})^c \\
    & = {\mathcal{F}}_{X_2|X_1} \cap {\mathcal{F}}_{X_2|X_3}.
  \end{align*}
Hence $U_k^{1:N}[{\mathcal{F}}_{X_2|X_1} \cap {\mathcal{F}}_{X_2|X_3}]\subset F_k^{(2)}$, which combined with $\bar{F}_k$ and $F'_k$ allows Terminal 3 to obtain $U_k^{1:N}[\mathcal{H}_{X_2|X_3}]$. Hence, Terminal 3 has the ability to run the successive cancellation decoder and reconstruct $\widehat{U}_k^{1:N}$.

For Block $i\in \llbracket k-1,1\rrbracket$, observe that if $\widehat{U}_{i+1}^{1:N} [\mathcal{F}_{X_{\mathcal{M}}}] = {U}_{i+1}^{1:N} [\mathcal{F}_{X_{\mathcal{M}}}]$, then we have 
\begin{align*}
&[F_{i+1}^{(1)}\oplus \bar{K}_{i} \oplus \widehat{U}_{i+1}^{1:N} [\mathcal{F}_{X_{\mathcal{M}}}],F_{i}^{(2)},F_{i}'] \\
& = [U_i^{1:N} [\bar{\mathcal{K}}_{X_{\mathcal{M}}}],F_{i}^{(2)},F_{i}']   \\
& = [U_i^{1:N} [{\mathcal{F}}_{X_2|X_3} \backslash {\mathcal{F}}_{X_2|X_1} ],F_{i}^{(2)},F_{i}']   \\
& \supset [U_i^{1:N} [{\mathcal{F}}_{X_2|X_3} \backslash {\mathcal{F}}_{X_2|X_1} ],U_i^{1:N} [{\mathcal{F}}_{X_2|X_1} \cap {\mathcal{F}}_{X_2|X_3}],F_{i}']\\
& \supset U_i^{1:N} [ \mathcal{H}_{X_2|X_3} ].
\end{align*}
Consequently, Terminal $3$ can form an estimate of $U_{i}^{1:N} [ \mathcal{H}_{X_2|X_3} ]$ with 
$$[F_{i+1}^{(1)} \oplus \bar{K}_{i} \oplus  \widehat{U}_{i+1}^{1:N} [\mathcal{F}_{X_{\mathcal{M}}}],F_{i}^{(2)},F_{i}']$$ and apply the successive cancellation decoder to form $\widehat{U}_i^{1:N}$ an estimate of $U_i^{1:N}$. 

%
\end{rem}

\begin{algorithm}[]
  \caption{Encoding algorithm for Terminal $2$ in Model 3}
  \label{alg:encoding_3bis}
  \begin{algorithmic} [1]   
    \REQUIRE $k$ independent secret keys $\{ \widetilde{K}_i \}_{i \in \llbracket 1,k \rrbracket}$ of size $|\bar{\mathcal{F}}_{X_2|X_1} \cup \bar{\mathcal{F}}_{X_2|X_3}|$ shared by all terminals beforehand; for every block $i\in\llbracket 1,k\rrbracket$, the observations $\left(X_2\right)_i^{1:N}$ from the source. 
    $\mathcal{F}_{X_{\mathcal{M}}}$, a subset of ${\mathcal{F}}_{X_2|X_1} \backslash {\mathcal{F}}_{X_2|X_3}$ with size $|\bar{\mathcal{K}}_{X_{\mathcal{M}}}|$.     \FOR{Block $i=1$ to $k$}
    \IF{$i=1$}
    \STATE $U_1^{1:N} \leftarrow (X_{2})_1^{1:N} G_N$
    \STATE $K_1 \leftarrow U_1^{1:N} [\mathcal{K}_{X_{\mathcal{M}}} ]$
    \STATE $\bar{K}_1 \leftarrow U_1^{1:N} [\bar{\mathcal{K}}_{X_{\mathcal{M}}} ]$
    \STATE $F_1 \leftarrow U_1^{1:N} [ {\mathcal{F}}_{X_2|X_1} ]$
    \STATE $F'_1 \leftarrow U_1^{1:N} [ \bar{\mathcal{F}}_{X_2|X_1} \cup \bar{\mathcal{F}}_{X_2|X_3}]$
    \STATE Transmit $M_1 \leftarrow [F_1, F_1' \oplus \widetilde{K}_1]$ publicly to all Terminals
    \ELSIF{$i=k$}
    \STATE $U_k^{1:N} \leftarrow (X_{2})_k^{1:N} G_N$
    \STATE $K_k \leftarrow  U_k^{1:N} [ \mathcal{K}_{X_{\mathcal{M}}} \cup  \mathcal{F}_{X_{\mathcal{M}}} ]$
    \STATE $F_k^{(1)} \leftarrow U_k^{1:N} [ \mathcal{F}_{X_{\mathcal{M}}} ] $
    \STATE $F_k^{(2)} \leftarrow U_k^{1:N} [\mathcal{F}_{X_2|X_1}\backslash \mathcal{F}_{X_{\mathcal{M}}} ]$
    \STATE $F'_k \leftarrow U_k^{1:N} [  \bar{\mathcal{F}}_{X_2|X_1} \cup \bar{\mathcal{F}}_{X_2|X_3} ]$ 
    \STATE $\bar{F}_k \leftarrow U_k^{1:N} [{\mathcal{F}}_{X_2|X_3} \backslash {\mathcal{F}}_{X_2|X_1}]$
    \STATE Transmit $M_k \leftarrow [F_k^{(1)}\oplus \bar{K}_{k-1},F_k^{(2)}, F_k' \oplus \widetilde{K}_{k}, \bar{F}_k]$ publicly to all Terminals
    \ELSE
    \STATE $U_i^{1:N} \leftarrow (X_{2})_i^{1:N} G_N$
    \STATE $K_i \leftarrow  U_i^{1:N} [ \mathcal{K}_{X_{\mathcal{M}}}  \cup \mathcal{F}_{X_{\mathcal{M}}}]$
    \STATE $\bar{K}_i \leftarrow U_i^{1:N} [\bar{\mathcal{K}}_{X_{\mathcal{M}}} ]  $
    \STATE $F_i^{(1)} \leftarrow U_i^{1:N} [ \mathcal{F}_{X_{\mathcal{M}}} ]$
    \STATE  $F_i^{(2)} \leftarrow U_i^{1:N} [\mathcal{F}_{X_2|X_1}\backslash \mathcal{F}_{X_{\mathcal{M}}} ]$
    \STATE $F'_i \leftarrow U_i^{1:N} [  \bar{\mathcal{F}}_{X_2|X_1} \cup \bar{\mathcal{F}}_{X_2|X_3} ]$
    \STATE Transmit $M_i \leftarrow [F_i^{(1)} \oplus \bar{K}_{i-1},F_i^{(2)}, F_i' \oplus \widetilde{K}_{i}]$ publicly to all Terminals
    \ENDIF
    \ENDFOR
    \RETURN $K_{1:k} \leftarrow [ K_{1}, K_2, \ldots, K_k ]$.
  \end{algorithmic}
\end{algorithm}

\begin{algorithm}[]
  \caption{Decoding algorithm for Terminal $1$ in Model 3}
  \label{alg:decoding_3bis1}
  \begin{algorithmic} [1]   
    \REQUIRE Secret keys $\{ \widetilde{K}_i \}_{i \in \llbracket 1,k \rrbracket}$ of size $|\bar{\mathcal{F}}_{X_2|X_1} \cup \bar{\mathcal{F}}_{X_2|X_3}|$ shared with Terminal 2; for every block $i\in\llbracket 1,k\rrbracket$, the observations $\left(X_1\right)_i^{1:N}$ from the source; the set $\mathcal{F}_{X_{\mathcal{M}}}$ defined in Algorithm \ref{alg:encoding_3bis}.
    \FOR{Block $i=1$ to $k$}
    \IF{$i=1$}
    \STATE Form $[F_1,F_1']$ from $M_1$ and $\widetilde{K}_1$ and extract an estimate of $U_1^{1:N}[\mathcal{H}_{X_2|X_1}]$ \COMMENT{See Remark~\ref{rem:justif_decoding_1} for a justification}
    \STATE Form $\widehat{U}_1^{1:N}$ with the successive cancellation decoder of \cite{Arikan10}
    \STATE $\widehat{K}_1\leftarrow \widehat{U}_1^{1:N} [\mathcal{K}_{X_{\mathcal{M}}} ]$
    \ELSE
    \STATE Estimate $[F_i^{(1)},F_i^{(2)},F_i']$ from $M_i$, $\widehat{U}_{i-1}^{1:N}$, and ${\widetilde{K}}_{i}$ and extract an estimate of $U_i^{1:N} [ \mathcal{H}_{X_2|X_1} ]$
    \STATE Form $\widehat{U}_i^{1:N}$ with the successive cancellation decoder of \cite{Arikan10}
    \STATE $\widehat{K}_i\leftarrow \widehat{U}_i^{1:N} [\mathcal{K}_{X_{\mathcal{M}}} ]$
    \ENDIF
    \ENDFOR
    \RETURN $\widehat{K}_{1:k} \leftarrow [ \widehat{K}_{1}, \widehat{K}_2, \ldots, \widehat{K}_k ]$.
  \end{algorithmic}
\end{algorithm}

\begin{algorithm}[]
  \caption{Decoding algorithm for Terminal $3$ in Model 3}
  \label{alg:decoding_3bis3}
  \begin{algorithmic} [1]  
  \REQUIRE Secret keys $\{ \widetilde{K}_i \}_{i \in \llbracket 1,k \rrbracket}$ of size $|\bar{\mathcal{F}}_{X_2|X_1} \cup \bar{\mathcal{F}}_{X_2|X_3}|$ shared with Terminal 2; for every block $i\in\llbracket 1,k\rrbracket$, the observations $\left(X_3\right)_i^{1:N}$ from the source; $\mathcal{F}_{X_{\mathcal{M}}}$ used in Algorithm \ref{alg:encoding_3bis}.
    \FOR{Block $i=k$ to $1$}
    \IF{$i=k$}
    \STATE Form $[F_k^{(2)}, F'_k, \bar{F}_k]$ from $M_k$ and $\widetilde{K}_k$ and extract and estimate $U_k^{1:N} [ \mathcal{H}_{X_2|X_3} ]$ \COMMENT{See Remark~\ref{rem:justif_decoding_3} for a justification}
    \STATE Form $\widehat{U}_k^{1:N}$ with the successive cancellation decoder of \cite{Arikan10}
    \STATE $\widehat{K}_1\leftarrow \widehat{U}_1^{1:N} [\mathcal{K}_{X_{\mathcal{M}}} ]$
    \ELSE
    \STATE Estimate $[\bar{K}_i ,F_i^{(2)},F_i']$ from $M_i$, $\widehat{U}_{i+1}^{1:N}$, and ${\widetilde{K}}_{i}$ and extract an estimate of $U_i^{1:N} [ \mathcal{H}_{X_2|X_3} ]$
    \STATE Form $\widehat{U}_i^{1:N}$ with the successive cancellation decoder of \cite{Arikan10}
    \STATE $\widehat{K}_i\leftarrow \widehat{U}_i^{1:N} [\mathcal{K}_{X_{\mathcal{M}}} ]$
    \ENDIF
    \ENDFOR
    \RETURN $\widehat{K}_{1:k} \leftarrow [ \widehat{K}_{1}, \widehat{K}_2, \ldots, \widehat{K}_k ]$. 
  \end{algorithmic}
\end{algorithm}

\begin{thm} \label{th3ter}
Assume the general setting of Section~\ref{SecStatemetn} with $m=3$, $\mathcal{X}_1 = \mathcal{X}_2 =\mathcal{X}_3 = \{ 0,1\}$, rate-unlimited public communication, i.e., $R_p = + \infty$, and $Z = \emptyset$, i.e., the eavesdropper does not have access to the observation of the source component $Z$. Assume that all terminals share a seed, whose rate can be chosen in $o(N)$. The secret-key rate 
$$
\displaystyle\max_{ j \in \{ 1,2,3 \}} \min_{i \in \{ 1,2,3 \} \backslash \{ j\} } I(X_j;X_i)
$$
is achieved by the polar coding scheme of Algorithm~\ref{alg:encoding_3bis} and Algorithms~\ref{alg:decoding_3bis1}, \ref{alg:decoding_3bis3}, which involves a chaining of $k$ blocks of size $N$, and whose complexity is $O(kN \log N)$. 
\end{thm}

\begin{proof}
See Appendix \ref{App_Th6}.
\end{proof}

As a corollary we obtain the following result for a broadcast model with three terminals.
\begin{cor}
Assume the broadcast setting of Section~\ref{Secstatmod3} with $m=3$, $\mathcal{X}_1 = \mathcal{X}_2 =\mathcal{X}_3 = \{ 0,1\}$, and an arbitrary distribution $p_{X_{\mathcal{M}}}$. Assume that all terminals share a seed, whose rate can be chosen in $o(N)$. The secret-key key capacity $C_s(+\infty) = \min (I(X_1;X_2), I(X_2;X_3))$ is achieved by the polar coding scheme of Algorithm~\ref{alg:encoding_3bis} and Algorithms~\ref{alg:decoding_3bis1}, \ref{alg:decoding_3bis3}, which involves a chaining of $k$ blocks of size $N$, and whose complexity is $O(kN \log N)$. 
\end{cor}

%


\section{Model 4: Multiterminal Markov Tree Model with Uniform Marginals} 
\label{sec_model4}

\subsection{Secret-key generation model} \label{Secstatmod4}
The model for which we now develop a polar coding scheme was first introduced in~\cite[Model 3]{Ye12}. We assume that all the observation alphabets are $\mathcal{X}_i = \{ 0,1\}$ for $i \in \mathcal{M}$. As illustrated in Fig.~\ref{figmodel4}, consider a tree $\mathcal{T}$ with vertex set $\mathcal{V}(\mathcal{T}) \triangleq \mathcal{M}$ and edge set $\mathcal{E}(\mathcal{T})$. The joint probability distribution $p_{X_{\mathcal{M}}}$ is characterized as follows.
$\forall (i,j) \in \mathcal{E}(\mathcal{T}), \forall x_i, x_j \in \{0,1\},$ 
\begin{multline*} p_{X_iX_j} (x_i,x_j) \\ \triangleq \frac{1}{2}(1-p_{i,j}) \mathds{1} \{ x_i = x_j \} + \frac{1}{2}p_{i,j}(1-\mathds{1} \{ x_i = x_j \}),
\end{multline*}
which means that $p_{X_i} = p_{X_j}$ is uniform and the test channel between $X_i$ and $X_j$ is a binary symmetric channel with parameter $p_{i,j}$.

Furthermore, we suppose that the eavesdropper does not have access to the observation of the source component $Z$. This setup is called the Markov tree model with uniform marginals. The expression of the secret-key capacity is recalled in the following proposition.
\begin{prop}[\!\! \cite{Csiszar00}] \label{prop_line}
Consider the Markov tree model with uniform marginal. The secret-key capacity $C_{\text{SK}} (+\infty)$ is given by  
$$C_{\text{SK}} (+\infty)=I(X_{n_0};X_{n_1}),$$
where $(n_0 , n_1) \triangleq \text{argmin}_{ (i,j) \in \mathcal{E}(\mathcal{T}) } I(X_i;X_{j})$.
\end{prop}

\subsection{Polar coding scheme} \label{Sec_scheme4}
We first introduce some notation for the coding scheme. For any $i \in \mathcal{M}$, we note $\mathcal{N}^{j}(i)$ the set of vertices in $\mathcal{V}(\mathcal{T})$ that are at distance $j$ from vertex $i$. We note $(n_0 , n_1) \triangleq \text{argmin}_{ (i,j) \in \mathcal{E}(\mathcal{T}) } I(X_i;X_{j})$. We also consider for the encoding process the tree $\mathcal{T}$ as a rooted tree with root $X_{n_0}$. An example is depicted in Figure~\ref{figmodel4}.

\begin{figure}
\centering
  \includegraphics[width=8.5cm]{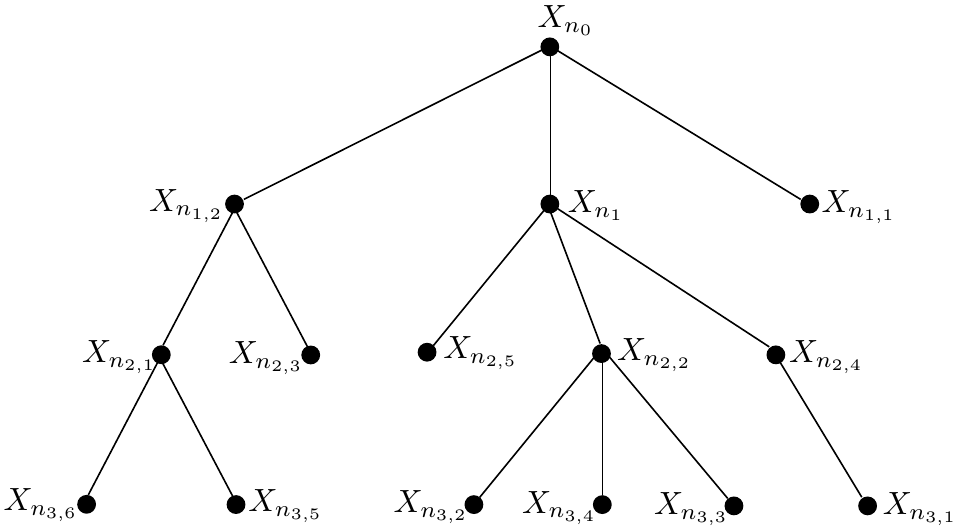}
  \caption{Example of Markov tree model with uniform marginal for $m=15$. Each vertex represent the random variable observed by a given terminal, and each edge can be seen as a binary symmetric test channel. We have noted $(n_0 , n_1) \triangleq \text{argmin}_{ (i,j) \in \mathcal{E}(\mathcal{T}) } I(X_i;X_{j})$, $\mathcal{N}^1(n_0) \triangleq \{ n_1 ,n_{1,1},n_{1,2} \}$, $\mathcal{N}^2(n_0) \triangleq \{ n_{2,i} \}_{i\in \llbracket 1,5 \rrbracket}$, $\mathcal{N}^3(n_0) \triangleq \{ n_{3,i} \}_{i\in \llbracket 1,6 \rrbracket}$ }
  \label{figmodel4}
\end{figure}

Let $n \in \mathbb{N}$ and $N \triangleq 2^n$. For $j \in \mathcal{M}$, we set $U_{j}^{1:N} \triangleq X_{j}^{1:N} G_N$. For $j_1 \in \mathcal{M}$, $j_2 \in \mathcal{M} \backslash \{ j_1\}$, and $\delta_N \triangleq 2^{-N^{\beta}}$, $\beta \in ]0,1/2[$, we define the sets
\begin{align*}
\mathcal{H}_{X_{j_1}|X_{j_2}} & \triangleq \left\{ i\in \llbracket 1,N \rrbracket : H \left( {U}_{j_1}^i | U_{j_1}^{1:i-1}X_{j_2}^{1:N} \right) \geq \delta_N \right\}.
\end{align*}
The exact encoding and decoding algorithms are given in Algorithm \ref{alg:encoding_4} and Algorithm \ref{alg:decoding_4}. The principle of their operation is to have all terminal reconstruct $U_{n_0}^{1:N}$ and choose the key as a subvector of $U_{n_0}^{1:N}$. The idea behind the inter-terminal communication, which is illustrated in Figure~\ref{figmodel4ex}, is to take advantage of the tree structure to make all Terminals reconstruct $X_{n_0}^{1:N}$; the source uniformity plays a crucial role to develop a universal result in Lemma~\ref{lemincl2}, similar to the one obtained for the broadcast model in Lemma~\ref{lemincl}. Although the assumption of uniform marginal is required in our proof, a side benefit is that no pre-shared seed is needed to ensure strong secrecy. 

\begin{algorithm}[]
  \caption{Encoding algorithm for Model 4}
  \label{alg:encoding_4}
  \begin{algorithmic} [1]   
    \STATE $F_{n_0} \leftarrow U_{n_0}^{1:N}  \left[ \mathcal{H}_{X_{n_0}|X_{n_1}} \right]$.  
    \STATE Terminal $n_0$ transmits $F_{n_0}$ publicly.
    \STATE Define $d$ as the maximal distance between the vertex $n_0$ and the vertices in $\mathcal{V}(\mathcal{T})$.
    \FOR{$i=1$ to $d-1$}
    \FOR{$j \in \mathcal{N}^i(n_0)$}
    \IF{$\mathcal{N}^1(j) \cap \mathcal{N}^{i+1}(n_0) \neq \emptyset$}
    \STATE Define $j^* \triangleq \displaystyle\argmax_{ \tilde{j} \in \mathcal{N}^1(j) \cap \mathcal{N}^{i+1}(n_0)} p_{\tilde{j},j}.$
    \STATE $F_{i,j} \leftarrow U_{j}^{1:N} \left[ \mathcal{H}_{X_{j}| X_{j^*}} \right],$
    \STATE Terminal $j$ transmits $F_{i,j}$ publicly
    \ENDIF
    \ENDFOR
    \ENDFOR
    \RETURN $K\leftarrow  U_{n_0}^{1:N}  \left[ \mathcal{H}^c_{X_{n_0}|X_{n_1}} \right]$
  \end{algorithmic}
\end{algorithm}

\begin{algorithm}[]
  \caption{Decoding algorithm for Model 4}
  \label{alg:decoding_4}
  \begin{algorithmic}  [1]  
    \REQUIRE Observations from the source, and public messages $\mathbf{F}$.
    \STATE With $F_{n_0}=U_{n_0}^{1:N}  \left[ \mathcal{H}_{X_{n_0}|X_{n_1}} \right]$, the terminals in $\mathcal{N}^1(n_0)$ estimate $X_{n_0}^{1:N}$ with the successive cancellation decoder of \cite{Arikan10}, and then form $\widehat{K}$ an estimate of $K$.
    \STATE Let $k \in \llbracket 1 ,d-1 \rrbracket$, $j \in \mathcal{N}^{k+1}(n_0)$ and define the singleton $\{i_k\} \triangleq \mathcal{N}^{k}(n_0) \cap \mathcal{N}^1(j)$. 
     With $F_{k,i_k}$ Terminal $j$ estimates $X_{i_k}^{1:N}$ (at distance $k$ from the root) with the successive cancellation decoder of \cite{Arikan10}.\\ By repeating this process, Terminal $j$ is successively able to form the estimate of sources closer to the root, $X_{i_{k-1}}^{1:N}$, $X_{i_{k-2}}^{1:N}$, \ldots, $X_{i_{1}}^{1:N}$, for some $i_1 \in \mathcal{N}^1(n_0)$, $i_2 \in \mathcal{N}^2(n_0)$, \ldots, $i_{k-1} \in \mathcal{N}^{k-1}(n_0)$.
Finally, from its estimate of $X_{i_{1}}^{1:N}$, Terminal $j$ estimates $X_{{n_0}}^{1:N}$ and forms $\widehat{K}$ an estimate of $K$.
    \RETURN $\widehat{K}$
  \end{algorithmic}
\end{algorithm}

We note $\mathcal{F}$ the set of indices $(i,j)$ for which $F_{i,j}$ is defined. We note the collective inter-terminals communication as $\mathbf{F} \triangleq \{ F_{i,j} \}_{ (i,j) \in  \mathcal{F}}$.

The analysis of the scheme in Section~\ref{SecproofTh4} leads to the following result.

\begin{thm} \label{prop_part}
Consider the Markov tree model with uniform marginals. The secret-key capacity $C_{\text{SK}} (+\infty)$ given in Proposition \ref{prop_line} is achievable with perfect secrecy with the polar coding scheme of Section~\ref{Sec_scheme4}, whose computational complexity is $O(N \log N)$. No pre-shared seed is required.
\end{thm}
\begin{proof}
See Section \ref{SecproofTh4}.
\end{proof}

\subsection{Analysis of polar coding scheme: Proof of Theorem \ref{prop_part}} \label{SecproofTh4}

\subsubsection{Key Rate} From~\cite{Arikan10}, we obtain the key rate
\begin{align*} 
	\lim_{N \to \infty } \frac{|\mathcal{H}_{X_{n_0}|X_{n_1}}^c |}{N} 
	& = 1 -  \lim_{N \to \infty } \frac{|\mathcal{H}_{X_{n_0}|X_{n_1}} |}{N} \\
	& = 1- H(X_{n_0}|X_{n_1}) \\
	&= I(X_{n_0};X_{n_1}).
\end{align*}

\subsubsection{Reliability} \label{SecrelModel4}
For $k \in \llbracket 1 , d \rrbracket$, we define the singleton $\{j_0\} \triangleq \mathcal{N}^1(j) \cap \mathcal{N}^{k-1}(n_0)$, and we show that Terminal $j \in \mathcal{N}^k(n_0)$ can reconstruct $X_{ j_0}$ from $F_{k-1,j_0}$. 
Specifically, we establish the following.

\begin{lem} \label{lemincl2}
Let $k \in \llbracket 1 , d \rrbracket$, $j \in \mathcal{N}^k(n_0)$,  and define the singleton $\{j_0\} \triangleq \mathcal{N}^1(j) \cap \mathcal{N}^{k-1}(n_0)$. Define $\mathcal{D}_{k,j_0} \triangleq \mathcal{N}^1(j_0) \cap \mathcal{N}^{k}(n_0)$, and $i^*\triangleq \displaystyle\argmax_{ \tilde{i} \in \mathcal{D}_{k,j_0}} p_{\tilde{i},j_0} $. We have
$$
\forall i \in \mathcal{D}_{k,j_0}, \text{ } \mathcal{H}_{X_{j_0}|X_i} \subset \mathcal{H}_{X_{j_0}|X_{i^*}}.
 $$
\end{lem}

\begin{proof}
For $i \in \mathcal{D}$, define $\bar{X}_i \triangleq X_{j_0} + B_i$, with $B_i \sim \mathcal{B}(p_{i,j_0})$. By Lemma~\ref{lemincl}, we now that for any $i \in \mathcal{D}$, $\mathcal{H}_{X_{j_0}|\bar{X}_{i}} \subset \mathcal{H}_{X_{j_0}|\bar{X}_{i^*}} $. Then, observe that for any $i \in \mathcal{D}$, for any $x,y \in \{ 0,1\}$, 
\begin{align*}
& p_{\bar{X}_i X_{j_0}} (x,y) \\
& = p_{X_{j_0}} (y)p_{\bar{X}_i | X_{j_0}} (x|y)\\
& = \frac{1}{2} ( \mathds{1}\{ x = y \} (1- p_{i,j_0})  + p_{i,j_0} (1- \mathds{1}\{ x = y \})) \\
& = p_{{X}_i X_{j_0}} (x,y),
\end{align*}
Hence, $ \mathcal{H}_{X_{j_0}|X_{i}} = \mathcal{H}_{X_{j_0}|\bar{X}_{i}} \subset \mathcal{H}_{X_{j_0}|\bar{X}_{i^*}} = \mathcal{H}_{X_{j_0}|X_{i^*}} $

\end{proof}

Lemma \ref{lemincl2} is similar to Lemma~\ref{lemincl}; however, unlike Lemma~\ref{lemincl}, the proof of Lemma \ref{lemincl2} requires uniform marginals.

Now, observe that with $F_{n_0}=U_{n_0}^{1:N}  \left[ \mathcal{H}_{X_{n_0}|X_{n_1}} \right]$, all terminals in $\mathcal{N}^1(n_0)$ can reconstruct $X_{n_0}^{1:N}$ with error probability $O(N\delta_N)$ by Lemma~\ref{lemincl2} and \cite{Arikan10}. We then show by induction that all terminals can reconstruct $X_{n_0}^{1:N}$ with error probability $O(N\delta_N)$. Assume that for $k \in \llbracket 1 ,d-1 \rrbracket$, $X_{n_0}^{1:N}$ can be reconstructed with error probability $O(N\delta_N)$ from any $X_j^{1:N}$, where $j \in \mathcal{N}^k(n_0)$. Let $j \in \mathcal{N}^{k+1}(n_0)$ and define the singleton $\{i\} = \mathcal{N}^{k}(n_0) \cap \mathcal{N}^1(j)$. With $F_{k,i}$ Terminal $j$ can reconstruct $X_i^{1:N}$ with error probability $O(N\delta_N)$ by Lemma~\ref{lemincl2} and \cite{Arikan10}. Then, since $X_i^{1:N} \in \mathcal{N}^k(n_0)$, Terminal $j$ can also reconstruct $X_{n_0}^{1:N}$ with error probability $O(N\delta_N)$ by induction hypothesis.

We conclude that all terminals can reconstruct $X_{n_0}^{1:N}$ and therefore $K = U_{n_0}^{1:N}  \left[ \mathcal{H}^c_{X_{n_0}|X_{n_1}} \right] $ with error probability $\mathbf{P}_e(\mathcal{S}_N) =O(N\delta_N)$. The global reconstruction process is illustrated in Figure~\ref{figmodel4ex}.

\begin{figure}
\centering
  \includegraphics[width=8.5cm]{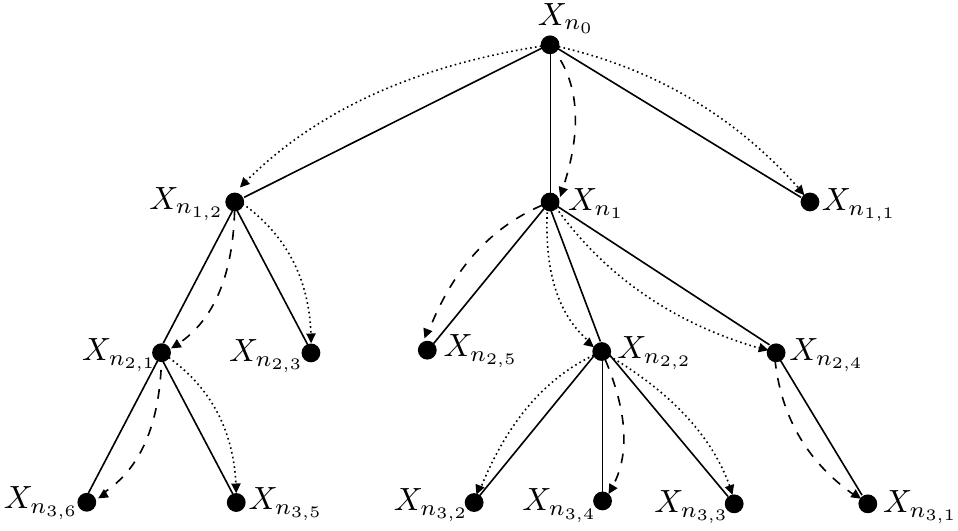}
  \caption{Example for the reconstruction process. A dashed-line from Terminal $i$ to Terminal $j$ represents a public transmission from Terminal $i$ of the information necessary for Terminal $j$ to reconstruct $X_i$. A dotted-line from Terminal $i$ to Terminal $j$ represents a ``virtual communication" and means that Terminal $j$ is able to reconstruct $X_i$ from the information corresponding to the dashed-line leaving Terminal $i$ -- this illustrates Lemma~\ref{lemincl2}. For this example we have assumed $I(X_{n_{1,2}}; X_{n_{2,1}}) \leq  I(X_{n_{1,2}}; X_{n_{2,3}})$, $I(X_{n_{1}}; X_{n_{2,5}}) \leq \min \{ I(X_{n_{1}}; X_{n_{2,i}}) \}_{ i \in \{ 2,4\} } $, $I(X_{n_{2,1}}; X_{n_{3,6}}) \leq  I(X_{n_{2,1}}; X_{n_{3,5}})$, $I(X_{n_{2,2}}; X_{n_{3,4}}) \leq \min \{ I(X_{n_{2,2}}; X_{n_{3,i}}) \}_{ i \in \{ 2,3\} } $. All in all, all the terminals can reconstruct $X_{n_0}$ }
  \label{figmodel4ex}
\end{figure}

\subsubsection{Key Uniformity} By definition of the model, $X_{n_0}$ is uniform, hence, $U_{n_0}^{1:N}$ and $K \triangleq U_{n_0}^{1:N} \left[\mathcal{H}_{X_{n_0}|n_1}^c \right]$ are also uniform.

\subsubsection{Perfect Secrecy} \label{SecsecrModel4}
We first introduce an equivalent model as follows.
We start by defining for $i \in \mathcal{N}^{1}(n_0)$, $\bar{X}_i \triangleq X_{n_0} + B_i$, with $B_i \sim \mathcal{B}(p_{i,n_0})$. Then, for $k \in \llbracket 2 , d \rrbracket$, for $i \in \mathcal{N}^{k}(n_0)$, define the singleton $ \{i_0\} \triangleq \mathcal{N}^{k-1}(n_0) \cap \mathcal{N}^{1}(i)$, and  $\bar{X}_i \triangleq \bar{X}_{ i_0 } + B_i$, with $B_i \sim \mathcal{B}(p_{i,i_0})$. 
Consequently, similarly to the proof of Lemma~\ref{lemincl2}, we have 
\begin{equation} \label{eq_jintdisteq}
p_{\bar{X}_{\mathcal{M}}} = p_{X_{\mathcal{M}}}. 
\end{equation}
Moreover, for $j \in \mathcal{M} \backslash \{n_0\}$. We have 
$$ \bar{U}_{j}^{1:N}= U_{n_0}^{1:N} \displaystyle\bigoplus_{i \in \mathcal{P}_{n_0,j}} \widetilde{B}_i^{1:N},$$
where $\mathcal{P}_{n_0,j}$ denotes the set of vertices that form a path between $X_{n_0}$ and $X_j$ including $j$ and excluding $n_0$, $\widetilde{B}_i^N \triangleq B_{i}^N G_N$, and $ \bar{U}_{j}^{1:N} \triangleq \bar{X}_{j}^{1:N} G_N$, $i \in \mathcal{M} \backslash \{ n_0 \}$. Recall that for $(i,j) \in \mathcal{F}$,
\begin{equation*}
F_{i,j}= U_j^{1:N} \left[ \mathcal{H}_{X_j|X_{j^*}}  \right].
\end{equation*}

We define
\begin{align}
\bar{F}_{i,j} 
& \triangleq \bar{U}_j^{1:N} \left[ \mathcal{H}_{X_j|X_{j^*}}  \right] \nonumber \\
& =   U_{n_0}^{1:N} \left[ \mathcal{H}_{X_j|X_{j^*}}  \right] \displaystyle\bigoplus_{i \in \mathcal{P}_{n_0,j}} \widetilde{B}_i^{1:N}\left[ \mathcal{H}_{X_j|X_{j^*}}  \right], \label{eq_f_bar_i}
\end{align}
and
\begin{equation} \label{eq_f_bar}
\bar{\mathbf{F}} \triangleq \{ \bar{F}_{i,j} \}_{ (i,j) \in  \mathcal{F}}.
\end{equation}

\begin{lem} \label{lemhelp1}
Let $j \in \mathcal{M} \backslash \{n_0\}$. There exists a unique $i \in \llbracket 1, d-1 \rrbracket $ such that $j \in \mathcal{N}^i(n_0)$. As in Algorithm \ref{alg:encoding_4}, define $j^* \triangleq \displaystyle\argmax_{ \tilde{j} \in \mathcal{N}^1(j) \cap \mathcal{N}^{i+1}(n_0)} p_{\tilde{j},j}$. We have $\mathcal{H}_{X_j|X_{j^*}} \subset \mathcal{H}_{X_{n_0}|X_{n_1}}$.
\end{lem}
\begin{proof}
Let $j \in \mathcal{M} \backslash \{n_0\}$. Let $r_j$ be such that $p_{n_0,n_1} = p_{j,j^*}\star r_j$ (such $r_j$ exists by definition of $(n_0,n_1)$), where $\star$ is defined as in Example \ref{examplecapl}. We define $\Delta_j^{(1)} \sim \mathcal{B}(p_{j,j^*})$ and $\Delta_j^{(2)} \sim \mathcal{B}(r_j)$ such that $B_{n_1} = \Delta_j^{(1)}  + \Delta_j^{(2)}$. We define the dummy random variables $\bar{\bar{X}}_{j^*} \triangleq X_{n_0} + \Delta_j^{(1)} $ and $\bar{\bar{X}}_{n_1} \triangleq X_{n_0} + \Delta_j^{(1)} + \Delta_j^{(1)}$. Then, for any $x,y \in \{ 0,1\}$, and by uniformity of the marginals of $p_{X_{\mathcal{M}}}$,
\begin{align*}
&p_{\bar{\bar{X}}_{j^*}X_{n_0}} (x,y)\\
& = p_{X_{n_0}} (y)  p_{\bar{\bar{X}}_{j^*}|X_{n_0}} (x|y) \\
& = \frac{1}{2}  p_{\bar{\bar{X}}_{j^*}|X_{n_0}} (x|y) \\
& = \frac{1}{2} \left[ (1 -\mathds{1}\{ x= y \}) p_{j,j^*} + (1- p_{j,j^*}) \mathds{1}\{ x= y \} \right] \\
& = \frac{1}{2} p_{X_{j^*}|X_{j}} (x|y) \\
& = p_{X_{j^*}X_{j}} (x,y),
\end{align*}
so that $\mathcal{H}_{X_j|X_{j^*}} = \mathcal{H}_{{X}_{n_0}|\bar{\bar{X}}_{j^*}}$. Similarly, we have $p_{{{X}}_{n_1}X_{n_0}} = p_{\bar{\bar{X}}_{n_1}X_{n_0}}$ so that  $\mathcal{H}_{X_{n_0}|X_{n_1}} = \mathcal{H}_{X_{n_0}|\bar{\bar{X}}_{n_1}}$. Hence, by the data processing inequality, we obtain $\mathcal{H}_{X_j|X_{j^*}} = \mathcal{H}_{{X}_{n_0}|\bar{\bar{X}}_{j^*}}  \subset \mathcal{H}_{X_{n_0}|\bar{\bar{X}}_{n_1}} = \mathcal{H}_{X_{n_0}|X_{n_1}}$. 
\end{proof}
We can now show that perfect secrecy holds as follows.
\begin{align*}
&\mathbf{L}(\mathcal{S}_N) \\
& = I(K;\mathbf{F}) \\
& = I \left( U_{n_0}^{1:N} \left[\mathcal{H}_{X_{n_0}|X_{n_1}}^c \right] ; \mathbf{F} \right)\\
& \stackrel{(a)}{=} I \left( \bar{U}_{n_0}^{1:N} \left[\mathcal{H}_{X_{n_0}|X_{n_1}}^c \right] ; \bar{\mathbf{F}} \right)\\
& \stackrel{(b)}{\leq} I \left(  \bar{U}_{n_0}^{1:N} \left[\mathcal{H}_{X_{n_0}|X_{n_1}}^c \right] \right. \\
&\phantom{------} ; \left.  \bar{U}_{n_0}^{1:N} \left[\mathcal{H}_{X_{n_0}|X_{n_1}} \right] , \widetilde{B}^{1:N}_{\llbracket 1, m\rrbracket \backslash \{ n_0 \}}\left[ \mathcal{H}_{X_{n_0}|X_{n_1}}  \right] \right) \displaybreak[0]\\
& = I \left(  \bar{U}_{n_0}^{1:N} \left[\mathcal{H}_{X_{n_0}|X_{n_1}}^c \right]  ;  \bar{U}_{n_0}^{1:N} \left[\mathcal{H}_{X_{n_0}|X_{n_1}} \right]  \right)   \\
& \phantom{--} + I \left(  \bar{U}_{n_0}^{1:N} \left[\mathcal{H}_{X_{n_0}|X_{n_1}}^c \right]  \right. \\
&\phantom{------} ; \left. \left. \widetilde{B}^{1:N}_{\llbracket 1, m\rrbracket \backslash \{ n_0 \}}\left[ \mathcal{H}_{X_{n_0}|X_{n_1}} \right] \right\rvert \bar{U}_{n_0}^{1:N} \left[\mathcal{H}_{X_{n_0}|X_{n_1}}   \right] \right)  \displaybreak[0] \\
& \stackrel{(c)}{=}    I \left(  \bar{U}_{n_0}^{1:N} \left[\mathcal{H}_{X_{n_0}|X_{n_1}}^c \right]  \right. \\
&\phantom{------} ; \left. \left. \widetilde{B}^{1:N}_{\llbracket 1, m\rrbracket \backslash \{ n_0 \}}\left[ \mathcal{H}_{X_{n_0}|X_{n_1}} \right] \right\rvert \bar{U}_{n_0}^{1:N} \left[\mathcal{H}_{X_{n_0}|X_{n_1}}   \right] \right)  \\
& \leq  I \left(  \bar{U}_{n_0}^{1:N}  ; \widetilde{B}^{1:N}_{\llbracket 1, m\rrbracket \backslash \{ n_0 \}}\left[ \mathcal{H}_{X_{n_0}|X_{n_1}} \right] \right) \\
& \stackrel{(d)}{=} 0,
\end{align*}
where $(a)$ follows by (\ref{eq_jintdisteq}), (\ref{eq_f_bar_i}), and (\ref{eq_f_bar}), $(b)$ follows from Lemma~\ref{lemhelp1} and Equation (\ref{eq_f_bar_i}), $(c)$ follows by uniformity of $\bar{U}_{n_0}^{1:N}$, $(d)$ holds by independence of $\bar{U}_{n_0}^{1:N}$ and $\widetilde{B}^{1:N}_{\llbracket 1, m\rrbracket \backslash \{ n_0 \}}$. 
We have thus shown perfect secrecy. 

\section{Application to Secrecy and Privacy for Biometric Systems} \label{Sec_bio}
In this final section, we show how the results obtained for Model 2 may be applied to the related problems of secrecy and privacy for biometric systems~\cite{Dodis08,Ignatenko09,Lai08,Rane13}. As noted in \cite{Ignatenko09}, the main difficulty in constructing practical codes for such problems is the need for vector quantization; we show here that polar codes offer a low-complexity solution and provably optimal solutions for the models studied in~\cite{Ignatenko09}.
\subsection{Biometric system models} \label{Sec_stat_bio}
Consider two biometric sequences $X^{1:N}$ and $Y^{1:N}$ distributed according to the memoryless source $(\mathcal{X}\mathcal{Y},p_{XY})$. Assume that $X^{1:N}$ is an enrollment sequence and $Y^{1:N}$ an authentication sequence observed by an encoder and a decoder, respectively. In~\cite{Ignatenko09}, four different models are considered. We only deal with the ``generated-secret systems" and the ``generated-secret systems with zero leakage," as codes for the latter models can be used for the ``chosen-secret systems'' and the ``chosen-secret systems with zero leakage'' using a masking technique~\cite{Ignatenko09}.

\subsubsection{Generated-secret systems} 

\begin{figure}
\centering
  \includegraphics[width=8.5cm]{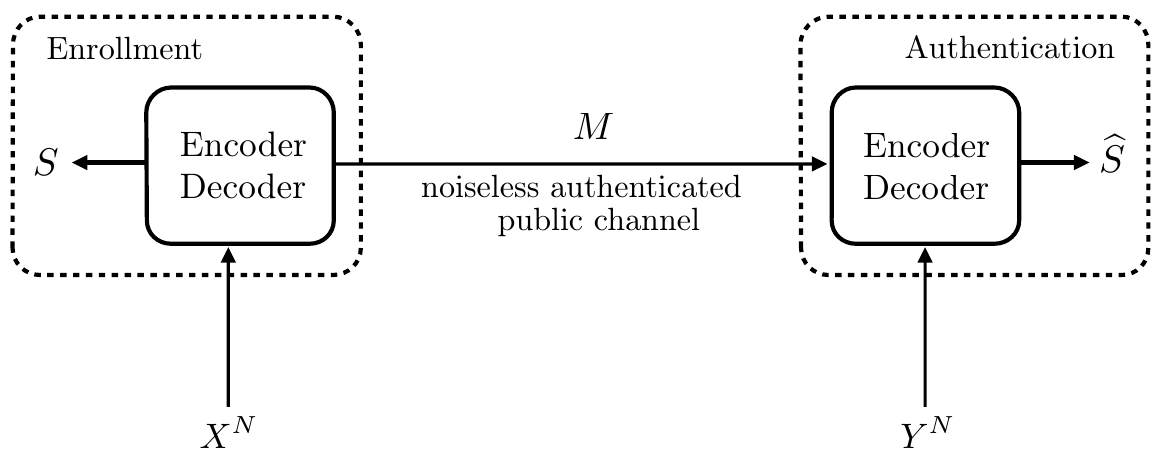}
  \caption{Model for biometric secret generation}
  \label{figbio}
\end{figure}

A biometric secret generation strategy $\mathcal{S}^{\textup{bio}}_N$ is illustrated in Fig.~\ref{figbio} and is formally defined as follows.
\begin{defn}
Let $R \in \mathbb{R}^+ $. Let  $\mathcal{S}$ be an alphabet of size $2^{NR}$. The protocol defined by the following steps is called a $(2^{NR},N,R)$ biometric secret generation strategy.
\begin{itemize}
\item The encoder observes the enrollment sequence $X^{1:N}$;
\item The encoder generates a secret $S \in \mathcal{S}$ from $X^{1:N}$;
\item The encoder transmits publicly to the decoder helper data ${M}$;
\item The decoder observes the authentication sequence $Y^{1:N}$, and computes $\widehat{S} \in \mathcal{S}$.
\end{itemize}
\end{defn}

The performance of a biometric secret generation strategy is measured in terms of 
\begin{itemize}
\item the average probability of error between the biometric secrets with $\textbf{P}_e(\mathcal{S}_N^{\textup{bio}}) \triangleq \mathbb{P} [ S \neq \widehat{S}],$
\item the information leakage of $M$ on $S$ with $\textbf{L}(\mathcal{S}_N^{\textup{bio}}) \triangleq {I} (M;S),$ 
\item the privacy leakage of $M$ on $X^{1:N}$ with $\textbf{P}_{\textup{c}}(\mathcal{S}_N^{\textup{bio}}) \triangleq I(M;X^{1:N}|S)$ (conditional case), or $\textbf{P}_{\textup{u}}(\mathcal{S}_N^{\textup{bio}}) \triangleq I(M;X^{1:N})$ (unconditional case), 
\item the uniformity of the biometric secret $\textbf{U}(\mathcal{S}_N^{\textup{bio}}) \triangleq \log \lceil 2^{NR} \rceil - {H}(S)$.
\end{itemize}
\begin{defn} \label{def2}
For a fixed privacy leakage threshold $L$, a biometric secret rate $R$ and information  is achievable if there exists a sequence of $(2^{NR},N,R)$ secret-key generation strategies $\left\{ \mathcal{S}_N^{\textup{bio}} \right\}_{N \geq 1}$ such that
\begin{align*}
 \displaystyle\lim_{N \to \infty } \textbf{\textup{P}}_e(\mathcal{S}_N^{\textup{bio}})  = & 0, \text{ (reliability) } \\
 \displaystyle\lim_{N \to \infty } \textbf{\textup{L}}(\mathcal{S}_N^{\textup{bio}}) = & 0, \text{ (strong secrecy)}  \\
  \displaystyle\lim_{N \to \infty } \textbf{\textup{P}}_{\textup{c}}(\mathcal{S}_N^{\textup{bio}})/N \leq & L, \text{ (privacy leakage)}  \\
 \displaystyle\lim_{N \to \infty } \textbf{\textup{U}}(\mathcal{S}_N^{\textup{bio}})= & 0. \text { (uniformity)}
\end{align*}
Moreover, the supremum of achievable rates is called the biometric secret capacity and is denoted $C_{\text{Bio}}^{\textup{c}}(L)$.
For the unconditional case, $\textbf{\textup{P}}_{\textup{c}}(\mathcal{S}_N^{\textup{bio}})$ is replaced with $\textbf{\textup{P}}_{\textup{u}}(\mathcal{S}_N^{\textup{bio}})$, and the biometric secret capacity and is denoted by $C_{\text{Bio}}^{\textup{u}}(L)$.
 \end{defn}
Note that we require a stronger security metric than in \cite{Ignatenko09}. The biometric secret capacities are known and recalled below.

\begin{thm} [\!\! \cite{Ignatenko09}]
Let $(\mathcal{X}\mathcal{Y},p_{XY})$ be a BMS and $L \in \mathbb{R}_+$ be a privacy leakage threshold. The conditional and unconditional biometric secret capacities are equal $C_{\text{Bio}}^{\textup{c}}(L) =C_{\text{Bio}}^{\textup{u}}(L)$, moreover,
\begin{align*}
C_{\text{Bio}}^{\textup{c}}(L) 
 = \displaystyle\max_{U} {I}(Y;U) 
\end{align*}
\vspace*{-1.11em}
\text{ subject to }
\vspace*{-1em}
\begin{align*} 
& L = {I}(U;X) - I(U;Y), \\ \nonumber
&  U \to X \to Y, \\ \nonumber
&|\mathcal{U}| \leq |\mathcal{X}|.
\end{align*}
\end{thm}
\begin{rem}
The equality $ L = {I}(U;X) - I(U;Y)$ and the range constraint $|\mathcal{U}| \leq |\mathcal{X}|$ are obtained~from~\cite{Chou12b}.
\end{rem}

\subsubsection{Generated-secret systems with zero leakage} 
\begin{figure}
\centering
  \includegraphics[width=8.5cm]{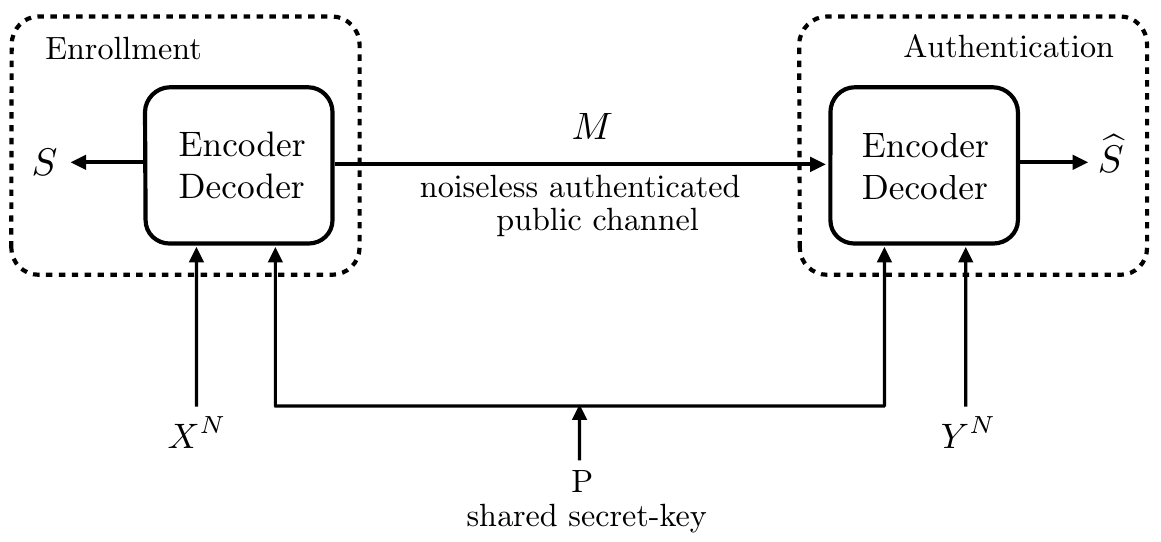}
  \caption{Model for biometric secret generation with zero leakage}
  \label{figbioZ}
\end{figure}

A biometric secret generation strategy with zero leakage $\mathcal{S}^{\textup{bioZ}}_N$ is describes in Figure \ref{figbioZ} and is formally defined as follows.
\begin{defn}
Let $R \in \mathbb{R}^+ $. Let  $\mathcal{S}$ be an alphabet of size $2^{NR}$. Assume that the encoder and decoder share a uniformly distributed secret-key $P$ beforehand. The protocol defined by the following steps is called a $(2^{NR},N,R)$ biometric secret generation strategy with zero leakage.
\begin{itemize}
\item The encoder observes the enrollment sequence $X^{1:N}$;
\item The encoder generates a secret $S \in \mathcal{S}$ from $X^{1:N}$ and~$P$;
\item The encoder transmits publicly to the decoder helper data ${M}$ which is a function of $X^{1:N}$ and $P$;
\item The decoder observes the authentication sequence $Y^{1:N}$, and computes $\widehat{S} \in \mathcal{S}$ from $Y^{1:N}$ and~$P$.
\end{itemize}
\end{defn}

The performance of a biometric secret generation strategy with zero leakage is measured in terms of 
\begin{itemize}
\item the average probability of error between the biometric secrets with $\textbf{P}_e(\mathcal{S}_N^{\textup{bio}}) \triangleq \mathbb{P} [ S \neq \widehat{S}],$
\item the information leakage of $M$ on $S$ and $X^{1:N}$ with $\textbf{L}_{\textup{c}}(\mathcal{S}_N^{\textup{bio}}) \triangleq I(SX^{1:N};M)$ (conditional case), or $\textbf{L}_{\textup{u}}(\mathcal{S}_N^{\textup{bio}}) \triangleq I(S;M) + I(X^{1:N};M)$ (unconditional case), 
\item the length of the secret-key $P$ with $\textbf{H}(\mathcal{S}_N^{\textup{bioZ}}) \triangleq |P| - {H}(P)$ , 
\item the uniformity of the biometric secret $\textbf{U}(\mathcal{S}_N^{\textup{bio}}) \triangleq \log \lceil 2^{NR} \rceil - {H}(S)$.
\end{itemize}
\begin{defn} \label{def3}
For a fixed secret-key length $K$, a biometric secret rate $R$ is achievable with zero leakage if there exists a sequence of $(2^{NR},N,R)$ biometric secret generation strategies with zero leakage $\left\{ \mathcal{S}_N^{\textup{bioZ}} \right\}_{N \geq 1}$ such that
\begin{align*}
 \displaystyle\lim_{N \to \infty } \textbf{\textup{P}}_e(\mathcal{S}_N^{\textup{bioZ}})  = & 0, \text{ (reliability) } \\
 \displaystyle\lim_{N \to \infty } \textbf{\textup{L}}_{\textup{c}}(\mathcal{S}_N^{\textup{bioZ}}) = & 0, \text{ (strong secrecy)}  \\
  \displaystyle\lim_{N \to \infty } \textbf{\textup{H}}(\mathcal{S}_N^{\textup{bioZ}})/N \leq & K, \text{ (secret-key length)}  \\
 \displaystyle\lim_{N \to \infty } \textbf{\textup{U}}(\mathcal{S}_N^{\textup{bioZ}})= & 0. \text { (uniformity)}
\end{align*}
Moreover, the supremum of achievable rates is called the zero-leakage biometric secret capacity and is denoted $C_{\text{BioZ}}^{\textup{c}}(L)$.
For the unconditional case $\textbf{\textup{P}}_{\textup{c}}(\mathcal{S}_N^{\textup{bioZ}})$ is replaced with $\textbf{\textup{P}}_{\textup{u}}(\mathcal{S}_N^{\textup{bioZ}})$, and the zero-leakage biometric secret capacity and is denoted $C_{\text{BioZ}}^{\textup{u}}(L)$.
 \end{defn}
Note that we require a stronger security metric than in \cite{Ignatenko09}. The zero-leakage biometric secret capacities are known and recalled below.

\begin{thm} [\!\! \cite{Ignatenko09}]
Let $(\mathcal{X}\mathcal{Y},p_{XY})$ be a BMS and $K \in \mathbb{R}_+$ be a fixed length. The conditional and unconditional zero-leakage biometric secret capacities are equal $C_{\text{BioZ}}^{\textup{c}}(K) =C_{\text{BioZ}}^{\textup{u}}(K))$, moreover,
\begin{align*}
C_{\text{BioZ}}^{\textup{c}}(L) 
 = \displaystyle\max_{U} {I}(Y;U) +K 
\end{align*}
\vspace*{-1.11em}
\text{ subject to }
\vspace*{-1em}
\begin{align*} 
& K = {I}(U;X) - I(U;Y), \\ \nonumber
&  U \to X \to Y, \\ \nonumber
&|\mathcal{U}| \leq |\mathcal{X}|.
\end{align*}
\end{thm}
\begin{rem}
The equality $ K = {I}(U;X) - I(U;Y)$ and the range constraint $|\mathcal{U}| \leq |\mathcal{X}|$ are obtained~from~\cite{Chou12b}.
\end{rem}

\subsection{Polar coding scheme for generated-secret systems} \label{Schemebio}
Let $n \in \mathbb{N}$ and $N \triangleq 2^n$. Fix a joint probability distribution $p_{XU}$. We note $V^{1:N} \triangleq U^{1:N} G_N$. For $\delta_N \triangleq 2^{-N^{\beta}}$, where $\beta \in ]0,1/2[$, define the following sets
\begin{align*}
\mathcal{H}_{U}  \triangleq &  \left\{ i\in \llbracket 1,N\rrbracket : {H} \left(V^i|V^{1:i-1}  \right) \geq  \delta_N \right\} , \\
\mathcal{V}_{U}  \triangleq &  \left\{ i\in \llbracket 1,N\rrbracket : {H} \left(V^i|V^{1:i-1} \right) \geq 1- \delta_N \right\} , \\
\mathcal{V}_{U|X}  \triangleq &  \left\{ i\in \llbracket 1,N\rrbracket : {H} \left(V^i|V^{1:i-1} X^{1:N} \right) \geq 1- \delta_N \right\} , \\
\mathcal{H}_{U|Y} \triangleq & \left\{ i\in \llbracket 1,N\rrbracket : {H}\left(V^i|V^{1:i-1} Y^{1:N}\right) \geq \delta_N \right\}, \\
\mathcal{H}_{U|X} \triangleq & \left\{ i\in \llbracket 1,N\rrbracket : {H}\left(V^i|V^{1:i-1} X^{1:N}\right) \geq \delta_N \right\}.
\end{align*}
The scheme proposed is a special case (it corresponds to the case $Z = \emptyset$) of the scheme in Section~\ref{Sec_scheme2}. However, for completeness and clarity, we provide its detailed description in Algorithm~\ref{alg:encoding_gen_sec} and Algorithm~\ref{alg:decoding_gen_sec} with the notation of the biometric secret generation problem. 

\begin{algorithm}[]
  \caption{Encoding algorithm for generated secret systems}
  \label{alg:encoding_gen_sec}
  \begin{algorithmic} [1]   
    \REQUIRE $\widetilde{S}_0$, a secret key of size $|(\mathcal{H}_{U|Y} \backslash {\mathcal{V}}_{U|X}) \backslash  \mathcal{V}_{U} |$; $\mathcal{A}_{UXY}$ be any subset of $\mathcal{V}_{U} \backslash \mathcal{H}_{U|Y}$ with size $|(\mathcal{H}_{U|Y} \backslash {\mathcal{V}}_{U|X}) \backslash  \mathcal{V}_{U}|$; Observations $X_i^{1:N}$ in every block $i\in \llbracket 1,k\rrbracket$; a vector $R_1$ of uniformly distributed bits with size $|\mathcal{V}_{U|X}|$.
    \STATE Transmit $R_1$ publicly.
    \FOR{Block $i=1$ to $k$}
    \STATE $\widetilde{V}_i^{1:N}[\mathcal{V}_{U|X}]\leftarrow R_1$
    \STATE Given $X_i^{1:N}$, successively draw the remaining bits of $\widetilde{V}_i^{1:N}$  according to $\widetilde{p}_{V_i^{1:N}X_i^{1:N}} \triangleq \prod_{j=1}^N \widetilde{p}_{V_i^{j}|V_i^{j-1} X^{1:N}} p_{X^{1:N}}$ with 
        \begin{align} 
      & \widetilde{p}_{V_i^j|V_i^{1:j-1}X^{1:N}} (v^j|\widetilde{V}_i^{1:j-1}X_i^{1:N}) \nonumber \\ 
      &\triangleq \begin{cases}
        {p}_{V^j|V^{1:j-1}X^{1:N}} (v^j| \widetilde{V}_i^{1:j-1}X_i^{1:N}) & \!\!\! \text{if } i \in \mathcal{H}_U \backslash {\mathcal{V}}_{U|X}\\
        {p}_{V^j|V^{1:j-1}} (v^j|\widetilde{V}_i^{1:j-1}) &  \!\!\! \text{if } i \in \mathcal{H}_U^c
      \end{cases} \label{eq_VQ_def_bio}
    \end{align}
    \STATE $\widetilde{S}_i \leftarrow \widetilde{V}_i^{1:N} [\mathcal{A}_{UXY}]$
    \STATE $S_i \leftarrow \widetilde{V}_i^{1:N} [(\mathcal{V}_{U} \backslash \mathcal{H}_{U|Y}) \backslash \mathcal{A}_{UXY}]$
    \STATE $F_i \leftarrow \widetilde{V}_i^{1:N} [(\mathcal{H}_{U|Y} \backslash {\mathcal{V}}_{U|X}) \cap \mathcal{V}_{U}]$
    \STATE $F_i' \leftarrow \widetilde{V}_i^{1:N} [(\mathcal{H}_{U|Y} \backslash {\mathcal{V}}_{U|X}) \backslash  \mathcal{V}_{U} ]$
    \STATE  Transmit $M_i \leftarrow [F_i, F_i' \oplus \widetilde{S}_{i-1}]$ publicly
    \ENDFOR
    \RETURN ${S}_{1:k} \leftarrow [ {S}_{1}, {S}_2, \ldots, {S}_k ]$
  \end{algorithmic}
\end{algorithm}

\begin{algorithm}[]
  \caption{Decoding algorithm for generated secret systems}
  \label{alg:decoding_gen_sec}
  \begin{algorithmic} [1]   
    \REQUIRE The secret-key $\widetilde{S}_0$, and the set $\mathcal{A}_{UXY}$ defined in Algorithm \ref{alg:encoding_gen_sec}; Observations $Y_i^{1:N}$ and message $M_i$ transmitted by other party in every block $i\in \llbracket 1,k\rrbracket$, vector~$R_1$.
    \FOR{Block $i=1$ to $k$}
    \STATE Form $\widetilde{V}_i^{1:N}[\mathcal{H}_{U|Y}]$ from $(F_i,F_i')= \widetilde{V}_i^{1:N} [\mathcal{H}_{U|Y} \backslash {\mathcal{V}}_{U|X}]$ and $R_1 = R_i = \widetilde{V}_i^{1:N} [ {\mathcal{V}}_{U|X}]$.
    \STATE Create estimate $\widehat{V}_i^{1:N}$ of $\widetilde{V}_i^{1:N}$ with the successive cancellation decoder of \cite{Arikan10}
    \STATE $\widehat{S}_i\leftarrow \widehat{V}_i^{1:N}[(\mathcal{V}_{U} \backslash \mathcal{H}_{U|Y}) \backslash \mathcal{A}_{UXY}]$ 
    \STATE ${\widetilde{S}}_i\leftarrow \widehat{V}_i^{1:N}[\mathcal{A}_{UXY}]$ 
    \ENDFOR .
    \RETURN $\widehat{S}_{1:k} \triangleq [ \widehat{S}_{1}, \widehat{S}_2, \ldots, \widehat{S}_k ]$.
  \end{algorithmic}
\end{algorithm}
\begin{rem} \label{rm:efficient_use_key3}
One may actually use $S_k^{1:N} [\mathcal{V}_{U} \backslash \mathcal{H}_{U|Y}]$ as the $S_k$ and slightly increase the biometric secret rate in Algorithm~\ref{alg:encoding_gen_sec}. However, one does not distinguish the last block from the others for convenience -- see Remark \ref{rm:efficient_use_key}.
%
%
\end{rem}

Based on the results established for Model 2 in Section~\ref{sec_model2}, we obtain the following.
\begin{thm} \label{Thbio}
Consider a BMS $(\mathcal{X}\mathcal{Y},p_{XY})$. Assume that the encoder and the decoder share a secret seed. For any $L \in \mathbb{R}$, the biometric secret capacities $C_{\text{Bio}}^{\textup{c}}(L)$, and $C_{\text{Bio}}^{\textup{u}}(L)$, are achieved by the polar coding scheme of Algorithm~\ref{alg:encoding_gen_sec} and Algorithms~\ref{alg:decoding_gen_sec}, which involves a chaining of $k$ blocks of size $N$, and whose complexity is $O(kN \log N)$. Moreover, the seed rate is in $o\left( 2^{-N^{\alpha} }\right)$, $\alpha < 1/2$.
\end{thm}
Theorem~\ref{Thbio} is a direct consequence of Theorem~\ref{Th2} for the particular case $Z = \emptyset$, since 
\begin{align*}
& \frac{1}{kN}\max (\textbf{\textup{P}}_{\textup{c}}(\mathcal{S}_N^{\textup{bio}}), \textbf{\textup{P}}_{\textup{u}}(\mathcal{S}_N^{\textup{bio}}))\\
  & \leq \frac{1}{kN} H(M_{1:k}) \\
  & \leq  \frac{1}{kN}\sum_{i=1}^k H(M_i) \\
  & \leq  \frac{1}{kN}\sum_{i=1}^k \log |M_i| \\
  & =   \frac{1}{N}|\mathcal{H}_{U|Y} \backslash {\mathcal{V}}_{U|X}| \\
  & =   \frac{1}{N}|\mathcal{H}_{U|Y}| - \frac{1}{N}|\mathcal{V}_{U|X}| \\
  & \xrightarrow{N \to \infty} I(U;X) - I(U;Y),
 \end{align*}
 where we have used $U \to X \to Y$, \cite{Arikan10}, and Lemma \ref{lemnewcard}.
 
  Note also that for $i \in \llbracket 0,k-1\rrbracket$, $\widetilde{S}_i = o(N)$.

\subsection{Polar coding scheme for generated-secret systems with zero leakage} \label{SchemebioZ}
The encoding and decoding algorithms are given in Algorithm~\ref{alg:encoding_gen_sec_zero} and Algorithm~\ref{alg:decoding_gen_sec_zero}. The difference with the scheme of Section~\ref{Schemebio} is that the public communication is protected with a secret-key shared by the encoder and the decoder.

\begin{algorithm}[]
  \caption{Encoding algorithm for generated secret systems with zero leakage}
  \label{alg:encoding_gen_sec_zero}
  \begin{algorithmic} [1]   
    \REQUIRE $k$ secret keys $\{{P}_i\}_{i \in \llbracket 1, k\rrbracket}$ of size $|\mathcal{H}_{U|Y} \backslash {\mathcal{V}}_{U|X} |$; observations $X_i^{1:N}$ in every block $i\in \llbracket 1,k\rrbracket$; a vector $R_1$ of uniformly distributed bits with size $|\mathcal{V}_{U|X}|$.
    \STATE Transmit $R_1$ publicly.
    \FOR{Block $i=1$ to $k$}
    \STATE $\widetilde{V}_i^{1:N}[\mathcal{V}_{U|X}]\leftarrow R_1$
    \STATE Given observations $X_i^{1:N}$, successively draw the remaining bits of $\widetilde{V}_i^{1:N}$ according to $\widetilde{p}_{V_iX_i}$ defined by~(\ref{eq_VQ_def_bio}).
    \STATE $F_i \leftarrow \widetilde{V}_i^{1:N} [(\mathcal{H}_{U|Y} \backslash {\mathcal{V}}_{U|X}) \cap \mathcal{V}_{U}]$
    \STATE $F_i' \leftarrow \widetilde{V}_i^{1:N} [(\mathcal{H}_{U|Y} \backslash {\mathcal{V}}_{U|X}) \backslash  \mathcal{V}_{U} ]$
    \STATE $S_i \leftarrow [\widetilde{V}_i^{1:N} [\mathcal{V}_{U} \backslash \mathcal{H}_{U|Y}], F_i]$
    \STATE  Transmit $M_i \leftarrow [F_i, F_i']\oplus P_i$ publicly
    \ENDFOR
    \RETURN ${S}_{1:k} \leftarrow [ {S}_{1}, {S}_2, \ldots, {S}_k ]$
  \end{algorithmic}
\end{algorithm}

\begin{algorithm}[]
  \caption{Decoding algorithm for generated secret systems with zero leakage}
  \label{alg:decoding_gen_sec_zero}
  \begin{algorithmic} [1]   
    \REQUIRE The secret key ${P}_i$, the message $M_i$ transmitted by other party, observations $Y_i^{1:N}$ in every block $i\in \llbracket 1,k\rrbracket$, and vector $R_1$.
    \FOR{Block $i=1$ to $k$}
    \STATE Form $\widetilde{V}_i^{1:N}[\mathcal{H}_{U|Y}]$ from $(F_i,F_i')= \widetilde{V}_i^{1:N} [\mathcal{H}_{U|Y} \backslash {\mathcal{V}}_{U|X}]$ and $R_1 = \widetilde{V}_i^{1:N} [ {\mathcal{V}}_{U|X}]$.
    \STATE Create estimate $\widehat{V}_i^{1:N}$ of $\widetilde{V}_i^{1:N}$ with the successive cancellation decoder of \cite{Arikan10}
    \STATE $\widehat{S}_i\leftarrow [\widehat{V}_i^{1:N}[\mathcal{V}_{U} \backslash \mathcal{H}_{U|Y}], F_i]$ 
    \ENDFOR .
    \RETURN $\widehat{S}_{1:k} \triangleq [ \widehat{S}_{1}, \widehat{S}_2, \ldots, \widehat{S}_k ]$.
  \end{algorithmic}
\end{algorithm}

The performance of the algorithms is ensured by the following result.

\begin{thm} \label{ThbioZ}
Consider a BMS $(\mathcal{X}\mathcal{Y},p_{XY})$. For any $P \in \mathbb{R}$, the zero-leakage biometric secret capacities $C_{\text{BioZ}}^{\textup{c}}(K)$, and $C_{\text{BioZ}}^{\textup{u}}(K)$, are achieved by the polar coding scheme of Algorithm~\ref{alg:encoding_gen_sec_zero} and Algorithms~\ref{alg:decoding_gen_sec_zero}, which involves a chaining of $k$ blocks of size $N$, and whose complexity is $O(kN \log N)$.
\end{thm}

Remark that one only needs to prove that $C_{\text{BioZ}}^{\textup{c}}(K)$ is achieved in Theorem \ref{ThbioZ}, since a code that achieves $C_{\text{BioZ}}^{\textup{c}}(K)$ also achieves $C_{\text{BioZ}}^{\textup{u}}(K)$ by \cite{Ignatenko09}. The proof of Theorem~\ref{ThbioZ} for $C_{\text{BioZ}}^{\textup{c}}(K)$ is similar to the proof of Theorem~\ref{Th2} and is thus omitted. To show that $S_i=[\widetilde{V}_i^{1:N} [\mathcal{V}_{U} \backslash \mathcal{H}_{U|Y}], F_i]$, $i\in \llbracket 1,k \rrbracket$, is uniform one can use Lemma~\ref{lem_U2}, then, similarly to Theorem~\ref{Th2}, one can show that $S_{1:k}$ is also uniform and that strong secrecy holds. Note also that for $i \in \llbracket 0,k-1\rrbracket$, $F_i' = o(N)$.

\section{Conclusion}
We have proposed low-complexity secret-key capacity-achieving schemes based on polar coding for several classes of sources. Our schemes jointly handle secrecy and reliability, which contrasts with sequential methods that successively perform reconciliation and privacy amplification. Although sequential methods apply to more general classes of sources, our polar coding schemes may be easier to design and may operate with lesser complexity. Nevertheless, the price to be paid for low complexity is that our schemes often require a pre-shared seed, whose rate is negligible compared to the blocklength. When the eavesdropper has no access to correlated observations of the source, and when the source has uniform marginals, we have identified several configurations, including multiterminal models, for which no pre-shared seed is required. Finally, we have applied our polar coding schemes to privacy and secrecy for some biometric systems.

Our polar coding schemes are particularly convenient to handle rate-limited public communication and vector quantization, which are often the major hurdle in designing optimal secret-key generation schemes.

\appendices

\section{Proofs for Model 1 in Section~\ref{sec_model1}} 
\subsection{Proof of Corollary~\ref{Cor_1}} \label{App_cor1}
We perform the same encoding as in Algorithm~\ref{alg:encoding_1} for Block~1 with $\mathcal{A}_{XYZ} = \emptyset$.  Define the set
\begin{align*}
\mathcal{H}_{X|Z}  \triangleq  &\left\{ i\in \llbracket 1,N\rrbracket : {H}\left(U^i|U^{1:i-1} Z^{1:N}\right) \geq \delta_N \right\}.
\end{align*}
We have
\begin{align*}
|F'_1| 
& = |\mathcal{H}_{X|Y} \backslash  \mathcal{V}_{X|Z}| \\
& \stackrel{(a)}{\leq} |\mathcal{H}_{X|Z} \backslash  \mathcal{V}_{X|Z}| \\
& \stackrel{(b)}{=} |\mathcal{H}_{X|Z} | -|  \mathcal{V}_{X|Z}|,
\end{align*}
where $(a)$ holds because $\mathcal{H}_{X|Y} \subset \mathcal{H}_{X|Z}$ since we have assumed $X \to Y \to Z$, $(b)$ holds because ${\mathcal{V}}_{X|Z} \subset \mathcal{H}_{X|Z}$. 

We conclude by Lemma~\ref{lemnewcard} and \cite{Arikan10} that $|F'_1| = o(N)$.

\subsection{Proof of Proposition \ref{ex1}} \label{App_Ex1}
\subsubsection{Polar Coding Scheme}
 Let $n \in \mathbb{N}$ and $N \triangleq 2^n$. 
%
We set $U^{1:N} \triangleq X^{1:N} G_N$. We define for $\delta_N \triangleq 2^{-N^{\beta}}$, $\beta \in ]0, 1/2[$, the following set
\begin{align*}
{\mathcal{V}}_{X|Z} & \triangleq \left\{ i\in \llbracket 1,N \rrbracket : H \left(U^{i} | U^{1:i-1}Z^{1:N}\right) \geq 1- \delta_N \right\}.
\end{align*}

Alice and Bob define the key as $K \triangleq U^{1:N}[{\mathcal{V}}_{X|Z}]$.

\subsubsection{Scheme analysis}
By Lemma \ref{lemnewcard}, we have a key rate that satisfies
\begin{equation*}
\lim_{N \rightarrow + \infty} |{\mathcal{V}}_{X|Z}| / N = {H}(X|Z).
\end{equation*}
Moreover, we also have secrecy and key uniformity
\begin{align*}
\mathbf{L}(\mathcal{S}_N) + \mathbf{U}(\mathcal{S}_N) 
&=  I\left(K;Z^{1:N}\right) +  |K| - H(K) \\
&= |K| - H\left(K|Z^{1:N}\right) \\
&= |{\mathcal{V}}_{X|Z}| - H\left(U^{1:N}[{\mathcal{V}}_{X|Z}]|Z^{1:N}\right) \\
& \stackrel{(a)}{\leq} |{\mathcal{V}}_{X|Z}| - \sum_{i\in {\mathcal{V}}_{X|Z}} H(U^i|U^{1:i-1}Z^{1:N}) \\
&  \stackrel{(b)}{\leq} |{\mathcal{V}}_{X|Z}|\delta_N\\
& \leq N \delta_N,
\end{align*}
where $(a)$ holds because conditioning reduces entropy, $(b)$ holds by definition of ${\mathcal{V}}_{X|Z}$.

Finally, since $X=Y$, we have $\mathbf{P}_e(\mathcal{S}_N) = 0$.

\subsection{Proof of Proposition \ref{ex3}} \label{App_Ex2}

\subsubsection{Polar Coding Scheme}
 Let $n \in \mathbb{N}$ and $N \triangleq 2^n$. 
%
We set $U^{1:N} \triangleq X^{1:N} G_N$. We define for $\delta_N \triangleq 2^{-N^{\beta}}$, $\beta \in ]0, 1/2[$, the following sets
\begin{align*}
\mathcal{H}_{X|Y} & \triangleq \left\{ i\in \llbracket 1,N \rrbracket : H \left(U^{i} | U^{1:i-1}Y^{1:N}\right) \geq \delta_N \right\}, \\
\mathcal{H}_{X} & \triangleq \left\{  i\in \llbracket 1,N \rrbracket : H\left(U^i | U^{1:i-1}\right) \geq \delta_N \right\}.
\end{align*}

We define a secret-key generation strategy $\mathcal{S}_N$ as follows. Define the key as $K \triangleq U^{1:N} [\mathcal{H}_{X} \backslash \mathcal{H}_{X|Y}]$, and the public message as $F \triangleq U^{1:N} [\mathcal{H}_{X|Y}]$.

\subsubsection{Scheme analysis}
Observe that $\mathcal{H}_{X|Y} \subset \mathcal{H}_{X}$, because conditioning reduces entropy. We thus have by \cite{Arikan10}, a key rate equal to
\begin{align*}
\lim_{N \rightarrow +\infty} \frac{|\mathcal{H}_{X} \backslash \mathcal{H}_{X|Y}| }{N} 
& = \lim_{N \rightarrow +\infty} \frac{|\mathcal{H}_{X}| - | \mathcal{H}_{X|Y}| }{N} \\
& = H(X)-H(X|Y)\\
& = I(X;Y).	
\end{align*}
 
Note that the key $K$ is uniform because $X^{1:N}$ is uniform, that is $$\mathbf{U}_e(\mathcal{S}_N) = 0.$$
Then, by~\cite[Theorem 3]{Arikan10}, Bob can reconstruct $K$ from $F$ with an error probability satisfying 
$$\mathbf{P}_e(\mathcal{S}_N) \leq N \delta_N .$$
Finally, by the key uniformity and because $(\mathcal{H}_{X} \backslash \mathcal{H}_{X|Y}) \cap \mathcal{H}_{X|Y} = \emptyset$ , we have
\begin{align*}
H(K|F) 
&= H\left(U^{1:N} [\mathcal{H}_{X} \backslash \mathcal{H}_{X|Y}]| U^{1:N} [\mathcal{H}_{X|Y}]\right) \\
&= H\left(U^{1:N} [\mathcal{H}_{X} \backslash \mathcal{H}_{X|Y}]\right) \\
& = H(K),
\end{align*}
which means that we obtain perfect secrecy, that is 
$$
\mathbf{L}(\mathcal{S}_N) = I(K;F) = H(K) - H(K|F) = 0.
$$

\subsection{Proof of Lemma~\ref{lemnewcard}} \label{App_bat}
As in \cite{Arikan10}, for a pair of random variables $(X,Y)$ distributed according to $p_{XY}$ over $\mathcal{X} \times \mathcal{Y}$, we define the Bhattacharyya parameter as
$$
Z(X|Y) = 2 \sum_{y} p_{Y}(y) \sqrt{p_{X|Y}(0|y)p_{X|Y}(1|y)}.
$$
We will need the following counterpart of \cite[Proposition 1]{Arikan10} that is proved using the same technique as~\cite[Lemma 20]{Korada10}. 
\begin{lem} \label{lem_counterpart}
If $\left(X_{1},Y_{1} \right)$ and $\left(X_{2},Y_{2} \right)$ are two independent drawings of $(X,Y)$, then
\begin{eqnarray*}
Z\left(X_{1} \oplus X_{2} |Y^{2}_1\right) \geq \sqrt{2Z(X|Y)^2 - Z(X|Y)^4}. 
\end{eqnarray*} 
\end{lem} 
\begin{proof}
We have for any $v_1,$ $v_2 \in \mathcal{X}$, $y_1,$ $y_2 \in \mathcal{Y}$, 
\begin{multline*}p_{X_1\oplus X_2, X_2,Y_1,Y_2} (v_1,v_2,y_1,y_2) \\= p_{XY}(v_1+v_2,y_1)p_{XY}(v_2,y_2).
	\end{multline*}
 Hence,
\begin{align*}
& Z \left(X_1 \oplus X_2 |Y^2_1 \right) \\
&= 2 \sum_{y_1,y_2} \left( \sum_{v_2} p_{XY}\left(v_2,y_1\right)p_{XY}\left(v_2,y_2\right) \right.\\
& \phantom{--} \left. \cdot \sum_{v_2'} p_{XY}\left(1+v_2',y_1\right)p_{XY}\left(v_2',y_2\right) \right)^{1/2},
\end{align*}
which can be rewritten as
\begin{align*}
& Z\left(X_1 \oplus X_2 |Y^2_1 \right) \\
&= \frac{1}{2} Z \left(X_1|Y_1\right) Z\left(X_2|Y_2\right) \\
& \phantom{--}\times  \sum_{y_1,y_2} P_1 \left(y_1\right) P_2\left(y_2\right) \sqrt{A\left(y_1\right)^2 + A\left(y_2\right)^2-4},
\end{align*}
where, for $i \in \llbracket 1,2\rrbracket$, $$P_i\left(y_{i}\right) \triangleq \frac{2 \sqrt{p_{XY}\left(0,y_{i}\right)p_{XY}\left(1,y_{i}\right)}}{Z\left(X_{i}|Y_{i}\right)}$$ and $$A\left(y_{i}\right) \triangleq \sqrt{\frac{p_{XY}\left(0,y_{i}\right)}{p_{XY}\left(1,y_{i}\right)}} + \sqrt{\frac{p_{XY}\left(1,y_{i}\right)}{p_{XY}\left(0,y_{i}\right)}}.$$
As observed in \cite[Lemma 20]{Korada10}, for $i \in \llbracket 1,2\rrbracket$, $A\left(y_{i}\right)^2 \geq 4$, by the arithmetic-geometric inequality, and $x \mapsto \sqrt{x^2 +a}$ is convex for $a>0$.
Hence, since for $i \in \llbracket 1,2\rrbracket$, $P_i$ defines a probability distribution over $\mathcal{Y}$, by Jensen's inequality applied twice
\begin{align*}
&Z\left(X_1 \oplus X_2 |Y^2_1 \right) \\
& \geq \frac{1}{2} Z\left(X_1|Y_1\right) Z\left(X_2|Y_2\right) \\
& \phantom{--} \times \sqrt{ \left(\mathbb{E}_{P_1} \left[A\left(y_1\right)\right]\right)^2 + \left(\mathbb{E}_{P_2} \left[A \left(y_2\right)\right]\right)^2-4 }.
\end{align*}
We conclude by substituting $\mathbb{E}_{P_i} \left[A\left(y_{i}\right)\right] = \frac{2}{Z \left(X_{i}|Y_{i}\right)}$, for $i \in \llbracket 1,2\rrbracket$.
\end{proof}

Let $\alpha \in ] \beta , 1/2[$. Define the sets
\begin{align*}
\mathcal{F}_{X|Z}  & \triangleq \left\{ i\in \llbracket 1,N\rrbracket : {Z}\left(U^i|U^{1:i-1} Z^{1:N}\right) \geq 1- 2^{-N^{\alpha}} \right\},
\end{align*}
\begin{align*}
\mathcal{H}_{X|Z}  \triangleq  &\left\{ i\in \llbracket 1,N\rrbracket : {H}\left(U^i|U^{1:i-1} Z^{1:N}\right) \geq \delta_N \right\}.
\end{align*}
Similar to \cite[Theorem 19]{Korada10}, which relies on the result in \cite{Arikanrate}, we can show with Lemma~\ref{lem_counterpart} 
$$   \lim_{N \rightarrow + \infty}   | \mathcal{F}_{X|Z}| / N  = H(X|Z)  .$$ But, by~\cite[Proposition 2]{Arikan10}, for $N$ large enough, $ | \mathcal{F}_{X|Z}| \leq |\mathcal{V}_{X|Z}|$, hence, $  \lim_{N \rightarrow + \infty}  |\mathcal{V}_{X|Z}| /N \geq H(X|Z)$. Since we also have $  \lim_{N \rightarrow + \infty}  |\mathcal{H}_{X|Z} | /N = H(X|Z)$, by~\cite{Arikan10}, and $\mathcal{V}_{X|Z} \subset \mathcal{H}_{X|Z}$, we conclude $$\lim_{N \rightarrow + \infty}  |\mathcal{V}_{X|Z}| /N = H(X|Z).$$

\subsection{Proof of Lemma~\ref{lem_KKtilde}} \label{App_lemKKtilde}

Let $i \in \llbracket 1,k \rrbracket$, we note $q_{\mathcal{U}_{K,\widetilde{K}}}$ the uniform distribution over $\llbracket 1 , 2^{|K_i|+|\widetilde{K}_i|} \rrbracket$. We have, 
\begin{align}
&\mathbb{V} \left( p_{K_i\widetilde{K}_i}, p_{K_i}p_{\widetilde{K}_i} \right) \nonumber \\ \nonumber
& \stackrel{(a)}{\leq} \mathbb{V} \left( p_{K_i\widetilde{K}_i}, q_{\mathcal{U}_{K,\widetilde{K}}} \right) + \mathbb{V}\left( q_{\mathcal{U}_{K,\widetilde{K}}}, q_{\mathcal{U}_{K}} p_{\widetilde{K}_i} \right) \\ \nonumber
& \phantom{--}+ \mathbb{V} \left( q_{\mathcal{U}_{K}} p_{\widetilde{K}_i}, p_{K_i}p_{\widetilde{K}_i} \right) \\ \nonumber
& = \mathbb{V} \left( p_{K_i\widetilde{K}_i}, q_{\mathcal{U}_{K,\widetilde{K}}} \right) + \mathbb{V}\left( q_{\mathcal{U}_{\widetilde{K}}}, p_{\widetilde{K}_i} \right) + \mathbb{V} \left( q_{\mathcal{U}_{K}}, p_{K_i} \right) \\ \nonumber
& \leq 3 \mathbb{V} \left( p_{K_i\widetilde{K}_i}, q_{\mathcal{U}_{K,\widetilde{K}}} \right) \\ 
& \stackrel{(b)}{\leq} 3 \sqrt{2N \delta_N \log 2} , \label{eqM_beta}
\end{align}
where $(a)$ holds by the triangle inequality, $(b)$ holds by Pinsker's inequality and Lemma~\ref{lem_U1}.

Then, for $N$ large enough ($|\tilde{K}|>4$), we have
\begin{align*}
I(K_i;\widetilde{K}_i) 
&\leq \mathbb{V}(p_{K_i\widetilde{K}_i}, p_{K_i}p_{\widetilde{K}_i}) \log_2 \frac{|\tilde{K}| }{\mathbb{V}(p_{K_i\widetilde{K}_i}, p_{K_i}p_{\widetilde{K}_i})} \\
& \leq \mathbb{V}(p_{K_i\widetilde{K}_i}, p_{K_i}p_{\widetilde{K}_i}) \log_2  |\tilde{K}| \\ \nonumber
& \phantom{--} - \mathbb{V}(p_{K_i\widetilde{K}_i}, p_{K_i}p_{\widetilde{K}_i}) \log_2 \mathbb{V}(p_{K_i\widetilde{K}_i}, p_{K_i}p_{\widetilde{K}_i}) \\
& \leq \delta_{N}^*,
\end{align*}
where $ \delta_{N}^* \triangleq 3 \sqrt{2N \delta_N \log 2} \left( N- \log_2 \left( 3 \sqrt{2N \delta_N \log 2}  \right)\right)$ by~(\ref{eqM_beta}) and because $x \mapsto x \log x$ is decreasing for $x>0$ small enough.

\subsection{Proof of Lemma~\ref{lem_secrec}} \label{App_lemsecrec}

Let $i\in \llbracket 2,k \rrbracket$. By applying the chain rule of mutual information repeatedly, we obtain
\begin{align} \label{eqleak}
 \widetilde{L}_e^{1:i} = \alpha_i + \beta_i + \gamma_i,
\end{align}
where 
\begin{align*}
\alpha_i &\triangleq  I \left(K_{i} \widetilde{K}_i; M_{i} Z^{1:N}_{i}  \right), \\
\beta_i & \triangleq I \left(K_{1:i-1} ; Z^{1:N}_{i} M_{i} | K_{i} \widetilde{K}_i \right),\\ 
\gamma_i & \triangleq I \left(K_{1:i} \widetilde{K}_i; Z^{1:N}_{1:i-1} M_{1:i-1} |Z_i^{1:N} M_i \right).
\end{align*}
Then, note that
\begin{align} 
\gamma_i  \nonumber
& \leq  I\left( K_{1:i} \widetilde{K}_{i-1:i} Z _i^{1:N} M_i; Z^{1:N}_{1:i-1} M_{1:i-1} \right) \\  \nonumber
& = I\left( K_{1:i-1} \widetilde{K}_{i-1} ; Z^{1:N}_{1:i-1} M_{1:i-1} \right)\\ \nonumber
& \phantom{--}  + I\left( K_{i} \widetilde{K}_{i} Z _i^{1:N} M_i; Z^{1:N}_{1:i-1} M_{1:i-1} | K_{1:i-1} \widetilde{K}_{i-1} \right) \\
& =  \widetilde{L}_e^{1:i-1} \label{eqleak1}, 
\end{align}
where the last equality follows from $K_{i} \widetilde{K}_{i} Z _i^{1:N} M_i \rightarrow K_{1:i-1} \widetilde{K}_{i-1} \rightarrow Z^{1:N}_{1:i-1} M_{1:i-1}$. 

We also have,
\begin{align} 
\beta_i \nonumber
& \leq I \left(K_{1:i-1} ; Z^{1:N}_{i} M_{i} \widetilde{K}_{i-1}| K_{i} \widetilde{K}_i \right) \\ \nonumber
& =  I\left(K_{1:i-1} ; \widetilde{K}_{i-1}| K_{i} \widetilde{K}_i \right) \\ \nonumber
& \phantom{--}+ I\left(K_{1:i-1} ; Z^{1:N}_{i} M_{i} | K_{i} \widetilde{K}_{i-1:i} \right) \\ \nonumber
& \stackrel{(a)}{=}  I\left(K_{1:i-1} ; \widetilde{K}_{i-1}| K_{i} \widetilde{K}_i \right)  \\ \nonumber
& \stackrel{(b)}{\leq}  I\left(K_{1:i-1} ; \widetilde{K}_{i-1} \right) \\  \nonumber
& = I\left(K_{i-1} ; \widetilde{K}_{i-1} \right) + I\left(K_{1:i-2} ; \widetilde{K}_{i-1} |K_{i-1} \right) \\ \nonumber
& \leq I\left(K_{i-1} ; \widetilde{K}_{i-1} \right) + I\left(K_{1:i-2} ; \widetilde{K}_{i-1} K_{i-1} \right) \\ \nonumber 
& \leq I\left(K_{i-1} ; \widetilde{K}_{i-1} \right) + I\left(X^{1:N}_{1:i-2} ; X^{1:N}_{i-1} \right) \\  
& \stackrel{(c)}{=} I\left(K_{i-1} ; \widetilde{K}_{i-1} \right) 
\label{eqleak2} 
\end{align}
where $(a)$ holds by $K_{1:i-1} \rightarrow K_{i} \widetilde{K}_{i-1:i} \rightarrow Z^{1:N}_{i} M_{i} $, $(b)$ holds by $K_{1:i-1} \rightarrow \widetilde{K}_{i-1} \rightarrow K_{i} \widetilde{K}_i$, $(c)$ holds by independence between $X^{1:N}_{1:i-2}$ and $X^{1:N}_{i-1}$. 

Finally, we conclude combining (\ref{eqleak}), (\ref{eqleak1}), and (\ref{eqleak2}).
\section{Proofs for Model 2 in Section~\ref{sec_model2}} 
\subsection{Proof of Corollary~\ref{Cor_2}} \label{App_cor2}
We perform the same encoding as in Algorithm~\ref{alg:encoding_2} for Block~1 with $\mathcal{A}_{UYZ} = \emptyset$. Note that $C_{\textup{WSK}}(R_p)$ is obtained when $U$ is uniformly distributed by \cite[Prop. 5.3]{Chou12b} since $X$ is uniform and the tests-channel $p_{Y|X}$ and $p_{Z|X}$ are symmetric. Hence, the rate $R_1$ of randomness to perform successive cancellation encoding can be set equal to zero by~\cite{Korada10}.
We also have
\begin{align*}
|F'_1| 
& = |(\mathcal{H}_{U|Y} \backslash {\mathcal{V}}_{U|X}) \backslash  \mathcal{V}_{U|Z} | \\
& \stackrel{(a)}{\leq}|\mathcal{H}_{U|Z}  \backslash  \mathcal{V}_{U|Z} | \\
& \stackrel{(b)}{=} |\mathcal{H}_{U|Z} | -|  \mathcal{V}_{U|Z}|,
\end{align*}
where $(a)$ holds because $\mathcal{H}_{U|Y} \subset \mathcal{H}_{U|Z}$ since we have assumed $X \to Y \to Z$, $(b)$ holds because ${\mathcal{V}}_{U|Z} \subset \mathcal{H}_{U|Z}$. 

We conclude by Lemma~\ref{lemnewcard} and \cite{Arikan10} that $|F'_1| = o(N)$.

\subsection{Proof of Lemma \ref{lemDivprob}} \label{App_lemdiv}
Using the notation of \cite{Cover91} for conditional relative entropy, we have for $i\in \llbracket 1, k \rrbracket$, $p_{U_i^{1:N}X_i^{1:N}} = p_{U^{1:N}X^{1:N}} $ and
\begin{align*}
&\mathbb{D}(p_{X^{1:N}U^{1:N}} || \widetilde{p}_{X_i^{1:N}U_i^{1:N}}) \\
& \stackrel{(a)}{=} \mathbb{D}(p_{X^{1:N}V^{1:N}} || \widetilde{p}_{X_i^{1:N}V_i^{1:N}}) \\
& \stackrel{(b)}{=} \mathbb{D}(p_{V^{1:N}|X^{1:N}} || \widetilde{p}_{V_i^{1:N}|X^{1:N}}) \\
& \stackrel{(c)}{=} \sum_{j=1}^{N} \mathbb{D}(p_{V^{j}|V^{1:j-1}X^{1:N}} || \widetilde{p}_{V^j_i|V_i^{1:j-1}X^{1:N}})\\
& \stackrel{(d)}{=} \sum_{j\in \mathcal{V}_{U|X}} \sum_{j\in \mathcal{H}_{U}^c}    \mathbb{D}(p_{V^{j}|V^{1:j-1}X^{1:N}} || \widetilde{p}_{V^j_i|V^{1:j-1}_iX^{1:N}})\\
& \stackrel{(e)}{=} \sum_{j\in \mathcal{V}_{U|X}} ( 1 -H(V^{j}|V^{1:j-1}X^{1:N}) ) \\ \nonumber
& \phantom{--}+  \sum_{j\in \mathcal{H}_{U}^c}    (H(V^{j}|V^{1:j-1}) - H(V^j|V^{1:j-1}X^{1:N}) )\\
& \leq |\mathcal{V}_{U|X}| \delta_N +  \sum_{j\in \mathcal{H}_{U}^c}    H(V^{j}|V^{1:j-1}) \\
& \leq |\mathcal{V}_{U|X}| \delta_N +  | \mathcal{H}_{U}^c|   \delta_N \\
&\leq N \delta_N,
\end{align*}
where $(a)$ holds by invertibility of $G_n$, $(b)$ and $(c)$ hold by the chain rule for divergence, $(d)$ and $(e)$ hold by~(\ref{eq_VQ_def}) and by uniformity of the components of $\widetilde{V}_i^{1:N}$ in $\mathcal{V}_{U|X}$.

\subsection{Proof of Lemma \ref{lem_U2_Div}} \label{App_lem_U2_Div}

We have by \cite[Lemma 2.7]{bookCsizar}
\begin{align*}
&|K_i| + |\widetilde{K}_i|  - H(K_i \widetilde{K}_i ) \\
& \leq \mathbb{V} (p_{K_i \widetilde{K}_i}, q_{\mathcal{U}_{K,\widetilde{K}}}) \log_2 \frac{ |K_i|+|\widetilde{K}_i|} {\mathbb{V} (p_{K_i\widetilde{K}_i}, q_{\mathcal{U}_{K,\widetilde{K}}})} \\
& \leq N \mathbb{V} (p_{K_i \widetilde{K}_i}, q_{\mathcal{U}_{K,\widetilde{K}}}) \\ \nonumber
& \phantom{--} -  \mathbb{V} (p_{K_i \widetilde{K}_i}, q_{\mathcal{U}_{K,\widetilde{K}}}) \log_2  {\mathbb{V} (p_{K_i\widetilde{K}_i}, q_{\mathcal{U}_{K,\widetilde{K}}})} \\
& \leq 2 \sqrt{2 \log2} \sqrt{N \delta_N} ( N - \log_2  (2 \sqrt{2 \log2} \sqrt{N \delta_N})),
\end{align*}
where the last inequality holds for $N$ large enough by Lemma~\ref{lem_U2} and because $x \mapsto x \log x$ is decreasing for $x>0$ small enough.

\subsection{Proof of Lemma \ref{lem_KKtilde2}} \label{App_lem_KKtilde2}
We only prove the first inequality, the other is obtained similarly.
Let $i \in \llbracket 1,k \rrbracket$. We have, 
\begin{align}
&\mathbb{V} \left( p_{K_i\widetilde{K}_iR_1}, p_{K_i}p_{\widetilde{K}_iR_1} \right) \nonumber\\ \nonumber
& \stackrel{(a)}{\leq} \mathbb{V} \left( p_{K_i\widetilde{K}_iR_1}, q_{\mathcal{U}_{K, \widetilde{K}, R}} \right) + \mathbb{V}\left( q_{\mathcal{U}_{K, \widetilde{K}, R}}, q_{\mathcal{U}_{ {K}}}p_{\widetilde{K}_iR_1} \right) \\ \nonumber
& \phantom{--}+ \mathbb{V} \left( q_{\mathcal{U}_{K}}p_{\widetilde{K}_iR_1}, p_{K_i}p_{\widetilde{K}_iR_1} \right) \\ \nonumber
& = \mathbb{V} \left( p_{K_i\widetilde{K}_iR_1}, q_{\mathcal{U}_{K, \widetilde{K}, R}} \right) + \mathbb{V}\left( q_{\mathcal{U}_{\widetilde{K}, R}}, p_{\widetilde{K}_iR_1} \right)\\ \nonumber
& \phantom{--} + \mathbb{V} \left( q_{\mathcal{U}_{K}}, p_{K_i} \right) \displaybreak[0]\\ \nonumber
& \leq 3 \mathbb{V} \left( p_{K_i\widetilde{K}_iR_1}, q_{\mathcal{U}_{K, \widetilde{K}, R}} \right) \\ 
& \stackrel{(b)}{\leq} 6\sqrt{2 \log2} \sqrt{N \delta_N}, \label{eqM_beta2}
\end{align}
where $(a)$ holds by the triangle inequality, $(b)$ holds by Pinsker's inequality and Lemma~\ref{lem_U2}.

Then, for $N$ large enough ($|{K}|>4$), we have by \cite{Csiszar96}
\begin{align*}
&I(K_i;\widetilde{K}_iR_1) \\
&\leq \mathbb{V}(p_{K_i\widetilde{K}_iR_1}, p_{K_i}p_{\widetilde{K}_iR_1}) \log_2 \frac{|{K}| }{\mathbb{V}(p_{K_i\widetilde{K}_iR_1}, p_{K_i}p_{\widetilde{K}_iR_1})} \\
& \leq N \mathbb{V}(p_{K_i\widetilde{K}_iR_1}, p_{K_i}p_{\widetilde{K}_iR_1})  \\ \nonumber
& \phantom{--}- \mathbb{V}(p_{K_i\widetilde{K}_iR_1}, p_{K_i}p_{\widetilde{K}_iR_1}) \log_2 \mathbb{V}(p_{K_i\widetilde{K}_iR_1}, p_{K_i}p_{\widetilde{K}_iR_1}) \\
& \leq \delta_{N}^{(2)},
\end{align*}
where $ \delta_{N}^{(2)} \triangleq 6\sqrt{2 \log2} \sqrt{N \delta_N} (N - \log_2( 6\sqrt{2 \log2} \sqrt{N \delta_N})) $ by (\ref{eqM_beta2}) and because $x \mapsto x \log x$ is decreasing for $x>0$ small enough.

\subsection{Proof of Lemma \ref{lem_secre_cra}} \label{App_lem_secre_cra}

We have for $i \in \llbracket 1 , k \rrbracket$, $p_{V_i^{1:N}X_i^{1:N}Z_i^{1:N}} = p_{V^{1:N}X^{1:N}Z^{1:N}} $ and
\begin{align}
& \mathbb{V}(\widetilde{p}_{V_i^{1:N}[\mathcal{V}_{U|Z}] Z_i^{1:N}},{p}_{V^{1:N}[\mathcal{V}_{U|Z}] Z^{1:N}}) \nonumber \\  \nonumber
& \leq \mathbb{V}(\widetilde{p}_{V_i^{1:N}X_i^{1:N} Z_i^{1:N}},{p}_{V^{1:N} X^{1:N} Z^{1:N}}) \\ \nonumber
& = \mathbb{V}(\widetilde{p}_{ Z_i^{1:N} | V_i^{1:N} X_i^{1:N}} \widetilde{p}_{  V_i^{1:N} X_i^{1:N}},{p}_{Z^{1:N} | V^{1:N} X^{1:N} } {p}_{  V^{1:N} X^{1:N}}) \\ \nonumber
& = \mathbb{V}(\widetilde{p}_{ Z_i^{1:N} | X_i^{1:N}} \widetilde{p}_{ V_i^{1:N} X_i^{1:N}},{p}_{Z^{1:N} |  X^{1:N} } {p}_{  V^{1:N} X^{1:N}}) \\ \nonumber
& = \mathbb{V}( \widetilde{p}_{ V_i^{1:N} X_i^{1:N}}, {p}_{  V^{1:N} X^{1:N}}) \\
& \leq \sqrt{2 \log2} \sqrt{N \delta_N}, \label{eq_sec_int1}
\end{align}
where the last inequality follows by Lemma \ref{lemDivprob}, and 
\begin{align}
&  \mathbb{V}({p}_{V^{1:N}[\mathcal{V}_{U|Z}] Z^{1:N}},\widetilde{p}_{V_i^{1:N}[\mathcal{V}_{U|Z}]} p_{ Z^{1:N}})   \nonumber\\  \nonumber
& \leq \mathbb{V}({p}_{V^{1:N}[\mathcal{V}_{U|Z}] Z^{1:N}},{p}_{V^{1:N}[\mathcal{V}_{U|Z}]} p_{ Z^{1:N}}) \\ \nonumber
& \phantom{--}+ \mathbb{V}({p}_{V^{1:N}[\mathcal{V}_{U|Z}]} p_{ Z^{1:N}},\widetilde{p}_{V_i^{1:N}[\mathcal{V}_{U|Z}]} p_{ Z^{1:N}}) \\  \nonumber
& \stackrel{(a)}{\leq} \mathbb{V}({p}_{V^{1:N}[\mathcal{V}_{U|Z}] Z^{1:N}},{p}_{V^{1:N}[\mathcal{V}_{U|Z}]} p_{ Z^{1:N}})  + \sqrt{2 \log2} \sqrt{N \delta_N} \\  \nonumber
& \stackrel{(b)}{\leq} \sqrt{2 \log2} \sqrt{ \mathbb{D}({p}_{V^{1:N}[\mathcal{V}_{U|Z}] Z^{1:N}}||{p}_{V^{1:N}[\mathcal{V}_{U|Z}]} p_{ Z^{1:N}})} \\ \nonumber
& \phantom{--} + \sqrt{2 \log2} \sqrt{N \delta_N} \\  \nonumber
& = \sqrt{2 \log2} \sqrt{ I( V^{1:N}[\mathcal{V}_{U|Z}] ;Z^{1:N} )}  + \sqrt{2 \log2} \sqrt{N \delta_N} \\  \nonumber
& \leq \sqrt{2 \log2} \sqrt{ |\mathcal{V}_{U|Z}|- H( V^{1:N}[\mathcal{V}_{U|Z}] | Z^{1:N} )} \\ \nonumber
& \phantom{--} + \sqrt{2 \log2} \sqrt{N \delta_N} \\
& \stackrel{(c)}{\leq} 2 \sqrt{2 \log2} \sqrt{ N \delta_N}, \label{eq_sec_int2}
\end{align}
where $(a)$ holds by (\ref{eq_sec_int1}), $(b)$ holds by Pinsker's inequality, $(c)$ holds because similar to the proof of Lemma~\ref{lem_sec_block} $|\mathcal{V}_{U|Z}|- H( V^{1:N}[\mathcal{V}_{U|Z}] | Z^{1:N} ) \leq N \delta_N$.

Hence, by (\ref{eq_sec_int1}) and (\ref{eq_sec_int2})
\begin{align}
&\mathbb{V}(\widetilde{p}_{V_i^{1:N}[\mathcal{V}_{U|Z}] Z_i^{1:N}},\widetilde{p}_{V_i^{1:N}[\mathcal{V}_{U|Z}]} p_{ Z^{1:N}}) \nonumber \\ \nonumber
& \leq \mathbb{V}(\widetilde{p}_{V_i^{1:N}[\mathcal{V}_{U|Z}] Z_i^{1:N}},{p}_{V^{1:N}[\mathcal{V}_{U|Z}] Z^{1:N}})
\\ \nonumber
& \phantom{--}+  \mathbb{V}({p}_{V^{1:N}[\mathcal{V}_{U|Z}] Z^{1:N}},\widetilde{p}_{V_i^{1:N}[\mathcal{V}_{U|Z}]} p_{ Z^{1:N}}) \\
& \leq 3 \sqrt{2 \log2} \sqrt{ N \delta_N}, \label{eq_sec_int3}
\end{align}
and for $N$ large enough by \cite{Csiszar96}
\begin{align*}
& I(\widetilde{V}_i^{1:N}[\mathcal{V}_{U|Z}] ;Z_i^{1:N}) \\
& \leq   \mathbb{V}(\widetilde{p}_{V_i^{1:N}[\mathcal{V}_{U|Z}] Z_i^{1:N}},\widetilde{p}_{V_i^{1:N}[\mathcal{V}_{U|Z}]} p_{ Z_i^{1:N}}) \\
& \phantom{--} \times \log_2 \frac{|\mathcal{V}_{U|Z}|}{\mathbb{V}(\widetilde{p}_{V_i^{1:N}[\mathcal{V}_{U|Z}] Z_i^{1:N}},\widetilde{p}_{V_i^{1:N}[\mathcal{V}_{U|Z}]} p_{ Z_i^{1:N}})} \\
& \leq  3 \sqrt{2 \log2} \sqrt{ N \delta_N} ( N - \log_2 (3 \sqrt{2 \log2} \sqrt{ N \delta_N} ) ).
\end{align*}

\subsection{Proof of Lemma \ref{lem_secrec2}} \label{App_lem_secrec2}

Let $i\in \llbracket 2,k \rrbracket$. By applying the chain rule of mutual information repeatedly, we obtain
\begin{align} \label{eqleak_b}
 \widetilde{L}_e^{1:i} = \alpha_i + \beta_i + \gamma_i,
\end{align}
where 
\begin{align*}
\alpha_i &\triangleq  I \left(K_{i} \widetilde{K}_i; R_1 M_{i}  Z^{1:N}_{i}  \right), \\
\beta_i & \triangleq I \left(K_{1:i-1} ; R_1 M_{i}  Z^{1:N}_{i}  | K_{i} \widetilde{K}_i \right),\\ 
\gamma_i & \triangleq I \left(K_{1:i} \widetilde{K}_i;  M_{1:i-1}  Z^{1:N}_{1:i-1} | R_1  M_i Z_i^{1:N} \right).
\end{align*}
Then, note that
\begin{align} 
\gamma_i  \nonumber
& \stackrel{(a)}{\leq}  I\left( K_{1:i} \widetilde{K}_{i-1:i} M_i Z _i^{1:N} ;   M_{1:i-1} Z^{1:N}_{1:i-1}|R_1 \right) \\  \nonumber
& = I\left( K_{1:i-1} \widetilde{K}_{i-1} ;  M_{1:i-1} Z^{1:N}_{1:i-1} |R_1\right)\\ \nonumber
& \phantom{--}  + I\left( K_{i} \widetilde{K}_{i} Z _i^{1:N} M_i;   M_{1:i-1} Z^{1:N}_{1:i-1} | R_1 K_{1:i-1} \widetilde{K}_{i-1} \right) \\ \nonumber
& \stackrel{(b)}{=} I\left( K_{1:i-1} \widetilde{K}_{i-1} ;  M_{1:i-1} Z^{1:N}_{1:i-1} |R_1\right)   \\ \nonumber
& \stackrel{(c)}{\leq} I\left( K_{1:i-1} \widetilde{K}_{i-1} ;R_1 M_{1:i-1} Z^{1:N}_{1:i-1}  \right) \\ 
& = \widetilde{L}_e^{1:i-1} \label{eqleak_bb}, 
\end{align}
where $(a)$ and $(c)$ hold by the chain rule and positivity of mutual information, $(b)$ holds because $K_{i} \widetilde{K}_{i} Z _i^{1:N} M_i \to  R_1 K_{1:i-1} \widetilde{K}_{i-1} \to M_{1:i-1} Z^{1:N}_{1:i-1}$. 

We also have,
\begin{align} 
\beta_k \nonumber
& \stackrel{(d)}{\leq} I \left(K_{1:i-1} ;  R_1 M_{i}  Z^{1:N}_{i}  \widetilde{K}_{i-1} | K_{i} \widetilde{K}_i \right) \\ \nonumber
& =  I\left(K_{1:i-1} ; \widetilde{K}_{i-1} R_1| K_{i} \widetilde{K}_i \right)\\ \nonumber
& \phantom{--} + I\left(K_{1:i-1} ; M_{i} Z^{1:N}_{i} | K_{i} \widetilde{K}_{i-1:i} R_1 \right) \\ \nonumber
& \stackrel{(e)}{=}  I\left(K_{1:i-1} ; \widetilde{K}_{i-1} R_1| K_{i} \widetilde{K}_i \right)  \\ \nonumber
& \stackrel{(f)}{\leq}  I\left(K_{1:i-1} ; \widetilde{K}_{i-1} R_1 \right) \\  \nonumber
& = I\left(K_{1:i-1} ;  R_1 \right) +  I\left(K_{1:i-1} ;  \widetilde{K}_{i-1} |R_1 \right) \\ \nonumber
& = I\left(K_{1:i-1} ;  R_1 \right) +  I\left(K_{i-1} ; \widetilde{K}_{i-1} | R_1 \right) \\ \nonumber
& \phantom{--}+ I\left(K_{1:i-2} ;  \widetilde{K}_{i-1} |K_{i-1} R_1 \right) \\ \nonumber
& \stackrel{(g)}{=} I\left(K_{1:i-1} ;  R_1 \right) +  I\left(K_{i-1} ; \widetilde{K}_{i-1} | R_1 \right) \\ \nonumber
& \leq  I\left(K_{1:i-1} ;  R_1 \right) +  I\left(K_{i-1} ; \widetilde{K}_{i-1} R_1 \right) \\ \nonumber
& = I\left(K_{1:i-2} ;  R_1 | K_{i-1}\right) + I\left(K_{i-1} ;  R_1 \right) \\ \nonumber
& \phantom{--}+  I\left(K_{i-1} ; \widetilde{K}_{i-1}  R_1 \right) \displaybreak[0]\\ \nonumber
&  \stackrel{(h)}{\leq} I\left(K_{1:i-2} ;  R_1 \right) + I\left(K_{i-1} ;  R_1 \right) +  I\left(K_{i-1} ; \widetilde{K}_{i-1}  R_1 \right) \\ 
&  \stackrel{(i)}{\leq} \sum_{j=1}^{i-1} I\left(K_{j} ;  R_1 \right) +  I\left(K_{i-1} ; \widetilde{K}_{i-1}  R_1 \right) \label{eqleak2_b} 
\end{align}
where $(d)$ holds by the chain rule and positivity of mutual information, $(e)$ holds because $K_{1:i-1} \to  K_{i} \widetilde{K}_{i-1:i} R_1 \to M_{i} Z^{1:N}_{i}$, $(f)$ holds because $K_{1:i-1} \to \widetilde{K}_{i-1} R_1 \to K_{i} \widetilde{K}_i $, $(g)$ holds because $K_{1:i-2} \to  K_{i-1} R_1 \to  \widetilde{K}_{i-1} $, $(h)$ holds because $K_{1:i-2} \to  R_1 \to  K_{i-1}$, $(i)$ holds by induction.

Finally, we conclude combining (\ref{eqleak_b}), (\ref{eqleak_bb}), and (\ref{eqleak2_b}).

\section{Proof of Theorem \ref{th3ter}} \label{App_Th6}

\noindent{} \textit{1) Existence of $\mathcal{F}_{X_{\mathcal{M}}}$:} 
The set $\mathcal{F}_{X_{\mathcal{M}}}$ exists because we have assumed $I(X_2;X_1) \leq I(X_2;X_3)$, i.e., $H(X_2|X_1) \geq H(X_2|X_3)$. Indeed, 
\begin{align*}
&|{\mathcal{F}}_{X_2|X_1} \backslash {\mathcal{F}}_{X_2|X_3}| -|\bar{\mathcal{K}}_{X_{\mathcal{M}}}| \\
& = |{\mathcal{F}}_{X_2|X_1} \backslash {\mathcal{F}}_{X_2|X_3}| -|{\mathcal{F}}_{X_2|X_3} \backslash {\mathcal{F}}_{X_2|X_1}|\\
& =  |{\mathcal{F}}_{X_2|X_1} | - | {\mathcal{F}}_{X_2|X_3}|,
\end{align*}
and $\lim_{N \to \infty} (|{\mathcal{F}}_{X_2|X_1} | - | {\mathcal{F}}_{X_2|X_3}|)/N = H(X_2|X_1) - H(X_2|X_3)$ by Lemma \ref{lemnewcard} and~\cite{Arikan10}.\\
\textit{2) Key Rate:} The key rate is 

\begin{align*}
& \frac{|\mathcal{K}_{X_{\mathcal{M}}} | + (k-1) |\mathcal{K}_{X_{\mathcal{M}}}  \cup \mathcal{F}_{X_{\mathcal{M}}}| }{kN} \\
& \stackrel{(a)}{=} \frac{|\mathcal{K}_{X_{\mathcal{M}}} | + (k-1) (|\mathcal{K}_{X_{\mathcal{M}}} |  +| \mathcal{F}_{X_{\mathcal{M}}}|) }{kN}\\
& = \frac{ |\mathcal{K}_{X_{\mathcal{M}}} | +| \mathcal{F}_{X_{\mathcal{M}}}| }{N} - \frac{   | \mathcal{F}_{X_{\mathcal{M}}}| }{kN} \displaybreak[0]\\
& = \frac{ |\mathcal{K}_{X_{\mathcal{M}}} |  +| \bar{\mathcal{K}}_{X_{\mathcal{M}}}| }{N} - \frac{ | \bar{\mathcal{K}}_{X_{\mathcal{M}}}| }{kN} \displaybreak[0] \\
& = \frac{ |\mathcal{V}_{X_2} \backslash \mathcal{H}_{X_2|X_1} |  }{N} - \frac{ | \bar{\mathcal{K}}_{X_{\mathcal{M}}}| }{kN} \displaybreak[0]\\
& \geq \frac{ |\mathcal{V}_{X_2} \backslash \mathcal{H}_{X_2|X_1} | }{N} - \frac{  |\mathcal{V}_{X_2} \backslash \mathcal{H}_{X_2|X_1} | }{kN}\\
& \xrightarrow{ N \to \infty} I(X_1;X_2) \left(1 - \frac{1}{k}\right)\\
& \xrightarrow{ k \to \infty} I(X_1;X_2) ,
\end{align*}
where $(a)$ holds because $\mathcal{F}_{X_{\mathcal{M}}} \cap \mathcal{K}_{X_{\mathcal{M}}} = \emptyset$, and where we have used Lemma \ref{lemnewcard} and~\cite{Arikan10} for the first limit.

\textit{3) Reliability: }
We do not detail the reliability analysis for Terminal $1$, since it is similar to the analysis for Terminal $3$.

Let $i \in \llbracket 1,k-1\rrbracket$. Note that Terminal $3$ forms an accurate estimate of $U_{i}^{1:N} [ \mathcal{H}_{X_2|X_3} ]$ only when $U_{i+1}^{1:N}$ is correctly reconstructed (see Remark \ref{rem:justif_decoding_3}). We note $\widehat{U}_{i}^{1:N} [ \mathcal{H}_{X_2|X_3} ]$ the estimate of $U_{i}^{1:N} [ \mathcal{H}_{X_2|X_3} ]$ formed by Terminal $3$ and define $\mathcal{E}_{i} \triangleq \{ \widehat{U}_{i}^{1:N} [ \mathcal{H}_{X_2|X_3} ] \neq  U_{i}^{1:N} [ \mathcal{H}_{X_2|X_3} ] \}$.

Hence, 
\begin{align*} 
&\mathbb{P} [K_i \neq \widehat{K}_i]\\
& \leq  \mathbb{P} [U_i^{1:N} \neq \widehat{U}_i^{1:N}]\\
& = \mathbb{P} [U_i^{1:N} \neq \widehat{U}_i^{1:N}| \mathcal{E}_{i}^c] \mathbb{P} [\mathcal{E}_{i}^c] + \mathbb{P} [U_i^{1:N} \neq \widehat{U}_i^{1:N}| \mathcal{E}_{i}] \mathbb{P} [\mathcal{E}_{i}]\\
& \leq \mathbb{P} [U_i^{1:N} \neq \widehat{U}_i^{1:N} | \mathcal{E}_{i}^c]  +  \mathbb{P} [\mathcal{E}_{i}]\\
& \leq \mathbb{P} [U_i^{1:N} \neq \widehat{U}_i^{1:N} | \mathcal{E}_{i}^c]  +  \mathbb{P} [U_{i+1}^{1:N} \neq \widehat{U}_{i+1}^{1:N}]\\
& \stackrel{(a)}{\leq} N \delta_N  +  \mathbb{P} [U_{i+1}^{1:N} \neq \widehat{U}_{i+1}^{1:N}]\\
& \stackrel{(b)}{\leq} (k -i ) N \delta_N  +  \mathbb{P} [U_{k}^{1:N} \neq \widehat{U}_{k}^{1:N}]\\
& \stackrel{(c)}{\leq} (k-i+1) N \delta_N  ,
\end{align*}
where $(a)$ holds because by \cite{Arikan10}, Terminal $3$ can reconstruct $U^{1:N}_i$ from $U_i^{1:N}[\mathcal{H}_{X_2|X_3}]$ and $(X_3)_i^{1:N}$ with error probability  less than $N\delta_N$, $(b)$ holds by recurrence, $(c)$ holds similarly as previous equations.

Then, by the union bound,
\begin{align*} 
\mathbf{P}_e(\mathcal{S}_N) \nonumber
&= \mathbb{P} [K_{1:k} \neq \widehat{K}_{1:k}] \nonumber \\
& \leq \sum_{i=1}^k \mathbb{P} [K_i \neq \widehat{K}_i] \nonumber\\\nonumber
& \leq \sum_{i=1}^{k} (k-i+1) N \delta_N \\
& = \frac{k(k+1)}{2} N\delta_N. 
\end{align*}

\textit{4) Key Uniformity: }
Similarly to Lemma \ref{lem_U1} we have the key uniformity for each block.%
\begin{lem} \label{lem_Umod3dem}
Uniformity of $[K_i, \bar{K}_i]$ holds for each block, where $i \in \llbracket 1,k-1 \rrbracket$. Specifically,
\begin{align*}
|K_i| + | \bar{K}_i|  - H(K_i  \bar{K}_i ) \leq N \delta_N.
\end{align*}
Hence, we also have
\begin{align*}
|\bar{K}_i| - H( \bar{K}_i) \leq N \delta_N,\\
|K_i| - H( K_i) \leq N \delta_N.
\end{align*}

\end{lem}


%
%
%
The global key $K_{1:k}$ is asymptotically uniform as, similarly to the proof of Theorem \ref{Th1} in Section~\ref{Sec_unif_model1},
%
we have
\begin{align*} 
\textbf{\textup{U}}(\mathcal{S}_N)= |K_{1:k}|  - H(K_{1:k}) \leq k N \delta_N. 
\end{align*}

\textit{5) Strong Secrecy: }
Similar to Lemma~\ref{lem_sec_block}, we obtain the following result showing that secrecy holds for each block.

\begin{lem} \label{lem_sec_blockmod3}
Let $i \in \llbracket 1, k \rrbracket$. For each Block $i$, secrecy of $[K_i, \bar{K}_i]$ holds. Specifically, we have 
\begin{align*}
I\left(K_{i}\bar{K}_{i}\widetilde{K}_{i}; M_{i} \right) \leq 2N \delta_N.
\end{align*}
\end{lem}

\begin{proof}
We have for $i\in \llbracket 2 ,k-1 \rrbracket$
\begin{align}
& I(K_{i}\bar{K}_{i}\widetilde{K}_{i}; F_{i}^{(2)} ) \nonumber \\ \nonumber
& \stackrel{(a)}{=} I(K_{i}\bar{K}_{i}; F_{i}^{(2)} ) \nonumber \\ \nonumber
& = H( K_i \bar{K}_{i}) - H(K_{i}\bar{K}_{i} | F_{i}^{(2)} ) \\ \nonumber
& \leq |K_i| + |\bar{K}_{i}|  - H(K_{i}\bar{K}_{i} F_{i}^{(2)}  ) + H(F_{i}^{(2)}) \\ \nonumber
& \leq |K_i| + |\bar{K}_{i}|  + |F_{i}^{(2)}| - H(K_{i}\bar{K}_{i} F_{i}^{(2)}  )  \\ \nonumber
%
%
& \stackrel{(b)}{=} |(\mathcal{V}_{X_2} \backslash \mathcal{H}_{X_2|X_1}) \cup (\mathcal{V}_{X_2} \cap \mathcal{H}_{X_2|X_1} \cap \mathcal{F}_{X_{\mathcal{M}}})| \nonumber
 \\  \nonumber
& \phantom{--}- H(U_i^{1:N}[(\mathcal{V}_{X_2} \backslash \mathcal{H}_{X_2|X_1}) \cup (\mathcal{V}_{X_2} \cap \mathcal{H}_{X_2|X_1} \cap \mathcal{F}_{X_{\mathcal{M}}})] ) \\\nonumber
& \stackrel{(c)}{\leq} |\mathcal{V}_{X_2}| - \sum_{j \in \mathcal{V}_{X_2}} H(U_i^{j}|U_i^{1:j-1} ) \\ \nonumber \displaybreak[0]
& \stackrel{(d)}{\leq} |\mathcal{V}_{X_2}| - \sum_{j \in \mathcal{V}_{X_2}} (1-\delta_N) \\  \nonumber
& \leq |\mathcal{V}_{X_2}| \delta_N \displaybreak[0] \\
& \leq N \delta_N, \label{eqsec1dem} 
\end{align}
where $(a)$ holds by independence between $\widetilde{K}_i$ and all the other random variables, $(b)$ holds by definition of $K_i$, $\bar{K}_i$, $\widetilde{K}_i$, and $F_i^{(2)}$, $(c)$ holds because $(\mathcal{V}_{X_2} \backslash \mathcal{H}_{X_2|X_1}) \cup (\mathcal{V}_{X_2} \cap \mathcal{H}_{X_2|X_1} \cap \mathcal{F}_{X_{\mathcal{M}}}) \subset \mathcal{V}_{X_2}$ because conditioning reduces entropy, $(d)$ holds by definition of $\mathcal{V}_{X_2}$.

Then, we obtain for $i\in \llbracket 2 ,k-1 \rrbracket$,
\begin{align*}
& I(K_{i}\bar{K}_{i} \widetilde{K}_{i}; M_{i} ) - N \delta_N\\
& \stackrel{(a)}{=} I(K_{i} \bar{K}_{i} \widetilde{K}_{i}; F_{i}^{(2)}  (F_{i}^{(1)} \oplus \bar{K}_{i-1}) (F_{i}' \oplus \widetilde{K}_{i-1}) )- N \delta_N \\
& = I(K_{i} \bar{K}_{i} \widetilde{K}_{i};   (F_{i}^{(1)} \oplus \bar{K}_{i-1}) (F_{i}' \oplus \widetilde{K}_{i-1}) | F_{i}^{(2)} ) \\ 
& \phantom{--}+ I(K_{i} \bar{K}_{i} \widetilde{K}_{i}; F_{i}^{(2)}  ) - N \delta_N\displaybreak[0]\\
& \stackrel{(b)}{\leq}  I(K_{i} \bar{K}_{i} \widetilde{K}_{i} F_{i}^{(1)}  F_{i}^{(2)} F_{i}'; (F_{i}^{(1)} \oplus \bar{K}_{i-1}) (F_{i}' \oplus \widetilde{K}_{i-1}) ) \\
& =  H((F_{i}^{(1)} \oplus \bar{K}_{i-1}) (F_{i}' \oplus \widetilde{K}_{i-1})) \\
& \phantom{--}- H((F_{i}^{(1)} \oplus \bar{K}_{i-1}) (F_{i}' \oplus \widetilde{K}_{i-1})|K_{i} \bar{K}_{i} \widetilde{K}_{i} F_{i}^{(1)}  F_{i}^{(2)} F_{i}') \\
& =  H((F_{i}^{(1)} \oplus \bar{K}_{i-1}) (F_{i}' \oplus \widetilde{K}_{i-1})) \\
&\phantom{--}- H( \bar{K}_{i-1} \widetilde{K}_{i-1}|K_{i} \bar{K}_{i} \widetilde{K}_{i} F_{i}^{(1)}  F_{i}^{(2)} F_{i}') \\
& \stackrel{(c)}{=}  H((F_{i}^{(1)} \oplus \bar{K}_{i-1}) (F_{i}' \oplus \widetilde{K}_{i-1})) - H(  \bar{K}_{i-1} \widetilde{K}_{i-1} ) \displaybreak[0]\\
& \leq  |\bar{K}_{i-1}| + |\widetilde{K}_{i-1}| - H( \bar{K}_{i-1}) -H( \widetilde{K}_{i-1}) \displaybreak[0]\\
&  \stackrel{(d)}{\leq} N \delta_N,
\end{align*}
where $(a)$ holds by definition of $M_i$, $(b)$ holds by~(\ref{eqsec1dem}) and the chain rule for mutual information, $(c)$ holds by independence between $U_i^{1:N}$ and $U_{i-1}^{1:N}$, $(d)$ holds by Lemma~\ref{lem_Umod3dem}. The cases $i \in \{ 1,k \}$ are treated similarly.
\end{proof}

Similar to Lemmas \ref{lem_KKtilde} and \ref{lem_secrec}  we also have the following lemmas.

\begin{lem} \label{lem_KKtildemod3}
For $i \in \llbracket 1,k \rrbracket$, we have for $N$ large enough
\begin{align*}
I(K_i;\bar{K}_i ) 
& \leq \delta_{N}^*,
\end{align*}
where 
\begin{equation*} 
 \delta_{N}^* \triangleq - 3 N \sqrt{2N \delta_N \log 2}  \log_2 \left( 3 \sqrt{2N\delta_N \log 2}  \right).
 \end{equation*} 
\end{lem}

\begin{lem} \label{lem_secrecmod3}
For $i\in \llbracket 2,k \rrbracket$, define
\begin{align*}
\widetilde{L}_e^{1:i}  \triangleq  I (K_{1:i} \bar{K}_i  ; M_{1:i}).
\end{align*}
We have 
\begin{align*}
\widetilde{L}_e^{1:i} - \widetilde{L}_e^{1:i-1} \leq I \left(K_{i} \bar{K}_i ; M_{i}   \right) + I\left(K_{i-1} ; \bar{K}_{i-1}  \right).
\end{align*}
\end{lem}

 Similar to the proof of Theorem \ref{Th1}, using Lemmas \ref{lem_sec_blockmod3}, \ref{lem_secrecmod3}, \ref{lem_KKtildemod3},  we obtain
\begin{align*}
\textbf{L}(\mathcal{S}_N)
& \leq 2kN\delta_N + (k-1) \delta_{N}^*. 
\end{align*}

\textit{6) Seed Rate: }
The seed rate is 
\begin{align*}
 \frac{ \sum_{i=1}^k|\widetilde{K}_i| }{kN} 
&  = \frac{k|\bar{\mathcal{F}}_{X_2|X_1} \cup \bar{\mathcal{F}}_{X_2|X_3} | }{kN} \\
&  \leq \frac{|\bar{\mathcal{F}}_{X_2|X_1} | +| \bar{\mathcal{F}}_{X_2|X_3} | }{N} \\
& \xrightarrow{ N \to \infty} 0,
\end{align*}
where we have used Lemma \ref{lemnewcard} and~\cite{Arikan10}. 


\bibliographystyle{IEEEtran}
\bibliography{bib}

\end{document}